\documentclass[12pt,a4paper]{article}
\PassOptionsToPackage{unicode}{hyperref}
\PassOptionsToPackage{hyphens}{url}
\PassOptionsToPackage{dvipsnames,svgnames,x11names}{xcolor}

\usepackage{amsmath}
\usepackage{graphicx, psfrag, epsf, orcidlink, setspace, enumitem}

\setlist[itemize]{
    itemsep=0pt,
    parsep=0pt
}

\usepackage[font=small]{caption}
\usepackage[round, authoryear]{natbib}
\usepackage{url} 
\usepackage[top=1.75in, bottom=1.75in, left=1.05in, right=1.05in]{geometry}
\usepackage{float,bm,booktabs,makecell,amsthm,amssymb,amsfonts,dsfont,mathtools,authblk,subcaption}

\newcommand{\const}{\mathsf{A}}
\newcommand{\esttype}{\mathcal E}

\DeclareMathOperator*{\argmin}{arg\,min}

\definecolor{darkblue}{rgb}{0,0,.6}
\hypersetup{citecolor=darkblue,linkcolor=darkblue,urlcolor=darkblue}

\usepackage{algorithm}
\usepackage{algorithmic}
\allowdisplaybreaks

\usepackage{mathrsfs}
\usepackage{titlesec}

\titlespacing*{\section}{0pt}{8pt plus 4pt minus 2pt}{6pt plus 2pt}
\titlespacing*{\subsection}{0pt}{6pt plus 3pt minus 2pt}{4pt plus 2pt}
\titlespacing*{\subsubsection}{0pt}{6pt plus 2pt minus 2pt}{4pt plus 1pt}

\usepackage[font=small]{caption}
\graphicspath{{plots/}}

\usepackage{amssymb}
\usepackage{hyperref}
\usepackage{multirow}
\usepackage{booktabs}
\usepackage{mathrsfs, xcolor} 
\definecolor{darkblue}{rgb}{0,0,.6}
\hypersetup{colorlinks = true, linkcolor=darkblue, urlcolor=darkblue, citecolor=darkblue}
\usepackage[patch=none]{microtype}
\usepackage{xr-hyper}
\usepackage{iftex}
\ifPDFTeX
  \usepackage[T1]{fontenc}
  \usepackage[utf8]{inputenc}
  \usepackage{textcomp} 
\else 
  \usepackage{unicode-math}
  \defaultfontfeatures{Scale=MatchLowercase}
  \defaultfontfeatures[\rmfamily]{Ligatures=TeX,Scale=1}
\fi
\usepackage{lmodern}
\ifPDFTeX\else  
\fi
\IfFileExists{upquote.sty}{\usepackage{upquote}}{}
\IfFileExists{microtype.sty}{
  \usepackage[]{microtype}
  \UseMicrotypeSet[protrusion]{basicmath} 
}{}
\makeatletter
\@ifundefined{KOMAClassName}{
  \IfFileExists{parskip.sty}{%
    \usepackage{parskip}
  }{
    \setlength{\parindent}{0pt}
    \setlength{\parskip}{6pt plus 2pt minus 1pt}}
}{
  \KOMAoptions{parskip=half}}
\makeatother
\usepackage{xcolor}
\makeatletter
\ifx\paragraph\undefined\else
  \let\oldparagraph\paragraph
  \renewcommand{\paragraph}{
    \@ifstar
      \xxxParagraphStar
      \xxxParagraphNoStar
  }
  \newcommand{\xxxParagraphStar}[1]{\oldparagraph*{#1}\mbox{}}
  \newcommand{\xxxParagraphNoStar}[1]{\oldparagraph{#1}\mbox{}}
\fi
\ifx\subparagraph\undefined\else
  \let\oldsubparagraph\subparagraph
  \renewcommand{\subparagraph}{
    \@ifstar
      \xxxSubParagraphStar
      \xxxSubParagraphNoStar
  }
  \newcommand{\xxxSubParagraphStar}[1]{\oldsubparagraph*{#1}\mbox{}}
  \newcommand{\xxxSubParagraphNoStar}[1]{\oldsubparagraph{#1}\mbox{}}
\fi
\makeatother

\usepackage{longtable,booktabs,array}
\usepackage{calc} 
\usepackage{etoolbox}
\makeatletter
\patchcmd\longtable{\par}{\if@noskipsec\mbox{}\fi\par}{}{}
\makeatother
\IfFileExists{footnotehyper.sty}{\usepackage{footnotehyper}}{\usepackage{footnote}}
\makesavenoteenv{longtable}
\usepackage{graphicx}
\makeatletter
\def\maxwidth{\ifdim\Gin@nat@width>\linewidth\linewidth\else\Gin@nat@width\fi}
\def\maxheight{\ifdim\Gin@nat@height>\textheight\textheight\else\Gin@nat@height\fi}
\makeatother
\setkeys{Gin}{width=\maxwidth,height=\maxheight,keepaspectratio}
\makeatletter
\def\fps@figure{htbp}
\makeatother

\makeatletter
\@ifpackageloaded{caption}{}{\usepackage{caption}}
\AtBeginDocument{%
\ifdefined\contentsname
  \renewcommand*\contentsname{Table of contents}
\else
  \newcommand\contentsname{Table of contents}
\fi
\ifdefined\listfigurename
  \renewcommand*\listfigurename{List of Figures}
\else
  \newcommand\listfigurename{List of Figures}
\fi
\ifdefined\listtablename
  \renewcommand*\listtablename{List of Tables}
\else
  \newcommand\listtablename{List of Tables}
\fi
\ifdefined\figurename
  \renewcommand*\figurename{Figure}
\else
  \newcommand\figurename{Figure}
\fi
\ifdefined\tablename
  \renewcommand*\tablename{Table}
\else
  \newcommand\tablename{Table}
\fi
}
\@ifpackageloaded{float}{}{\usepackage{float}}
\floatstyle{ruled}
\@ifundefined{c@chapter}{\newfloat{codelisting}{h}{lop}}{\newfloat{codelisting}{h}{lop}[chapter]}
\floatname{codelisting}{Listing}

\makeatother
\makeatletter
\@ifpackageloaded{caption}{}{\usepackage{caption}}
\@ifpackageloaded{subcaption}{}{\usepackage{subcaption}}
\makeatother

\ifLuaTeX
  \usepackage{selnolig}  
\fi
\usepackage[]{natbib}
\usepackage{bookmark}

\IfFileExists{xurl.sty}{\usepackage{xurl}}{} 
\hypersetup{
  pdftitle={Title},
  pdfauthor={Author 1; Author 2},
  pdfkeywords={3 to 6 keywords, that do not appear in the title},
  colorlinks=true,
  linkcolor={blue},
  filecolor={Maroon},
  citecolor={Blue},
  urlcolor={Blue},
  pdfcreator={LaTeX via pandoc}}

\newcommand{\anon}{1}

\begin{document}

\def\spacingset#1{\renewcommand{\baselinestretch}%
{#1}\small\normalsize} \spacingset{1}


\def\be{\begin{equation}}
\def\ee{\end{equation}} 
\def\ben{\begin{equation*}}
\def\een{\end{equation*}}
\def\bea{\begin{eqnarray}}
\def\eea{\end{eqnarray}}
\def\bda{\begin{eqnarray*}}
\def\eda{\end{eqnarray*}}
\numberwithin{equation}{section}

\newtheoremstyle{spacedplain}
  {8pt plus 2pt minus 1pt} 
  {8pt plus 2pt minus 1pt} 
  {\itshape}
  {0pt}
  {\bfseries}
  {.}
  {5pt plus 1pt minus 1pt}
  {}
\theoremstyle{spacedplain}
\newtheorem{definition}{Definition}
\newtheorem{theorem}{Theorem}
\newtheorem{proposition}{Proposition}
\newtheorem{corollary}{Corollary}
\newtheorem{assumption}{Assumption}
\renewcommand\theassumption{A\arabic{assumption}}
\newtheorem{lemma}{Lemma}
\newtheorem{remark}{Remark}
\newtheorem{example}{Example}

\newtheorem{exmp}{Example}[section]
\AtEndDocument{\refstepcounter{theorem}\label{finalthm}}
\AtEndDocument{\refstepcounter{proposition}\label{finalprop}}
\newcommand{\pkg}[1]{{\normalfont\fontseries{b}\selectfont #1}} \let\proglang=\textsf \let\code=\texttt
\setlength{\abovedisplayskip}{4pt}
\setlength{\belowdisplayskip}{4pt}

\newcommand{\dist}{\mathscr D}

\if1\anon
{
  \title{\bf Quantifying Long-Range Dependence in Object-Valued Time Series}
  \author{\normalsize Won-Ki Seo\thanks{Corresponding author: School of Economics, Level 5, Social Sciences Building, University of Sydney, Sydney, NSW 2050, Australia; Email: won-ki.seo@sydney.edu.au} \orcidlink{0000-0001-6629-3027} \\
  \vspace{-.14in}
    School of Economics \\
    University of Sydney\\
    \vspace{.14in}
    Jiazhen Xu \orcidlink{0009-0006-7870-0340} \ and \ Han Lin Shang\thanks{Shang is grateful for financial support from an Australian Research Council Future Fellowship (FT240100338).} \orcidlink{0000-0003-1769-6430}\\
    Department of Actuarial Studies and Business Analytics \\
    Macquarie University}
    \date{}
  \maketitle
} \fi

\if0\anon
{
  \bigskip
  \bigskip
  \bigskip
  \begin{center}
    {\LARGE\bf Quantifying Long-Range Dependence in Object-Valued Time Series}
\end{center}
  \medskip
} \fi

\bigskip
\begin{abstract}
Object-valued time series, including distributions, covariance matrices, networks, and compositions, lack the linear structure required by conventional autocovariance-based definitions of long memory. We develop an intrinsic framework for defining and estimating long memory in metric spaces of negative type. An isometric embedding yields a centered Hilbert-valued process, and the trace of its lag-covariance operator provides a signed, additive measure of temporal dependence. Crucially, this trace equals the difference between the marginal mean pairwise distance and expected lagged distance, so the framework and its estimators use only distances between the original objects. We define the memory parameter through the nonsummable decay of this trace and show that it coincides with the usual parameter in Hilbert-valued settings. Estimation uses Bartlett aggregates of distance-based lag measures. Estimating the common marginal distance from the same dependent sample induces a common-centering bias in these aggregates. We derive its finite-sample form and propose iterated block-difference corrections, log-ratio and multi-bandwidth log-slope estimators, and a localized self-consistency refinement. We establish consistency, document substantial bias reduction in simulations, and find long-memory evidence in foreign-exchange return distributions and U.S. electricity-generation compositions.
\end{abstract}

\noindent%
\vfill

\newpage
\spacingset{1.7} 

\section{Introduction}



The analysis of modern time series increasingly involves observations that take the form of random objects, such as probability distributions, covariance matrices, networks, and geometric shapes \citep{marron2021object}. Unlike Euclidean-valued observations, these objects generally reside in metric spaces without linear operations such as addition and subtraction. Recent work has developed autoregressive models for distributional, network, compositional, and spherical time series \citep[see, e.g.,][]{zhang2022wasserstein, chen2023wasserstein, zhu2023autoregressive, jiang2023autoregressive, zhu2024spherical, xu2026spherically}. Because these models are typically formulated under short-memory dynamics, assessing their adequacy requires determining whether the observed series has short or long memory. Ignoring long memory can distort variance estimation, invalidate standard inference, and impair long-horizon forecasting. Long memory has been documented in various applications, including age-specific fertility and mortality rates \citep[see, e.g.,][]{peters2021statistical,HNZ26}, intraday volatility curves \citep{KSZ26}, and volatility network time series \citep{boetti2025long}. These analyses, however, rely on a linear space structure, either intrinsic or induced by a suitable transformation, and do not extend to general object-valued time series.   

In scalar time series, long memory is characterized by nonsummable autocovariance decay or low-frequency behavior of the spectral density \citep[see, e.g.,][]{granger1980introduction,geweke1983estimation,robinson2003time}. For functional time series in Hilbert spaces, \citet{salish2019moment} develop a similar moment-based notion through autocovariance operators. These formulations rely on linear and inner-product structure for centering by subtraction, orthonormal basis expansions, and covariance operator analysis, none of which is available in a general metric space. A purely metric-based alternative is the autodistance covariance of \citet{zhou2012measuring}, recently extended to object-valued time series by \citet{jiang2024testing}. It characterizes lagwise independence but is quadratic in the embedded lag covariance operator, and the resulting notion of memory differs from the conventional one (Section~\ref{sect::compare with ADCV} of the supplementary material). The comparison illustrates what is required here: a metric-based dependence measure whose decay across lags yields a memory parameter compatible with the conventional one. To our knowledge, no general framework based on such a measure has been established.

We address this by introducing a canonical definition of long memory for object-valued time series taking values in metric spaces of negative type, a broad class that includes many object spaces used in practice \citep{lyons2013distance}. Such a space admits a natural Hilbert-space representation, under which the embedded process can be centered and its lag covariance operators are well defined. We show that the trace of the resulting lag covariance operator provides a signed measure of temporal dependence, and define long memory through the nonsummable decay of this trace as the lag increases. Importantly, the trace can be computed directly from pairwise distances among the original objects. Thus, although the Hilbert-space embedding provides the theoretical foundation, the definition itself is intrinsic: neither the embedding map nor the embedded observations need to be explicitly constructed. Under suitable operator-decay conditions, the resulting parameter recovers the conventional memory parameter for Hilbert-valued time series, thereby providing a common notion of memory across this broad class of object-valued time series. It also agrees with the natural memory parameter in the representative examples considered below.

Building on this definition, we develop an estimation framework based entirely on pairwise distances among the observed objects. The trace autocovariance at each lag can be estimated directly from distance averages, and the resulting lagwise estimates can then be combined through Bartlett-weighted aggregates whose growth across bandwidths identifies the memory parameter. This leads naturally to log-ratio and multi-bandwidth log-slope estimators that can be implemented directly in the original metric space, without imposing a linear structure on the object space, explicitly constructing the Hilbert-space embedding, or estimating projection directions. While the use of Bartlett-weighted aggregates is related to the variance-type approach of \citet{giraitis1999variance}, estimation in the present setting introduces a significant difficulty that does not arise in the same way for conventional covariance-based methods. The marginal distance level needed to center the lagwise dependence measures is unknown and must be estimated from the same dependent observations. The resulting centering error is repeated across all lags and accumulates when the lagwise estimates are aggregated, leading to potentially substantial finite-sample distortion even though the basic estimators remain consistent.

To address this problem, we first characterize the finite-sample effect of the common centering error and then exploit how this effect changes with sample size to target its leading component. This motivates a data-driven block-difference adjustment based on pairwise-distance averages computed over blocks of different sizes. Iterating the adjustment yields bias-corrected versions of the log-ratio and log-slope estimators, which are further improved through a localized self-consistency refinement. The resulting procedures admit zero as a value of the memory parameter, and cover both long- and short-memory regimes within a unified estimation framework.

Together, these developments provide, to our knowledge, a general framework for defining and estimating long memory in object-valued time series. Its main contributions are as follows:
\begin{itemize}
\item We define long memory intrinsically through pairwise distances in metric spaces of negative type, without requiring an explicit Hilbert embedding. The parameter agrees with its conventional counterpart for familiar examples of Hilbert-valued functional time series.
\item We develop estimators from the bandwidth growth of Bartlett aggregates and derive the exact effect of common centering on these aggregates. Iterated block-difference adjustment targets the leading component of this effect, followed by localized self-consistency refinement. All estimators admit zero, accommodating short memory.
\item We prove consistency under primitive dependence and rate conditions, distinguish the main finite-sample bias sources, and assess performance in extensive numerical studies.
\end{itemize}

The remainder of the paper is organized as follows. Section~\ref{sect::memory def} introduces the long-memory parameter for object-valued time series. Section~\ref{sect::estimation} develops the estimators, together with bias-corrected versions and a localized self-consistency refinement. Section~\ref{sec_sim} evaluates the finite-sample performance of the proposed methods through simulations, while the real data analysis on intraday log return distributional time series and the U.S. energy generation compositional time series can be found in Section~\ref{sect::real data}. Finally, Section~\ref{sec:conclu} synthesizes our primary contributions. To enhance readability, certain technical details and additional numerical results are deferred to the supplementary material.

\section{Long-memory parameter for object-valued time series}\label{sect::memory def}
\subsection{Preliminaries}
Let $(\Omega,\dist)$ be a separable metric space and fix $x_0\in\Omega$. Let $M_1(\Omega)$ denote the set of Borel probability measures $\mu$ satisfying $\int_\Omega \dist(x,x_0)\,d\mu(x)<\infty$. Let $\{X_t\}_{t\in\mathbb Z}$ be a strictly stationary $\Omega$-valued time series with marginal law $\nu\in M_1(\Omega)$. For independent $X,X'\sim\nu$, define the mean pairwise distance by $D_\nu=\mathbb E \dist(X,X')$. The triangle inequality gives $D_\nu\leq2\mathbb E \dist(X,x_0)<\infty$. Additional distance-moment conditions are imposed where needed. 

Throughout, we reserve $\dist$ for the metric and the upright $\mathrm d$ for the memory parameter introduced in Definition~\ref{def_trace}. Let $\mathcal H$ be a separable Hilbert space with inner product $\langle\cdot,\cdot\rangle_{\mathcal H}$ and norm $\|\cdot\|_{\mathcal H}$. For $h_1,h_2,h\in\mathcal H$, we use the tensor convention $(h_1\otimes h_2)h=\langle h_1,h\rangle_{\mathcal H}h_2$. For a trace-class operator $A$, write $\operatorname{tr}(A)$ for its trace and $\|A\|_1$ for its trace norm; $\|A\|_{\mathcal S_2}$ denotes the Hilbert--Schmidt norm. A measurable function $L:(0,\infty)\to(0,\infty)$ is slowly varying if $L(cx)/L(x)\to1$ for every fixed $c>0$. For real sequences $a_n$ and $b_n$ with $b_n\neq0$ eventually, $a_n\sim b_n$ means $a_n/b_n\to1$. We write $\to_p$ for convergence in probability and use $O_p(\cdot)$ and $o_p(\cdot)$ in their standard senses.

Establishing a canonical long-memory parameter requires isolating the second-order temporal decay rate, and we achieve this by focusing on metric spaces of negative type \citep{schoenberg1937certain,schoenberg1938metric}, which admit a natural Hilbert-space embedding. A metric space $(\Omega,\dist)$ is of negative type if, for arbitrary $x_1,\ldots,x_m\in\Omega$ and real numbers $r_1,\ldots,r_m$ satisfying $\sum_{i=1}^m r_i=0$, we have $\sum_{i=1}^m\sum_{j=1}^m r_i r_j \dist(x_i,x_j)\leq ~0$. Examples include $L^p$ spaces for $1\leq p\leq 2$, the unit sphere with great circle distance, and the hyperbolic space \citep[Theorem 3.6]{meckes2013positive}.

\begin{lemma}\label{lem_embed}
Suppose that $(\Omega,\dist)$ is of negative type, and let $\nu\in M_1(\Omega)$. Then the negative centered-distance kernel $K_\nu(x,x')=\int_\Omega \dist(x,u)\,d\nu(u)+\int_\Omega \dist(x',u)\,d\nu(u)-D_\nu-\dist(x,x')$ is positive semidefinite. More precisely, there exist a Hilbert space $\mathcal H$ and an isometric embedding $\phi:(\Omega,\dist^{1/2})\to\mathcal H$ such that $K_\nu(x,x')=2\langle\phi(x)-m_\phi(\nu),\phi(x')-m_\phi(\nu)\rangle_{\mathcal H}$, where $m_\phi(\nu)=\int_\Omega\phi(x)\,d\nu(x)$. Hence, with $\Phi_\nu(x)=\sqrt{2}\{\phi(x)-m_\phi(\nu)\}$, $K_\nu(x,x')=\langle\Phi_\nu(x),\Phi_\nu(x')\rangle_{\mathcal H}$.
\end{lemma}

Lemma~\ref{lem_embed} serves as a mathematical bridge. Because general metric spaces lack standard vector operations such as addition and inner products, classical covariance operators cannot be directly formulated. However, the isometric embedding guarantees that the negative centered-distance kernel $K_\nu$ in the original metric space is represented as an inner product in the Hilbert space $\mathcal H$. Furthermore, $\Phi_\nu(x)$ centers the embedded data. This motivates us to consider the mean-zero Hilbert-valued process $Y_t=\Phi_\nu(X_t)$, whose lag covariance operators capture temporal dependence after removing the nondecaying contribution of the embedded marginal mean.

The embedding identity and the standing condition $\nu\in M_1(\Omega)$ imply that $\mathbb E\|Y_t\|_{\mathcal H}^2=\mathbb E K_\nu(X_t,X_t)=D_\nu<\infty$. Because $\{Y_t\}$ is a well-defined, square-integrable, mean-zero process in a Hilbert space, we can leverage the standard functional time series machinery. We therefore define its lag-$k$ covariance operator as $\Gamma_k=\mathbb E(Y_t\otimes Y_{t-k})$. The following proposition establishes that this operator is trace-class and directly connects its trace to the pairwise distances in the original metric space.

\begin{proposition}\label{prop_cv}
Let $(\Omega,\dist)$ be of negative type and $Y_t=\Phi_\nu(X_t)$ be the centered Hilbert-valued process defined above. Then, for every $k\in\mathbb Z$, $\Gamma_k=\mathbb E(Y_t\otimes Y_{t-k})$ is trace-class and satisfies
\begin{equation}\label{eq_c}
C(k):=\operatorname{tr}(\Gamma_k)=\mathbb E\langle Y_t,Y_{t-k}\rangle_{\mathcal H}=D_\nu-\mathbb E \dist(X_t,X_{t-k}).
\end{equation}
\end{proposition}
We call $C(k)$ the \emph{trace autocovariance function} at lag $k$. As shown next, this signed, additive measure of second-order temporal dependence yields a canonical notion of long-range dependence.

\subsection{A canonical notion of long memory}

For the embedded process $Y_t=\Phi_\nu(X_t)$, the trace autocovariance function plays the same additive role in the variance of its partial sums as the conventional autocovariance  does in Euclidean time series: $\mathbb E\left\|\sum_{t=1}^nY_t\right\|_{\mathcal H}^2=nC(0)+2\sum_{k=1}^{n-1}(n-k)C(k)$. This identity motivates the use of $C(k)$ as the canonical covariance measure for object-valued time series. Because $\nu\in M_1(\Omega)$ already ensures that $C(k)$ is finite for every $k$, the following definition requires no stronger marginal moment conditions.

\begin{definition}\label{def_trace}
The object-valued time series $\{X_t\}$ has short memory if $\sum_{k=-\infty}^{\infty}|C(k)|<\infty$, in which case we assign the boundary value $\mathrm d=0$. It has long memory with parameter $\mathrm d\in(0,1/2)$ if there exist a constant $\const_C>0$ and a positive slowly varying function $L_C$ such that
\begin{equation}\label{eq_trace}
C(k)\sim\const_C k^{2\mathrm d-1}L_C(k),\quad k\to\infty.
\end{equation}
\end{definition}
\begin{remark}[Sign of the trace-memory coefficient]\label{rem_trace}
The positivity of $\const_C$ in Definition~\ref{def_trace} is forced by the covariance structure. To see this, consider the asymptotic relation \eqref{eq_trace} with a nonzero coefficient $\const_C$ of either sign. For $S_n=\sum_{t=1}^nY_t$, $\mathbb E\|S_n\|_{\mathcal H}^2=nC(0)+2\sum_{k=1}^{n-1}(n-k)C(k)\geq0$. Karamata's theorem (Lemma~\ref{lem_sum}) gives $2\sum_{k=1}^{n-1}(n-k)C(k)\sim\const_C\{\mathrm d(2\mathrm d+1)\}^{-1}n^{2\mathrm d+1}L_C(n)$, which dominates $nC(0)$ in absolute value. Hence $\const_C<0$ would make $\mathbb E\|S_n\|_{\mathcal H}^2$ negative for all large~$n$. Thus, positivity is inherent to the nonoscillatory regular-variation specification; oscillatory persistence would require a different tail or spectral formulation.
\end{remark}

The parameter $\mathrm d\in (0,1/2)$ is a genuine decay parameter only in the long-memory case. To place the two regimes within a common estimation framework, we adopt the boundary convention $\mathrm d=0$ under short memory. This convention does not impose a common decay rate on absolutely summable trace autocovariances; rather, it reflects the fact, established below, that the proposed estimators converge to zero under the short-memory condition.

The proposed definition is compatible with a natural operator notion of long memory for Hilbert-valued time series. Specifically, suppose that the lag covariance operators of the embedded process decay in a fixed operator direction at a nonsummable regularly varying rate:
\begin{equation}\label{eq_model}
\Gamma_k = k^{2\mathrm d_0-1}L(k)A+R_k, \qquad \mathrm d_0\in(0,1/2),
\end{equation}
where $A$ is a trace-class operator and $L$ is positive and slowly varying. Under this representation, $\mathrm d_0$ determines the temporal decay rate of the covariance operators, while $A$ describes its operator direction. The following proposition shows that the trace autocovariance function inherits the same decay rate whenever the leading operator component has a nonvanishing trace.

\begin{proposition}\label{prop_common}
Suppose that~\eqref{eq_model} holds and $\|R_k\|_1=o\bigl(k^{2\mathrm d_0-1}L(k)\bigr)$. If, in addition, $\operatorname{tr}(A)\neq0$, then $C(k)\sim \operatorname{tr}(A)k^{2\mathrm d_0-1}L(k)$, and hence the long-memory parameter in Definition~\ref{def_trace} is $\mathrm d=\mathrm d_0$.
\end{proposition}

Thus, when regularly decaying covariance operators characterize long memory of the (embedded) Hilbert-valued process, Definition~\ref{def_trace} recovers its memory parameter from their traces.

We next illustrate Definition~\ref{def_trace} in three examples, each built from an underlying process with memory parameter $\mathrm d_{\rm E}\in(0,1/2)$. All three have long memory in the sense of Definition~\ref{def_trace}, with $\mathrm d=\mathrm d_{\rm E}$, as established in Proposition~\ref{prop::example memory parameter}. Throughout, $c>0$ and $L$ is positive and slowly varying.

\begin{example}[Real-valued time series]\label{ex_gaussian}
Consider $\Omega=\mathbb R$ with the Euclidean metric $\dist(x,y)=|x-y|$, and let $\{X_t\}$ be a standard stationary Gaussian process with $\gamma(k)\sim ck^{2\mathrm d_{\rm E}-1}L(k)$.
\end{example}

\begin{example}[Matrix-valued time series]\label{ex_psd}
Consider $\Omega=S_p^+$ with the Frobenius metric $\dist(x,y)=\|x-y\|_{\mathrm F}$, where $S_p^+$ is the cone of $p\times p$ positive-semidefinite (PSD) matrices. Let $X_t=M_0+(u+\vartheta G_tv)(u+\vartheta G_tv)^\top$, where $M_0\in S_p^+$ is fixed, $\vartheta>0$, $u,v\in\mathbb R^p$ are orthonormal, and $\{G_t\}$ is a standard stationary Gaussian process with $\gamma(k)\sim ck^{2\mathrm d_{\rm E}-1}L(k)$.
\end{example}

\begin{example}[Distributional time series]\label{ex_wass}
Consider $\Omega=\mathcal P_2(\mathbb R)$, the space of probability distributions on $\mathbb R$ with finite second moments, equipped with the $\mathcal L^2$-Wasserstein metric $\dist=W_2$. Let $X_t=N(\mu_t,\sigma^2)$, where $\{\mu_t\}$ is a standard stationary Gaussian process with $\gamma(k)\sim ck^{2\mathrm d_{\rm E}-1}L(k)$, or $X_t=N(\mu_t,\sigma_t^2)$ with $\sigma_t=\sigma\exp(\lambda\mu_t)$ for $\sigma>0$ and $\lambda\geq0$.
\end{example}

The following result states $\mathrm d=\mathrm d_{\rm E}$ formally; the proof is in the supplementary material.
\begin{proposition}\label{prop::example memory parameter}
For the metric spaces and the time series in Examples~\ref{ex_gaussian}--\ref{ex_wass}, $\mathrm d=\mathrm d_{\rm E}$ for $\mathrm d$ in Definition~\ref{def_trace}.
\end{proposition}


A natural alternative is to define long memory through the autodistance covariance (ADCV) \citep{zhou2012measuring,lyons2013distance,jiang2024testing}, which characterizes lagwise independence in metric spaces of strong negative type. Although well suited to detecting general dependence, ADCV is quadratic in the embedded lag covariance operator and therefore generally yields a different decay rate and summability threshold from those of conventional autocovariance. An ADCV-based definition thus does not reproduce the correspondence established in Proposition~\ref{prop_common}, under which the trace autocovariance inherits the decay rate of the lag covariance operators. Section~\ref{sect::compare with ADCV} of the supplementary material compares the two notions in further detail.


\section{Estimation of long-memory parameter}\label{sect::estimation}

We develop estimators of the long-memory parameter from pairwise distances among the observed objects. Section~\ref{subsect::bias} constructs a distance-based estimator of the trace autocovariance and identifies the common centering bias induced by estimating the marginal distance level $D_\nu$. Section~\ref{subsect::raw est} combines the lagwise estimates through punctured Bartlett aggregates to obtain raw log-ratio and log-slope estimators. Section~\ref{subsect::bias correction} then develops an iterated block-difference correction for this bias and a localized self-consistency refinement.

\subsection{Distance averages and the common centering bias}\label{subsect::bias}

Let the object-valued time series $X_1,\ldots, X_{n}$ be observed. For $1\leq k<n$, define the global sample pairwise-distance average and the lag-$k$ sample mean distance by
\begin{equation}\label{eq_dhat} 
\widehat D_n=\frac{1}{n(n-1)}\sum_{1\leq i\neq j\leq n}\dist(X_i,X_j),\qquad \widehat\delta_n(k)=\frac{1}{n-k}\sum_{t=k+1}^n \dist(X_t,X_{t-k}). 
\end{equation}
Here, $\widehat D_n$ is the sample analog of the population mean pairwise distance $D_\nu$. Writing $\delta(k)=\mathbb E \dist(X_t,X_{t-k})$ for the population lag-$k$ mean distance, $\widehat\delta_n(k)$ is its sample analog. Accordingly, a natural estimator of $C(k)$ is $\widehat C_n(k)=\widehat D_n-\widehat\delta_n(k)$. The next lemma gives an exact finite-sample expression for the bias of $\widehat D_n$ and characterizes its asymptotic behavior under long and short memory.
\begin{lemma}[Bias of the global pairwise-distance average]\label{lem_dhat}
The estimator $\widehat D_n$ satisfies 
\begin{equation}\label{eq_dhat_1} 
\mathbb E\widehat D_n-D_\nu=-\frac{2}{n(n-1)}\sum_{k=1}^{n-1}(n-k)C(k). 
\end{equation}
If the long memory behavior in \eqref{eq_trace} holds, then
\begin{equation}\label{eq_dhat_2} 
\mathbb E\widehat D_n-D_\nu\sim-\frac{\const_C}{\mathrm d(2\mathrm d+1)}n^{2\mathrm d-1}L_C(n). \end{equation}
If $\sum_{k\geq1}|C(k)|<\infty$, then $\mathbb E\widehat D_n-D_\nu=O(n^{-1})$. If, in addition, $\sum_{k\geq1}k|C(k)|<\infty$ and $\mathsf B_\infty=\sum_{k\neq0}C(k)$, then
\begin{equation}\label{eq_dhat_3} 
\mathbb E\widehat D_n-D_\nu=-{n^{-1}}{\mathsf B_\infty}+O(n^{-2}). 
\end{equation}
\end{lemma}


These expansions of the bias of $\widehat D_n$ serve as input to the aggregation and correction developed below.

\subsection{Raw estimation with punctured Bartlett aggregates}\label{subsect::raw est}

We first introduce two easily computed estimators based on Bartlett-weighted aggregates of $\widehat C_n(k)$ and establish their asymptotic properties. Although consistent, they may exhibit substantial finite-sample bias due to the common centering error in $\widehat D_n$. We therefore use either as a preliminary input (\emph{raw pilot}) for the block-difference correction in the next section.


Fix a bandwidth multiplier $q>1$ and let $m=m_n$ be the base Bartlett bandwidth, with $m\to\infty$ and $m/n\to0$. Define the integer bandwidth grid $\mathcal G_{m}=\{m,m+1,\ldots,\lfloor qm\rfloor\}$. We define
\begin{equation}
\mathsf B_C(r)=\sum_{0<|k|<r}\left(1-\frac{|k|}{r}\right)C(k)=2\sum_{k=1}^{r-1}\left(1-\frac{k}{r}\right)C(k), \quad  r\in\mathcal G_{m}=\{m,\ldots,\lfloor qm\rfloor\}, \label{eq_bartlett_1}
\end{equation}
where the second equality follows from $C(-k)=C(k)$. We call \eqref{eq_bartlett_1} the \emph{punctured Bartlett aggregate}, with ``punctured'' referring to the omission of lag zero. Under \eqref{eq_trace}, Karamata's theorem yields
\begin{equation}\label{eq_bartlett_growth}
\mathsf B_C(r)\sim\frac{\const_C}{\mathrm d(2\mathrm d+1)}r^{2\mathrm d}L_C(r),\qquad r\to\infty,
\end{equation}
so $\mathrm d$ is identified from its growth across bandwidths. Under short memory, $\mathsf B_C(r)\to\mathsf B_\infty=\sum_{k\neq0}C(k)=2\sum_{k\geq 1}C(k)$, the \emph{punctured long-run covariance}. Since $C(0)=D_\nu$, including lag zero replaces $\mathsf B_C(r)$ by $D_\nu+\mathsf B_C(r)$. This leaves the memory exponent unchanged but flattens the finite-bandwidth slope and shifts the short-memory limit from $\mathsf B_\infty$ to $D_\nu+\mathsf B_\infty$.

Under short memory, absolute summability of $C(k)$ is part of Definition~\ref{def_trace}. For estimation, we additionally impose the \emph{one-summability} condition $\sum_{k\geq1}k|C(k)|<\infty$, which yields $\mathsf B_C(r)-\mathsf B_\infty=O(r^{-1})$. The estimators proposed below take logarithms of the absolute aggregate before differencing across bandwidths. If $\mathsf B_\infty=0$, then $|\mathsf B_C(r)|=O(m^{-1})$ uniformly over $m\leq r\leq\lfloor qm\rfloor$, making these logarithms unstable. We therefore define the vanishing stabilization coefficient $a_{n,m}=c_a(m/n)^\eta$, where $c_a,\eta>0$, and require $m^{-1}=o(a_{n,m})$ in this case. Provided $D_\nu>0$, the term $2a_{n,m}D_\nu$ then dominates the zero-limit short-memory remainder but vanishes relative to the diverging punctured aggregate under long memory. We call this ridge term the \emph{vanishing stabilizer}. The factor of two places it under the same two-sided normalization as \eqref{eq_bartlett_1}. The proposed estimators are based on the absolute values of the stabilized Bartlett aggregates:
\begin{equation*}
S_{B,n}(r;m)=|\mathsf B_{C,n}(r;m)|,\qquad
\widehat S_{B,n}(r;m)=|\widehat{\mathsf B}_{C,n}(r;m)|,
\end{equation*}
where $\mathsf B_{C,n}(r;m)=2a_{n,m}D_\nu+\mathsf B_C(r)$ and $\widehat{\mathsf B}_{C,n}(r;m)=2a_{n,m}\widehat D_n+2\sum_{k=1}^{r-1}(1-k/r)\widehat C_n(k)$ are the population and sample stabilized signed aggregates, respectively. We refer to $S_{B,n}(r;m)$ and $\widehat S_{B,n}(r;m)$ as the population and sample stabilized Bartlett scales. The following proposition gives their scaling under the long- and short-memory regimes.
\begin{proposition}[Scaling of the stabilized aggregate]\label{prop_scale}
The following hold uniformly in $\ell\in[1,q]$:
\begin{enumerate}[label=(\roman*)]
\item\label{prop_scale1} Under long-memory condition \eqref{eq_trace}, ${S_{B,n}(\lfloor\ell m\rfloor;m)}/\{m^{2\mathrm d}L_C(m)\}\rightarrow {\const_C\ell^{2\mathrm d}}/\{\mathrm d(2\mathrm d+1)\}$.
\item\label{prop_scale2} Under short memory and $\sum_{k\geq1}k|C(k)|<\infty$, if $\mathsf B_\infty\neq0$, then $S_{B,n}(\lfloor\ell m\rfloor;m)\to|\mathsf B_\infty|$.
\item\label{prop_scale3} Under short memory and $\sum_{k\geq1}k|C(k)|<\infty$, if $\mathsf B_\infty=0$, $m^{-1}=o(a_{n,m})$, and $D_\nu>0$, then $S_{B,n}(\lfloor\ell m\rfloor;m)\sim2a_{n,m}D_\nu$.
\end{enumerate}
Consequently, $S_{B,n}(\lfloor\ell m\rfloor;m)/S_{B,n}(m;m)$ converges uniformly to $\ell^{2\mathrm d}$ under long memory and to one under short memory.
\end{proposition}

We reserve the subscripts $\mathrm R$ and $\mathrm S$ for the log-ratio and log-slope constructions introduced below, and write $\esttype\in\{\mathrm R,\mathrm S\}$ when a statement covers both. For a closed interval $\mathcal I=[a,b]$, let $\Pi_{\mathcal I}(x)=\min\{b,\max(a,x)\}$ denote the projection of $x\in\mathbb R$ onto $\mathcal I$, that is, the point in $\mathcal I$ closest to $x$. Let $\mathcal I_{\mathrm O}$ be a prespecified output interval chosen to reflect the admissible or plausible range of $\mathrm d$. For example, $\mathcal I_{\mathrm O}=[0,1/2]$ enforces the theoretical memory range, while a wider interval allows finite-sample deviations. The two-bandwidth log-ratio estimator is
\begin{equation}\label{eq_raw_r}
\widehat{\mathrm d}_{\mathrm R}(q,m)=\Pi_{\mathcal I_{\mathrm O}}\!\left[\frac{\log\{\widehat S_{B,n}(\lfloor qm\rfloor;m)/\widehat S_{B,n}(m;m)\}}{2\log\{\lfloor qm\rfloor/m\}}\right].
\end{equation}
For the multi-bandwidth ordinary-least-squares (OLS) log-slope estimator, let $\bar\ell_m=|\mathcal G_m|^{-1}\sum_{r\in\mathcal G_m}\log r$ and define
\begin{equation}\label{eq_raw_s}
\widehat{\mathrm d}_{\mathrm S}(q,m)=\Pi_{\mathcal I_{\mathrm O}}\!\left[\frac{\sum_{r\in\mathcal G_m}(\log r-\bar\ell_m)\log\widehat S_{B,n}(r;m)}{2\sum_{r\in\mathcal G_m}(\log r-\bar\ell_m)^2}\right].
\end{equation}
The projection onto $\mathcal I_{\mathrm O}$ in \eqref{eq_raw_r}--\eqref{eq_raw_s} is part of each raw estimator, not a reporting-only truncation. If any scale entering a logarithm is zero, we set the corresponding unprojected estimator to zero; the consistency conditions below imply that this convention is used with probability tending to zero. The factor $1/2$ in both estimators converts the Bartlett growth exponent $2\mathrm d$ into $\mathrm d$ and is unrelated to the one- or two-sided normalization.

The bias of $\widehat D_n$ enters the aggregate at every lag. The following lemma makes its cumulative effect exact, and this expression underlies the bias-corrected estimators developed in Section~\ref{subsect::bias correction}.
\begin{lemma}[Exact effect of the common centering bias]\label{lem_raw_bias}
Let $b_n=\mathbb E\widehat D_n-D_\nu$. For every $r\in\mathcal G_m$,
\begin{equation}\label{eq_raw_bias}
\mathbb E\widehat{\mathsf B}_{C,n}(r;m)=\mathsf B_{C,n}(r;m)+\{2a_{n,m}+r-1\}b_n.
\end{equation}
\end{lemma}


\begin{remark}[Role of the vanishing stabilizer]\label{rem_ridge}
As discussed above, under the stated conditions the vanishing stabilizer dominates the $O(m^{-1})$ remainder when $\mathsf B_\infty=0$, while remaining negligible relative to the diverging aggregate under long memory. Scaling by $\widehat D_n$ gives the stabilizer the same units as $\widehat C_n(k)$ and preserves equivariance under rescaling of the metric. Any positive scale-equivariant statistic could be used instead. Such a replacement affects only the stabilizing term of order $a_{n,m}$, whereas the repeated centering error from the lagwise estimates enters with coefficient $r-1$. Since $a_{n,m}=o(r)$ for $r\in\mathcal G_m$, changing the scaling statistic cannot replace the block-difference correction developed below.
\end{remark}

We first establish consistency under a high-level condition that controls the signed Bartlett aggregates uniformly over the integer bandwidths used by the log-ratio and log-slope estimators. Section~\ref{sec_conditions} then translates it into primitive moment, dependence, and rate conditions.
\begin{assumption}[Accuracy of the raw signed punctured aggregates]\label{ass_raw}
Under long memory, assume $\max_{r\in\mathcal G_m}|\widehat{\mathsf B}_{C,n}(r;m)-\mathsf B_{C,n}(r;m)|=o_p\{m^{2\mathrm d}L_C(m)\}$.
Under short memory, assume $\max_{r\in\mathcal G_m}|\widehat{\mathsf B}_{C,n}(r;m)-\mathsf B_{C,n}(r;m)|=o_p\{|\mathsf B_\infty|+2a_{n,m}D_\nu\}$.
\end{assumption}

\begin{theorem}[Consistency of the raw estimators]\label{thm_raw}
Let $q>1$ be fixed, let $m=m_n\to\infty$ with $m/n\to0$, and let $\mathcal I_{\mathrm O}$ be a fixed closed interval containing $\mathrm d$, where $\mathrm d=0$ denotes the short-memory convention. Suppose that Assumption~\ref{ass_raw} holds together with the conditions in Proposition~\ref{prop_scale}\ref{prop_scale1}, \ref{prop_scale2}, or \ref{prop_scale3}, as appropriate. Then all logarithms in \eqref{eq_raw_r}--\eqref{eq_raw_s} are well defined with probability tending to one, and $\widehat{\mathrm d}_{\esttype}(q,m)\to_p\mathrm d$ for $\esttype\in\{\mathrm R,\mathrm S\}$.
\end{theorem}

\subsection{Block-difference bias correction and local refinement}\label{subsect::bias correction}
\subsubsection{An iterated block-difference adjustment}
Although the raw estimators are consistent under the conditions of Theorem~\ref{thm_raw}, the common centering bias $b_n=\mathbb E\widehat D_n-D_\nu$ can cause substantial finite-sample distortion through its accumulation across lags; see Lemma~\ref{lem_raw_bias} and the simulation results in Section~\ref{sec_sim}. We therefore target this bias through block differences and adjust $\widehat D_n$ upward by the estimated magnitude of its negative part.
Fix $\esttype\in\{\mathrm R,\mathrm S\}$ and let $\widehat{\mathrm d}_{\mathrm P}$ denote the projected pilot used in a block-difference update. Initially, $\widehat{\mathrm d}_{\mathrm P}=\Pi_{\mathcal I_{\mathrm P}}(\widehat{\mathrm d}_{\esttype})$, where $\mathcal I_{\mathrm P}=[-\varepsilon_{\mathrm P},\overline{\mathrm d}_{\mathrm P}]$, $\varepsilon_{\mathrm P}>0$, and $0<\overline{\mathrm d}_{\mathrm P}<1/2$. Subsequent updates use the latest estimate in place of $\widehat{\mathrm d}_{\esttype}$. The upper bound keeps the block regression away from degeneracy at $1/2$; see Remark~\ref{rem_block}. For the initial pilot, Theorem~\ref{thm_raw} gives $\widehat{\mathrm d}_{\mathrm P}\to_p\mathrm d$ whenever $\mathrm d\in\mathcal I_{\mathrm P}$. 

Let $\mathcal S=\{1=s_1<s_2<\cdots<s_J\}$ be a fixed set of integer block counts with $J\geq2$. For $s\in\mathcal S$, set $\ell_{n,s}=\lfloor n/s\rfloor$ and form the consecutive blocks $I_{j,s}=\{(j-1)\ell_{n,s}+1,\ldots,j\ell_{n,s}\}$, $j=1,\ldots,s$. Define the within-block pairwise-distance averages and their average across blocks by
\begin{equation*}
\widehat D_{j,s}=\frac{1}{\ell_{n,s}(\ell_{n,s}-1)}\sum_{\substack{u,v\in I_{j,s}\\u\neq v}}\dist(X_u,X_v),\qquad \widehat D_{n,s}=\frac{1}{s}\sum_{j=1}^s\widehat D_{j,s}.
\end{equation*}
In particular, $\widehat D_{n,1}=\widehat D_n$. By stationarity, applying Lemma~\ref{lem_dhat} at the block length $\ell_{n,s}$ gives the following expansions.

\begin{lemma}[Bias expansion for $\widehat D_{n,s}$]\label{lem_block_bias}
Uniformly over $s\in\mathcal S$, $\mathbb E\widehat D_{n,s}=D_\nu+b_n(\ell_{n,s}/n)^{2\mathrm d-1}+o(|b_n|)$ under \eqref{eq_trace}, and $\mathbb E\widehat D_{n,s}=D_\nu- \ell_{n,s}^{-1}\mathsf B_\infty+O(\ell_{n,s}^{-2})$ under $\sum_{k\geq1}k|C(k)|<\infty$.
\end{lemma}

For a trial memory value $\eth<1/2$, define the block-difference regressor
\begin{equation}\label{eq_block_x}
x_{n,s}(\eth)=\left({\ell_{n,s}}/{n}\right)^{2\eth-1}-1,\qquad s\in\mathcal S=\{1=s_1<\cdots<s_J\}.
\end{equation}
Under \eqref{eq_trace}, subtracting $\mathbb E\widehat D_n=D_\nu+b_n$ from the expansion in Lemma~\ref{lem_block_bias} eliminates $D_\nu$ and gives $\mathbb E(\widehat D_{n,s}-\widehat D_n)=b_nx_{n,s}(\mathrm d)+o(|b_n|)$.
To account for finite block lengths, define, for integers $\ell\geq2$,
\begin{equation}\label{eq_block_finite}
g_\ell(\eth)=\frac{2}{\ell(\ell-1)}\sum_{k=1}^{\ell-1}(\ell-k)k^{2\eth-1},\qquad
x_{n,s}^{\circ}(\eth)=\frac{g_{\ell_{n,s}}(\eth)}{g_n(\eth)}-1.
\end{equation}
Under the exact power-law benchmark $C(k)=\const_C k^{2\mathrm d-1}$ for every $k\geq1$, with $0<\mathrm d<1/2$, the identity \eqref{eq_block_bias_4} in the supplementary material gives $b_n=-\const_C g_n(\mathrm d)$ and $\mathbb E(\widehat D_{n,s}-\widehat D_n)=b_nx_{n,s}^{\circ}(\mathrm d)$ exactly. Moreover, $g_\ell(\eth)\sim\ell^{2\eth-1}/\{\eth(2\eth+1)\}$ and $\max_{s\in\mathcal S}|x_{n,s}^{\circ}(\eth)-x_{n,s}(\eth)|\to0$, uniformly for~$\eth$ in compact subsets of $(0,1/2)$. Thus, the refined regressor matches the exact block-bias ratio under this benchmark and retains the asymptotic scaling in \eqref{eq_block_x}.
Using this regressor, define the no-intercept least-squares coefficient as a function of the trial value:
\begin{equation}\label{eq_block_est}
\widehat b_n(\eth)=\frac{\sum_{s\in\mathcal S\setminus\{1\}}x_{n,s}^{\circ}(\eth)(\widehat D_{n,s}-\widehat D_n)}{\sum_{s\in\mathcal S\setminus\{1\}}\{x_{n,s}^{\circ}(\eth)\}^2}.
\end{equation}
The regression weights all nontrivial block counts equally and requires no intercept because differencing removes $D_\nu$. At the current projected pilot, we write $\widehat b_n=\widehat b_n(\widehat{\mathrm d}_{\mathrm P})$ for the resulting bias estimate.

\begin{remark}[Degeneration of the block regressors]\label{rem_block}
The block regression is not uniformly informative as $\eth\uparrow1/2$. For each fixed $\eth\in(0,1/2)$, the denominator in \eqref{eq_block_est} satisfies
\begin{equation*}
\sum_{s\in\mathcal S\setminus\{1\}}\{x_{n,s}^{\circ}(\eth)\}^2\to H(\eth):=\sum_{s\in\mathcal S\setminus\{1\}}\{s^{1-2\eth}-1\}^2.
\end{equation*}
For instance, if $\mathcal S=\{1,2,4,8,16\}$, then $H(.40)=.939$, $H(.45)=.183$, and $H(.49)=.006$. More generally, $s^{1-2\eth}-1=(1-2\eth)\log s+o(1-2\eth)$, so $H(\eth)$ vanishes quadratically. This degeneration motivates restricting the trial memory value $\eth$ to a prespecified interval whose upper endpoint is strictly below $1/2$. This is a restriction on the input to the block regression, not an ex-post restriction on the reported estimate.
\end{remark}

\begin{proposition}[Population block regression]\label{prop_block}
Suppose that \eqref{eq_trace} holds. For $\eth<1/2$, define the population counterpart of \eqref{eq_block_est} by
\begin{equation}\label{eq_block_pop}
\bar b_n(\eth)=\frac{\sum_{s\in\mathcal S\setminus\{1\}}x_{n,s}^{\circ}(\eth)\{\mathbb E\widehat D_{n,s}-\mathbb E\widehat D_n\}}{\sum_{s\in\mathcal S\setminus\{1\}}\{x_{n,s}^{\circ}(\eth)\}^2}.
\end{equation}
Then, $\bar b_n(\mathrm d)=b_n\{1+o(1)\}$. If $\widehat{\mathrm d}_{\mathrm P}\to_p\mathrm d$, then $\bar b_n(\widehat{\mathrm d}_{\mathrm P})=b_n\{1+o_p(1)\}$.
\end{proposition}

Under long memory, \eqref{eq_dhat_2} implies $b_n<0$ for all sufficiently large $n$, motivating an upward adjustment to $\widehat D_n$. This correction targets the long-memory bias and should become inactive under short memory, while Remark~\ref{rem_block} shows that the block regressors degenerate as the trial value approaches $1/2$. We thus apply the correction only when the projected pilot lies in a prespecified range bounded away from zero and $1/2$. Choose fixed thresholds $0<\underline{\mathrm d}<\overline{\mathrm d}<1/2$ and define
\begin{equation}\label{eq_corr}
\widehat c_n=
\begin{cases}
\max\{-\widehat b_n,0\},&\underline{\mathrm d}<\widehat{\mathrm d}_{\mathrm P}\leq\overline{\mathrm d},\\
0,&\text{otherwise}.
\end{cases}
\end{equation}
\begin{remark}[Pilot projection and bias correction]\label{rem_pilot_activation}
The pilot interval $\mathcal I_{\mathrm P}$ bounds the regressor input and includes zero to preserve pilot consistency under short memory. The block correction, however, relies on the long-memory expansion in Lemma~\ref{lem_block_bias}, whose justification requires $\mathrm d>0$. We therefore use a fixed cutoff $\underline{\mathrm d}>0$, making correction inactive with probability tending to one under short memory when the pilot is consistent. The upper pilot bound avoids degeneration at~$1/2$; see Remark~\ref{rem_block}. In the benchmark, $\overline{\mathrm d}=\overline{\mathrm d}_{\mathrm P}$, so pilots above this bound are projected to it and remain eligible for correction.
\end{remark}
The adjusted distance level and trace autocovariance estimates are $\widehat D_n^{\mathrm{BD}}=\widehat D_n+\widehat c_n$ and $\widehat C_n^{\mathrm{BD}}(k)=\widehat D_n^{\mathrm{BD}}-\widehat\delta_n(k)$.
The corresponding signed Bartlett aggregate and its absolute scale are
\begin{equation*}
\widehat{\mathsf B}_{C,n,\mathrm{BD}}(r;m)=\widehat{\mathsf B}_{C,n}(r;m)+(r-1)\widehat c_n,\qquad \widehat S_{B,n,\mathrm{BD}}(r;m)=|\widehat{\mathsf B}_{C,n,\mathrm{BD}}(r;m)|.
\end{equation*}
Here the stabilizer remains scaled by the raw $\widehat D_n$. Using $\widehat D_n^{\mathrm{BD}}$ instead would add $2a_{n,m}\widehat c_n$, which is asymptotically negligible under Assumption~\ref{ass_bd} but may affect finite-sample slopes. The adjusted log-ratio and log-slope estimators are
\begin{equation}\label{eq_bd_est}
\begin{aligned}
\widehat{\mathrm d}_{\mathrm R,\mathrm{BD}}(q,m)&=\Pi_{\mathcal I_{\mathrm O}}\!\left[\frac{\log\{\widehat S_{B,n,\mathrm{BD}}(\lfloor qm\rfloor;m)/\widehat S_{B,n,\mathrm{BD}}(m;m)\}}{2\log\{\lfloor qm\rfloor/m\}}\right],\\
\widehat{\mathrm d}_{\mathrm S,\mathrm{BD}}(q,m)&=\Pi_{\mathcal I_{\mathrm O}}\!\left[\frac{\sum_{r\in\mathcal G_m}(\log r-\bar\ell_m)\log\widehat S_{B,n,\mathrm{BD}}(r;m)}{2\sum_{r\in\mathcal G_m}(\log r-\bar\ell_m)^2}\right].
\end{aligned}
\end{equation}

For a trial value $\eth\in\mathbb R$, let $\widehat c_n(\eth)$ be given by \eqref{eq_corr} with $\widehat{\mathrm d}_{\mathrm P}$ replaced by $\Pi_{\mathcal I_{\mathrm P}}(\eth)$ and $\widehat b_n$ by $\widehat b_n\{\Pi_{\mathcal I_{\mathrm P}}(\eth)\}$. For $\esttype\in\{\mathrm R,\mathrm S\}$, define the pilot-to-estimate map $\mathcal T_{n,\esttype}$ through
\begin{equation*}
\eth\ \xrightarrow{\ \text{block correction}\ }\ \widehat c_n(\eth)\ \xrightarrow{\ \text{adjust original }\widehat D_n\ }\ \widehat D_n+\widehat c_n(\eth)\ \xrightarrow{\ \text{log-ratio or log-slope}\ }\ \mathcal T_{n,\esttype}(\eth),
\end{equation*}
where the final step uses \eqref{eq_bd_est}. Starting from the raw estimator, define
\begin{equation}\label{eq_update}
\widehat{\mathrm d}_{\esttype}^{(0)}=\widehat{\mathrm d}_{\esttype},\qquad \widehat{\mathrm d}_{\esttype}^{(j)}=\mathcal T_{n,\esttype}(\widehat{\mathrm d}_{\esttype}^{(j-1)}),\qquad j=1,2.
\end{equation}
We call $\widehat{\mathrm d}_{\esttype}^{(0)}=\widehat{\mathrm d}_{\esttype}$ the raw estimator (Raw) and $\widehat{\mathrm d}_{\esttype}^{(2)}$ the two-step bias-corrected estimator (BC). The second update uses $\widehat{\mathrm d}_{\esttype}^{(1)}$ as a bias-corrected pilot to reduce the influence of Raw's substantial finite-sample bias, analyzed in Section~\ref{sec_bias} and illustrated in Section~\ref{sec_sim}.
Each update recomputes the correction to the original $\widehat D_n$; corrections are not accumulated. Section~\ref{sec_localrefine} develops a further local refinement.
The following condition ensures that the adjustment preserves the first-order Bartlett scaling; primitive sufficient conditions are given in Proposition~\ref{prop_rates}.
\begin{assumption}[Asymptotic size of the block-difference adjustment]\label{ass_bd}
Under long memory, $\widehat c_n=o_p\{m^{2\mathrm d-1}L_C(m)\}$. Under short memory, $\Pr(\widehat c_n=0)\to1$.
\end{assumption}
\begin{theorem}[Consistency of a block-difference update]\label{thm_bd}
Suppose that the conditions of Theorem~\ref{thm_raw} hold, together with Assumption~\ref{ass_bd}. Then both estimators in \eqref{eq_bd_est} converge in probability to $\mathrm d$.
\end{theorem}
The same argument applies to any fixed number of updates.
\begin{corollary}[A fixed number of pilot updates]\label{cor_update}
Let $\esttype\in\{\mathrm R,\mathrm S\}$ and extend the recursion in \eqref{eq_update} to $j\geq1$. Suppose that the conditions of Theorem~\ref{thm_raw} hold and the correction $\widehat c_n(\widehat{\mathrm d}_{\esttype}^{(j-1)})$ satisfies Assumption~\ref{ass_bd} at each fixed step. Then $\widehat{\mathrm d}_{\esttype}^{(j)}\to_p\mathrm d$ for every fixed $j\geq1$.
\end{corollary}

\subsubsection{Local refinement around the bias-corrected pilot}\label{sec_localrefine}

We further refine BC (i.e., $\widehat{\mathrm d}_{\esttype}^{(2)}$) by seeking a nearby trial value that is approximately unchanged by the pilot-to-estimate map $\mathcal T_{n,\esttype}$. For $\esttype\in\{\mathrm R,\mathrm S\}$, define the squared discrepancy $\mathcal Q_{n,\esttype}(\eth)=\{\mathcal T_{n,\esttype}(\eth)-\eth\}^2$. Using the last update size $\widehat\varrho_{n,\esttype}=|\widehat{\mathrm d}_{\esttype}^{(2)}-\widehat{\mathrm d}_{\esttype}^{(1)}|$ as the search radius, set
\begin{equation}\label{eq_fp_set}
\mathcal N_{n,\esttype}=\big[\widehat{\mathrm d}_{\esttype}^{(2)}-\widehat\varrho_{n,\esttype},\widehat{\mathrm d}_{\esttype}^{(2)}+\widehat\varrho_{n,\esttype}\big]\cap\mathcal I_{\mathrm O}.
\end{equation}
This choice determines the search radius from the observed updates without an additional tuning constant. Other neighborhoods centered at BC may also be used; the consistency argument in Theorem~\ref{thm_fp} applies whenever their radii converge to zero in probability. The localized minimum-discrepancy estimator, denoted FP, is
\begin{equation}\label{eq_fp_est}
\widehat{\mathrm d}_{\esttype,\mathrm{FP}}\in\argmin_{\eth\in\mathcal N_{n,\esttype}}\mathcal Q_{n,\esttype}(\eth).
\end{equation}
An exact fixed point is selected when the minimum discrepancy is zero; otherwise, FP minimizes the discrepancy within the neighborhood. If $\widehat\varrho_{n,\esttype}=0$, FP equals BC. As with Raw ($\widehat{\mathrm d}_{\esttype}^{(0)}$) and BC, the benchmark choice $\mathcal I_{\mathrm O}=[-.25,.75]$ permits negative estimates under short memory. Numerically, we approximate the minimization on a grid with mesh $h_n=\min(.0025,n^{-1/2})$ within $\mathcal N_{n,\esttype}$, followed by a local quadratic refinement constrained to the same neighborhood. If the neighborhood contains no grid point, we retain BC.

\begin{theorem}[Consistency of the localized refinement]\label{thm_fp}
Suppose that the conditions of Corollary~\ref{cor_update} hold. Then $\widehat\varrho_{n,\esttype}=o_p(1)$ and $\widehat{\mathrm d}_{\esttype,\mathrm{FP}}\to_p\mathrm d$ for $\esttype\in\{\mathrm R,\mathrm S\}$.
\end{theorem}
The result follows because BC is consistent and every point in $\mathcal N_{n,\esttype}$ lies within $\widehat\varrho_{n,\esttype}=o_p(1)$ of BC. Thus, the same consistency argument applies to any measurable selection from this neighborhood, including the numerical refinement above. Under short memory, Assumption~\ref{ass_bd} at both updates makes both corrections zero with probability tending to one. On this event, the first and second updates equal Raw, so $\widehat\varrho_{n,\esttype}=0$ and Raw, BC, and FP coincide.
\begin{remark}[Why the last-update neighborhood?]\label{rem_fp}
A radius based on the full Raw-to-BC displacement makes Raw itself an admissible candidate, allowing the refinement to reverse the entire correction. A fixed penalty around BC introduces an additional relative weight between self-consistency and localization. By contrast, $\widehat\varrho_{n,\esttype}$ uses the most recent change in the same update scheme, introduces no additional tuning parameter, and collapses when the second update no longer changes the estimate. We thus use the last-update neighborhood throughout.
\end{remark}

\subsubsection{Finite-sample bias diagnostics}\label{sec_bias}
We use deterministic slope diagnostics to distinguish the effects of common centering, inclusion of lag zero, and finite-bandwidth Bartlett curvature. These diagnostics clarify the role of the block correction and help interpret the possible transition from underestimation to overestimation observed in Section~\ref{sec_sim}.
\par\noindent\textbf{Common centering.}
By Lemma~\ref{lem_raw_bias}, for $r\in\{m,\lfloor qm\rfloor\}$, the common centering bias shifts the expected aggregate by
\begin{equation}\label{eq_delta}
\Delta_n(r)=\mathbb E\widehat{\mathsf B}_{C,n}(r;m)-\mathsf B_{C,n}(r;m)=\{2a_{n,m}+r-1\}b_n.
\end{equation}
Under long memory and the maintained bandwidth conditions, the population aggregates and the expected sample aggregates are positive for all sufficiently large $n$. Define their corresponding two-bandwidth slopes by
\begin{equation}\label{eq_bias_slope}
\mathrm d_n^{\mathrm{det}}(q,m)=\frac{\log\{\mathsf B_{C,n}(\lfloor qm\rfloor;m)/\mathsf B_{C,n}(m;m)\}}{2\log\{\lfloor qm\rfloor/m\}},\,\,
\mathrm d_n^{\mathrm{mean}}(q,m)=\frac{\log\{\mathbb E\widehat{\mathsf B}_{C,n}(\lfloor qm\rfloor;m)/\mathbb E\widehat{\mathsf B}_{C,n}(m;m)\}}{2\log\{\lfloor qm\rfloor/m\}}.
\end{equation}
Their difference isolates the common centering effect, while the population slope itself may still exhibit finite-bandwidth distortion. The mean-aggregate slope need not equal $\mathbb E\widehat{\mathrm d}_{\mathrm R}(q,m)$, since the sample estimator involves nonlinear transformations and output projection.


Under long memory, $2a_{n,m}+r-1\sim r$ and $\mathsf B_{C,n}(r;m)\sim\mathsf B_C(r)$. Thus, \eqref{eq_dhat_2}, \eqref{eq_bartlett_growth}, and \eqref{eq_delta} give, for $r\in\{m,\lfloor qm\rfloor\}$,
\begin{equation}\label{eq_bias_rel}
\frac{\Delta_n(r)}{\mathsf B_{C,n}(r;m)}\sim\frac{rb_n}{\mathsf B_C(r)}\sim-\left(\frac{r}{n}\right)^{1-2\mathrm d}\frac{L_C(n)}{L_C(r)}.
\end{equation}
The common coefficient $\const_C/\{\mathrm d(2\mathrm d+1)\}$ cancels in the second equivalence, giving $-1$.

The relative perturbations in \eqref{eq_bias_rel} converge to zero. Applying $\log(1+x)=x+o(x)$ to \eqref{eq_bias_slope}, together with $\lfloor qm\rfloor/m\to q$ and $L_C(\lfloor qm\rfloor)/L_C(m)\to1$, therefore gives
\begin{equation}\label{eq_bias_center}
\mathrm d_n^{\mathrm{mean}}(q,m)-\mathrm d_n^{\mathrm{det}}(q,m)\sim-\frac{q^{1-2\mathrm d}-1}{2\log q}\left(\frac{m}{n}\right)^{1-2\mathrm d}\frac{L_C(n)}{L_C(m)}<0.
\end{equation}
Thus, common centering lowers the mean-aggregate slope relative to its population counterpart: the relative negative perturbation at $\lfloor qm\rfloor$ is asymptotically $q^{1-2\mathrm d}$ times that at $m$. This effect vanishes as $m/n\to0$, but can decay slowly when $\mathrm d$ is close to $1/2$.

\par\noindent\textbf{Retaining lag zero.}
With the stabilizer omitted, let $\mathrm d_n^{\mathrm{punct}}(q,m)$ and $\mathrm d_n^{\mathrm{full}}(q,m)$ denote the two-bandwidth slopes based on $\mathsf B_C(r)$ and $D_\nu+\mathsf B_C(r)$, respectively. A first-order logarithmic expansion using \eqref{eq_bartlett_growth} gives
\begin{equation*}
\mathrm d_n^{\mathrm{full}}(q,m)-\mathrm d_n^{\mathrm{punct}}(q,m)\sim\frac{\mathrm d(2\mathrm d+1)D_\nu(q^{-2\mathrm d}-1)}{2\const_C\log q}\frac{m^{-2\mathrm d}}{L_C(m)}<0.
\end{equation*}
Retaining lag zero therefore flattens the finite-bandwidth slope without changing the long-memory exponent, supporting its omission in \eqref{eq_bartlett_1}. Under short memory, including lag zero changes the limit to $D_\nu+\mathsf B_\infty$, which may still vanish. The vanishing stabilizer contributes only $O\{a_{n,m}m^{-2\mathrm d}/L_C(m)\}=o\{m^{-2\mathrm d}/L_C(m)\}$ to the long-memory slope.
\par\noindent\textbf{Finite-bandwidth Bartlett curvature.}
The third component is the finite-bandwidth distortion of the population slope itself, even when $D_\nu$ is known. The following result quantifies this effect under an exact power-law benchmark used only for bias diagnostics.
\begin{proposition}[Bartlett curvature under an exact power law]\label{prop_bartlett}
Suppose that $C(k)=\const_C k^{2\mathrm d-1}$ for every $k\geq1$, where $\const_C>0$ and $\mathrm d\in(0,1/2)$. Then
\begin{equation}\label{eq_curve_1}
\mathsf B_C(r)=\frac{\const_C}{\mathrm d(2\mathrm d+1)}r^{2\mathrm d}+2\const_C\zeta(1-2\mathrm d)+O(r^{-1}),
\end{equation}
where $\zeta$ is the analytically continued Riemann zeta function, with $\zeta(1-2\mathrm d)<0$. Consequently, the population slope in \eqref{eq_bias_slope} satisfies
\begin{equation}\label{eq_curve_2}
\mathrm d_n^{\mathrm{det}}(q,m)-\mathrm d\sim\frac{\mathrm d(2\mathrm d+1)\zeta(1-2\mathrm d)(q^{-2\mathrm d}-1)}{\log q}\,m^{-2\mathrm d}>0.
\end{equation}
The stabilizer contributes $O\{a_{n,m}m^{-2\mathrm d}\}=o(m^{-2\mathrm d})$, leaving the leading term unchanged.
\end{proposition}
\par\noindent\textbf{Comparison of decay rates.}
For the proposed punctured construction, the relevant comparison is between the negative common-centering effect and the positive Bartlett curvature under the power-law benchmark. Reducing the former can reveal the latter. Let $m\sim c_mn^\kappa$, where $c_m>0$ and $0<\kappa<1$. The magnitudes of the negative term in \eqref{eq_bias_center} and the positive term in \eqref{eq_curve_2} have orders $n^{-(1-\kappa)(1-2\mathrm d)}$ and $n^{-2\mathrm d\kappa}$, respectively. Their decay rates agree at
\begin{equation*}
\mathrm d^\star=(1-\kappa)/2.
\end{equation*}
For the cube-root bandwidth, $\mathrm d^\star=1/3$. Below this value, the negative component decays faster and can expose the positive curvature; above it, its decay is slower. These deterministic diagnostics do not imply a universal crossing or monotonicity of the sample estimator's bias.

\begin{remark}[Interpretation of the block adjustment]
The diagnostics above concern population and mean-aggregate slopes. Likewise, Proposition~\ref{prop_block} establishes bias tracking for the population block regression. The bound in Lemma~\ref{lem_block_rate} does not by itself guarantee $\widehat b_n/b_n\to_p1$. The consistency results control the adjustment through Assumption~\ref{ass_bd} and do not require this relative consistency. The procedure targets the leading common-centering bias without guaranteeing exact unbiasedness.
\end{remark}

\subsubsection{Primitive sufficient conditions}\label{sec_conditions}
We provide sufficient conditions for Assumptions~\ref{ass_raw} and~\ref{ass_bd} through stochastic bounds for the distance averages. For $k\geq1$, let $Z_{t,k}=\dist(X_t,X_{t-k})$ and $\gamma_{Z,k}(j)=\operatorname{Cov}(Z_{t,k},Z_{t-j,k})$.
\begin{assumption}[Distance-transformed temporal dependence]\label{ass_dt}
The variables $Z_{t,k}$ are square-integrable for every $k\geq1$. There exist a constant $\bar c>0$ and a positive slowly varying function $L_C^\star$ such that $\sup_{k\geq1}\sum_{|j|<n}|\gamma_{Z,k}(j)|\leq\bar c n^{2\mathrm d}L_C^\star(n)$, $n\geq1$, where $\mathrm d=0$ under short memory.
\end{assumption}
The function $L_C^\star$ controls the covariance bound for the distance transforms and need not coincide with $L_C$.
\begin{lemma}[Stochastic rates for the distance averages]\label{lem_dt}
Let $m\to\infty$, $m/n\to0$, and fix $q>1$. Under Assumption~\ref{ass_dt}, $(\lfloor qm\rfloor)^{-1}\sum_{k=1}^{\lfloor qm\rfloor}|\widehat\delta_n(k)-\delta(k)|=O_p\{n^{\mathrm d-1/2}\sqrt{L_C^\star(n)}\}$ and $|\widehat D_n-\mathbb E\widehat D_n|=O_p\{n^{\mathrm d-1/2}\sqrt{L_C^\star(n)}\}$.
\end{lemma}
The first bound controls the average lagwise error entering the Bartlett aggregate, while the second controls the stochastic fluctuation of the global distance average. The next lemma extends these bounds to the block averages and the bias regression.
\begin{lemma}[Stochastic rates for $\widehat D_{n,s}$ and $\widehat b_n$]\label{lem_block_rate}
Under Assumption~\ref{ass_dt}, for the fixed block set $\mathcal S$,
\begin{equation}\label{eq_block_rate_1}
\max_{s\in\mathcal S}|\widehat D_{n,s}-\mathbb E\widehat D_{n,s}|=O_p\{n^{\mathrm d-1/2}\sqrt{L_C^\star(n)}\}.
\end{equation}
If, in addition, \eqref{eq_trace} holds and $\widehat{\mathrm d}_{\mathrm P}\to_p\mathrm d$, then
\begin{equation}\label{eq_block_rate_2}
\widehat b_n=b_n+O_p\{n^{\mathrm d-1/2}\sqrt{L_C^\star(n)}\}+o_p(|b_n|).
\end{equation}
\end{lemma}
Independence across blocks is not required. In \eqref{eq_block_rate_2}, the remainder $o_p(|b_n|)$ accounts for the block-bias approximation and the use of an estimated pilot.

\begin{proposition}[Primitive sufficient conditions for the raw and adjusted estimators]\label{prop_rates}
Let $q>1$ be fixed and $m\to\infty$ with $m/n\to0$. Suppose that Assumption~\ref{ass_dt} holds and $\mathrm d\in\mathcal I_{\mathrm O}$, where $\mathrm d=0$ under short memory. Write $\mathfrak r_n=n^{\mathrm d-1/2}\sqrt{L_C^\star(n)}$.
\begin{enumerate}[label=(\roman*)]
\item Under the long-memory condition \eqref{eq_trace}, Assumption~\ref{ass_raw} holds if
\begin{equation}\label{eq_rate_lm}
\mathfrak r_n=o\{m^{2\mathrm d-1}L_C(m)\}.
\end{equation}
Under the same condition, the correction in \eqref{eq_corr} also satisfies Assumption~\ref{ass_bd}.
\item Under short memory, suppose that $\sum_{k\geq1}k|C(k)|<\infty$. If $\mathsf B_\infty\neq0$, Assumption~\ref{ass_raw} holds provided $m\mathfrak r_n+m/n=o(1)$. If $\mathsf B_\infty=0$ and $D_\nu>0$, it holds provided
\begin{equation}\label{eq_rate_sm}
m\mathfrak r_n+m/n=o(a_{n,m}),\qquad m^{-1}=o(a_{n,m}).
\end{equation}
In either case, Theorem~\ref{thm_raw} gives $\widehat{\mathrm d}_{\mathrm P}\to_p0$ for the initial projected raw pilot. The positive lower activation threshold then implies $\Pr(\widehat c_n=0)\to1$, establishing Assumption~\ref{ass_bd}. The same argument applies recursively to each fixed pilot update.
\end{enumerate}
\end{proposition}
For the long-memory adjustment, the fixed interval $[\underline{\mathrm d},\overline{\mathrm d}]\subset(0,1/2)$ keeps the block-regression denominator uniformly bounded away from zero for sufficiently large $n$. Thus, Lemma~\ref{lem_block_bias} and \eqref{eq_block_rate_1} give $\widehat c_n=O_p(\mathfrak r_n+|b_n|)$ for any projected pilot. Since \eqref{eq_dhat_2} and $m/n\to0$ imply $|b_n|=o\{m^{2\mathrm d-1}L_C(m)\}$, condition \eqref{eq_rate_lm} also controls the adjustment.
\begin{corollary}[Simple bandwidth and stabilization orders]\label{cor_orders}
Let $m=c_mn^\kappa\{1+o(1)\}$ with $c_m>0$ and $0<\kappa<1/2$. Then \eqref{eq_rate_lm} holds for every fixed $\mathrm d\in(0,1/2)$. Under short memory, $m\mathfrak r_n+m/n=o(1)$. If $\mathsf B_\infty=0$, let $a_{n,m}=c_a(m/n)^\eta$ with $c_a>0$. Then \eqref{eq_rate_sm} holds whenever
\begin{equation}\label{eq_eta}
0<\eta<\min\left\{\frac{1/2-\kappa}{1-\kappa},\frac{\kappa}{1-\kappa}\right\}.
\end{equation}
For $m$ proportional to $n^{1/3}$, any $0<\eta<1/4$ is admissible.
\end{corollary}
The first upper bound in \eqref{eq_eta} ensures $m\mathfrak r_n=o(a_{n,m})$, while the second ensures $m^{-1}=o(a_{n,m})$; $m/n=o(a_{n,m})$ then follows automatically. The strict inequalities ensure that these conclusions hold with the slowly varying factors included.

Assumption~\ref{ass_dt} is an estimation condition, not part of the definition of memory. The following lemma provides a sufficient condition through covariance bounds for pairs of distances.
\begin{lemma}[Shifted covariance envelope]\label{lem_shift}
Suppose that the pairwise distances are square-integrable and that there exist a nonnegative sequence $\{\varpi(h)\}_{h\geq0}$, a constant $0<C<\infty$, and a positive slowly varying function $L_C^\star$ such that 
\begin{equation}\label{eq_shift}
\sup_{s\in\mathbb Z}\sum_{|j|<n}\varpi(|j-s|)\leq Cn^{2\mathrm d}L_C^\star(n),\qquad n\geq1,
\end{equation}
and, for all $a_1,a_2,b_1,b_2\in\mathbb Z$,
\begin{equation}\label{eq_pair}
|\operatorname{Cov}\{\dist(X_{a_1},X_{a_2}),\dist(X_{b_1},X_{b_2})\}|\leq C\sum_{u=1}^2\sum_{v=1}^2\varpi(|a_u-b_v|).
\end{equation}
Then Assumption~\ref{ass_dt} holds.
\end{lemma}
For $Z_{t,k}$ and $Z_{t-j,k}$, the four cross-pair separations in \eqref{eq_pair} are $|j|$, $|j+k|$, $|j-k|$, and $|j|$. The supremum over shifts in \eqref{eq_shift} therefore makes the resulting covariance-sum bound uniform in $k$.
We next apply this criterion to the examples in Section~\ref{sect::memory def}.
\begin{proposition}[Verification for the examples]\label{prop::example assumption}
For the processes and metrics in Examples~\ref{ex_gaussian}--\ref{ex_wass}, Assumption~\ref{ass_dt} holds.
\end{proposition}
\begin{remark}[Gaussian-driven sufficient condition]\label{rem_dt}
The same approach applies more generally to Gaussian-driven object-valued processes; Section~\ref{app_dt} of the supplementary material gives the full argument, summarized here. Suppose that $X_t=G(\xi_t)$, where $\{\xi_t\}$ is a stationary Gaussian vector process, and put $\mathcal R_\xi(h)=\|\operatorname{Cov}(\xi_t,\xi_{t-h})\|_{\max}$, the largest absolute matrix entry. Under the uniform conditions on nondegeneracy, second moments, and Hermite rank stated there, \citet[Lemma~1]{arcones1994limit} and Cauchy--Schwarz give $|\gamma_{Z,k}(j)|\leq C\{\mathcal R_\xi(|j|)^\tau+\mathcal R_\xi(|j-k|)^\tau+\mathcal R_\xi(|j+k|)^\tau\}$ uniformly over $k\geq1$ and $j\in\mathbb Z$, where $\tau\geq1$ lower-bounds the Hermite ranks of the centered distance transforms after Gaussian standardization. Consequently, the shifted-sum argument in Lemma~\ref{lem_shift}  verifies Assumption~\ref{ass_dt} if $\varpi(h)=\mathcal R_\xi(h)^\tau$ satisfies \eqref{eq_shift}.
\end{remark}

\section{Simulation study}\label{sec_sim}
We use Examples~\ref{ex_gaussian}--\ref{ex_wass} to assess finite-sample performance across real-valued, matrix-valued, and distributional series. Their trace-memory parameters are characterized in Proposition~\ref{prop::example memory parameter}. They also satisfy the primitive condition underlying our consistency theory (Assumption~\ref{ass_dt}); see Proposition~\ref{prop::example assumption} and the Gaussian-driven verification in Section~\ref{app_dt} of the supplement.
\subsection{Design and implementation}
In each replication, we generate a stationary Gaussian autoregressive fractionally integrated moving average (ARFIMA) process $\{G_t\}$ of order $(1,\mathrm d,1)$, drawing the AR and MA coefficients $\varphi$ and $\theta_{\mathrm{MA}}$ independently from $\operatorname{Unif}(-.25,.25)$. We compute its autocovariances by inverse Fourier transformation of the spectrum, normalize them by the lag-zero value, and simulate $\{G_t\}$ with $\operatorname{Var}(G_t)=1$ using circulant embedding. At zero frequency, the short-run ARMA multiplier $|1+\theta_{\mathrm{MA}}|^2/|1-\varphi|^2$ lies between $(.75/1.25)^2$ and $(1.25/.75)^2$, ruling out near cancellation of the memory component. We use 1,000 Monte Carlo replications with $n\in\{250,500,1000,1500,2000\}$ and $\mathrm d\in\{0,.1,.2,.3,.4\}$. For each replication and value of $\mathrm d$, we generate one path of length $2000$ and use its first $n$ observations for each sample size. We consider four settings: (i) the real-valued series in Example~\ref{ex_gaussian}, with $X_t=G_t$; (ii) the $3\times 3$ SPD matrix-valued series in Example~\ref{ex_psd}, with $M_0=I_3$ (the identity matrix), $u=(1,0,0)^\top$, $v=(0,1,0)^\top$, and $\vartheta=1/4$; (iii) the distributional location variation setting in Example~\ref{ex_wass}, with $\mu_t=G_t$; and (iv) the distributional location and scale variation setting in Example~\ref{ex_wass}, with $\mu_t=G_t$, baseline scale $\sigma=1$, and scale loading $\lambda=1/4$. Table~\ref{table::setting summary} lists the designs and abbreviations.

\begin{table}[!htb]
\centering
\renewcommand{\arraystretch}{1} 
\caption{\small Summary of the object-valued time series designs in simulation.}\label{table::setting summary}
\scalebox{0.8}{%
\begin{tabular}{@{}ll@{}}
\toprule
Design & Abbreviation \\
\midrule
Real-valued time series & Real \\
Matrix-valued time series & Matrix \\
Distributional time series with varying locations & Dist--Location \\
Distributional time series with varying locations and scales & Dist--Location\&Scale \\
\bottomrule
\end{tabular}
}
\end{table}
As shown in Proposition~\ref{prop::example memory parameter}, the driver and all four settings have the same trace-memory parameter $\mathrm d$ for $\mathrm d>0$. At $\mathrm d=0$, the randomized driver is a short-memory ARMA$(1,1)$ process. The boundary design therefore retains nontrivial short-run serial dependence rather than reducing the comparison to an independent sequence. 
The baseline tuning choices are $m=\max\{3,\lfloor n^{1/3}\rceil\}$, $q=2$, $a_{n,m}=\frac13(m/n)^{1/8}$, and $\mathcal S=\{1,2,4,8,16\}$, where $\lfloor x\rceil$ denotes nearest-integer rounding. Corollary~\ref{cor_orders} permits $0<\eta<1/4$ for $m\asymp n^{1/3}$, so $\eta=1/8$ is the midpoint of this admissible interval. We use $c_a=1/3$ as a finite-sample baseline without claiming optimality. Note that our analysis is not limited to this baseline: we also consider several prespecified tuning configurations and evaluate the resulting averaged estimators; see below.

All raw and corrected log-ratio and log-slope outputs, including those of the map used by FP, are projected onto $\mathcal I_{\mathrm O}=[-.25,.75]$. This common bounded interval contains the model range $[0,1/2)$ in its interior, so projection is asymptotically inactive for consistent unprojected estimators. For the block correction in \eqref{eq_corr}, each trial value is projected onto $\mathcal I_{\mathrm P}=[-.10,.45]$, and correction is applied only when the projected pilot lies in $(.05,.45]$. Including zero in the pilot interval preserves pilot consistency under short memory. The positive cutoff reflects the long-memory bias expansion underlying the correction, while the upper bound avoids degeneration of the block regressor near $1/2$; see Remark~\ref{rem_pilot_activation}. A trial value above $.45$ is projected to $.45$ and remains eligible for correction. The FP neighborhood $\mathcal N_{n,\esttype}$ is contained in the same output interval $\mathcal I_{\mathrm O}$, allowing negative FP estimates. These projections are part of the estimators used to compute Monte Carlo means and root mean squared errors (RMSEs).

For both the log-ratio and log-slope constructions, the main tables report the raw estimator Raw ($\widehat{\mathrm d}_{\esttype}^{(0)}$), the two-step bias-corrected estimator BC ($\widehat{\mathrm d}_{\esttype}^{(2)}$), and the localized minimum-discrepancy estimator FP ($\widehat{\mathrm d}_{\esttype,\mathrm{FP}}$). The label $\mathrm{BD}$ denotes a generic block-difference-adjusted scale. BC uses the first corrected estimate as its pilot, and FP minimizes the fixed-point discrepancy $\mathcal Q_{n,\esttype}$ over the last-update neighborhood of BC; see \eqref{eq_fp_est}.
\par
To address tuning uncertainty without selecting a specification based on the results, we also report equal-weight averages over the prespecified set $\mathfrak T=\{(1,2,1/8),(.8,2,1/8),(1.25,2,1/8),(1,1.5,1/8),$  $(1,2.5,1/8),(1,2,1/12),(1,2,1/6)\}$. Each triple $\mathbf t=(c_m,q,\eta)$ specifies $m=\lfloor c_mn^{1/3}\rceil$ and $a_{n,m}=\frac13(m/n)^\eta$. For $M\in\{\mathrm{Raw},\mathrm{BC},\mathrm{FP}\}$ and $\esttype\in\{\mathrm R,\mathrm S\}$, let $\widehat{\mathrm d}_{\esttype,M}(\mathbf t)$ denote the corresponding estimator evaluated under $\mathbf t$, and define $\widehat{\mathrm d}_{\esttype,M}^{\mathrm{avg}}=|\mathfrak T|^{-1}\sum_{\mathbf t\in\mathfrak T}\widehat{\mathrm d}_{\esttype,M}(\mathbf t)$. The set varies one tuning coordinate at a time around the baseline. Since $\mathfrak T$ is finite, consistency of each component estimator implies consistency of the average.
The tables report the baseline and tuning-average versions of Raw, BC, and FP.

\subsection{Monte Carlo results}

The simulation results assess the finite-sample bias of the estimators, the reduction achieved by BC and FP, and sensitivity to tuning choices.

Table~\ref{tab_mc_aggregate} compares the estimators over all 25 $(n,\mathrm d)$ cells in each design and shows substantial reductions in mean absolute bias (MAB) after correction. For the Real design, the baseline log-ratio estimator's MAB decreases from $0.0385$ for Raw to $0.0103$ for BC and $0.0080$ for FP, corresponding to reductions of approximately $73\%$ and $79\%$, respectively. For the Matrix design, MAB decreases from $0.0411$ for Raw to $0.0105$ for BC and $0.0077$ for FP, corresponding to reductions of approximately $74\%$ and $81\%$, respectively. Similar reductions occur in the two distributional designs, Dist--Location and Dist--Location\&Scale, and for the log-slope estimators. The log-ratio construction has slightly lower RMSEs than the log-slope construction, and FP has modestly higher RMSEs than BC. Across all four designs, correction substantially reduces bias with little change in aggregate RMSE.
\begin{table}[!htb]
\centering
\renewcommand{\arraystretch}{0.6} 
\caption{\small Aggregate Monte Carlo performance over the 25 $(n,\mathrm d)$ cells in each design, based on 1,000 replications per cell; MAB denotes mean absolute bias.}
\label{tab_mc_aggregate}
\small
\setstretch{1.0}
\setlength{\tabcolsep}{17pt}
\scalebox{0.8}{%
\begin{tabular}{@{}lllrrrr@{}}
\toprule
& & & \multicolumn{2}{c}{Baseline} & \multicolumn{2}{c}{Tuning average} \\
\cmidrule(lr){4-5}\cmidrule(lr){6-7}
Design & Slope & Method & MAB & RMSE & MAB & RMSE \\ 
\midrule
Real & Ratio & Raw & .0385 & .0912 & .0385 & .0894 \\
 &  & BC & .0103 & .0914 & .0102 & .0895 \\
 &  & FP & .0080 & .0940 & .0079 & .0921 \\
\addlinespace
 & Log slope & Raw & .0385 & .0922 & .0385 & .0900 \\
 &  & BC & .0102 & .0924 & .0101 & .0902 \\
 &  & FP & .0079 & .0950 & .0078 & .0927 \\
\midrule\addlinespace
Matrix & Ratio & Raw & .0411 & .0910 & .0410 & .0895 \\
 &  & BC & .0105 & .0897 & .0106 & .0882 \\
 &  & FP & .0077 & .0921 & .0077 & .0906 \\
\addlinespace
 & Log slope & Raw & .0411 & .0919 & .0411 & .0902 \\
 &  & BC & .0105 & .0907 & .0105 & .0889 \\
 &  & FP & .0077 & .0930 & .0077 & .0912 \\
\midrule\addlinespace
Dist--Location & Ratio & Raw & .0385 & .0912 & .0385 & .0894 \\
 &  & BC & .0103 & .0914 & .0102 & .0895 \\
 &  & FP & .0080 & .0940 & .0079 & .0921 \\
\addlinespace
 & Log slope & Raw & .0385 & .0922 & .0385 & .0900 \\
 &  & BC & .0102 & .0924 & .0101 & .0902 \\
 &  & FP & .0079 & .0950 & .0078 & .0927 \\
 \midrule\addlinespace
Dist--Location\&Scale & Ratio & Raw & .0387 & .0912 & .0386 & .0894 \\
 &  & BC & .0103 & .0913 & .0102 & .0894 \\
 &  & FP & .0080 & .0939 & .0079 & .0920 \\
\addlinespace
 & Log slope & Raw & .0387 & .0922 & .0386 & .0901 \\
 &  & BC & .0102 & .0923 & .0101 & .0901 \\
 &  & FP & .0079 & .0949 & .0078 & .0927 \\
\bottomrule
\end{tabular}
}
\end{table}

Averaging over the prespecified tuning set reduces aggregate RMSE by $0.0015$ to $0.0023$ across all designs and methods in Table~\ref{tab_mc_aggregate}, while MAB changes by at most $0.0001$ in absolute value. The improvement is therefore mainly in RMSE, with little additional bias reduction. 

We next report results for the log-ratio estimators under the Matrix design. The corresponding log-slope results are given in Section~\ref{supp subsect::matrix} of the supplementary material. Results for the Real and the two distributional designs are provided in Sections~\ref{supp subsect::real} and~\ref{supp subsect::distribution}, respectively.

Table~\ref{tab_mc_psd_ratio} summarizes the finite-sample performance of the log-ratio estimators under the Matrix design. RMSEs decrease as the sample size $n$ increases. Under long memory, Raw exhibits increasingly pronounced downward bias as $\mathrm d$ increases for fixed $n$. BC and FP substantially reduce this bias at moderate and high memory levels. At some memory levels they leave a small upward bias: once the dominant negative common-centering effect is removed, the positive finite-bandwidth curvature in \eqref{eq_curve_2} becomes visible in finite samples. Both effects vanish under the rate conditions of Proposition~\ref{prop_rates}; see Section~\ref{supp-bias-behavior} of the supplement.

\begin{table}[!htb]
\centering
\small
\renewcommand{\arraystretch}{0.6} 
\caption{\small Monte Carlo results for the Matrix design using the $\log$-ratio estimators. Entries are means with RMSEs in parentheses, based on 1,000 replications.}
\label{tab_mc_psd_ratio}
\fontsize{7.3pt}{9.4pt}\selectfont
\setstretch{1.0}
\setlength{\tabcolsep}{16pt}
\scalebox{0.82}{%
\begin{tabular}{@{}crrrrrr@{}}
\toprule
& \multicolumn{3}{c}{Baseline} & \multicolumn{3}{c}{Tuning average} \\
\cmidrule(lr){2-4}\cmidrule(lr){5-7}
$\mathrm d$ & Raw & BC & FP & Raw & BC & FP \\
\midrule
\multicolumn{7}{l}{\hspace{-.16in}{$n=250$}} \\
.00 & -.021 (.127) & -.015 (.135) & -.015 (.135) & -.021 (.123) & -.014 (.131) & -.014 (.131) \\
.10 & .065 (.123) & .085 (.137) & .087 (.139) & .065 (.121) & .085 (.135) & .087 (.137) \\
.20 & .139 (.121) & .179 (.134) & .185 (.140) & .138 (.120) & .179 (.133) & .184 (.138) \\
.30 & .212 (.126) & .280 (.125) & .291 (.132) & .211 (.126) & .279 (.124) & .290 (.131) \\
.40 & .278 (.146) & .373 (.113) & .389 (.115) & .276 (.146) & .372 (.113) & .388 (.115) \\
\addlinespace
\multicolumn{7}{l}{\hspace{-.16in}{$n=500$}} \\
.00 & -.013 (.112) & -.010 (.116) & -.010 (.116) & -.013 (.107) & -.010 (.111) & -.010 (.112) \\
.10 & .079 (.095) & .093 (.105) & .094 (.106) & .079 (.093) & .093 (.102) & .093 (.103) \\
.20 & .158 (.094) & .191 (.103) & .193 (.107) & .158 (.093) & .191 (.102) & .194 (.105) \\
.30 & .236 (.095) & .294 (.094) & .303 (.102) & .236 (.094) & .294 (.093) & .304 (.101) \\
.40 & .297 (.121) & .376 (.091) & .390 (.093) & .298 (.120) & .377 (.090) & .391 (.092) \\
\addlinespace
\multicolumn{7}{l}{\hspace{-.16in}{$n=1000$}} \\
.00 & -.013 (.099) & -.011 (.101) & -.011 (.101) & -.013 (.096) & -.011 (.098) & -.011 (.098) \\
.10 & .089 (.077) & .099 (.083) & .099 (.083) & .088 (.076) & .098 (.082) & .099 (.082) \\
.20 & .174 (.068) & .200 (.075) & .201 (.077) & .174 (.067) & .200 (.075) & .201 (.076) \\
.30 & .253 (.074) & .299 (.075) & .305 (.081) & .253 (.073) & .299 (.074) & .305 (.081) \\
.40 & .316 (.098) & .383 (.072) & .395 (.074) & .315 (.098) & .383 (.072) & .395 (.074) \\
\addlinespace
\multicolumn{7}{l}{\hspace{-.16in}{$n=1500$}} \\
.00 & -.010 (.088) & -.009 (.090) & -.009 (.090) & -.010 (.086) & -.008 (.087) & -.008 (.087) \\
.10 & .096 (.066) & .105 (.071) & .105 (.071) & .096 (.065) & .105 (.070) & .105 (.070) \\
.20 & .186 (.055) & .208 (.063) & .209 (.065) & .186 (.054) & .208 (.063) & .209 (.064) \\
.30 & .263 (.061) & .303 (.063) & .309 (.069) & .263 (.060) & .303 (.063) & .308 (.068) \\
.40 & .325 (.086) & .387 (.064) & .397 (.067) & .325 (.086) & .387 (.064) & .397 (.067) \\
\addlinespace
\multicolumn{7}{l}{\hspace{-.14in}{$n=2000$}} \\
.00 & -.010 (.087) & -.010 (.088) & -.010 (.088) & -.010 (.084) & -.009 (.085) & -.009 (.085) \\
.10 & .098 (.061) & .106 (.066) & .106 (.066) & .098 (.059) & .106 (.065) & .106 (.065) \\
.20 & .187 (.052) & .208 (.060) & .209 (.061) & .188 (.051) & .209 (.059) & .209 (.060) \\
.30 & .263 (.058) & .301 (.058) & .305 (.062) & .264 (.057) & .301 (.057) & .305 (.062) \\
.40 & .328 (.083) & .385 (.059) & .395 (.062) & .328 (.082) & .386 (.059) & .395 (.062) \\
\bottomrule
\end{tabular}
}
\end{table}

The $\mathrm d=0$ rows in Table~\ref{tab_mc_psd_ratio} show that the methods accommodate short memory without imposing nonnegative estimates. For the baseline log-ratio estimators, the means of Raw, BC, and FP are $-.021$, $-.015$, and $-.015$ at $n=250$, respectively, and approximately $-.010$ for all three at $n=2000$. BC and FP have nearly identical means and RMSEs. Their close agreement is consistent with the construction in \eqref{eq_fp_set}: when correction is inactive in both updates, the neighborhood collapses and FP equals BC.

\section{Empirical data analysis}\label{sect::real data}

We apply the proposed long memory estimation approach to two real datasets, the foreign exchange intraday log return distributional time series in Section~\ref{sect::real data return}, and the U.S. monthly energy generation compositional time series in Section~\ref{sect::real data energy}.

In addition to the full-sample estimates, we assess the empirical stability of the proposed estimators using an overlapping-subsample analysis. For a time series $\{X_t\}_{t=1}^n$, we first set the nominal subsample length to $n_0=\lfloor 0.7n\rfloor$. Because the Bartlett bandwidth is integer-valued, directly using $n_0$ may result in a different bandwidth from that used for the full-sample estimator, thereby introducing an additional finite-sample tuning effect. Thus we choose $n_{\rm sub}$ as the smallest integer $b\geq n_0$ such that the bandwidth computed from a sample of size $b$ equals the full-sample bandwidth. We then consider all $n-n_{\rm sub}+1$ overlapping windows $\{X_j,\ldots,X_{j+n_{\rm sub}-1}\}$, for $j=1,\ldots,n-n_{\rm sub}+1$, and re-estimate the memory parameter separately within each window. Apart from the bandwidth, which is matched to its full-sample value by construction, all sample-size-dependent quantities entering the estimation procedure are recalculated using the selected subsample size $n_{\rm sub}$.

For each estimator, we report the estimate obtained from the full time series together with the mean, standard deviation, and empirical 2.5\% and 97.5\% quantiles of the estimates across the overlapping subsamples, used as descriptive measures of finite-sample stability. We read these as descriptive measures of stability across subsamples rather than as standard errors or formal confidence intervals.

\subsection{Foreign exchange intraday return distributional time series}\label{sect::real data return}

Over recent decades, an extensive literature has focused on modeling financial asset return distributions; see, e.g., \cite{weigend2000predicting,hallam2014forecasting}. Compared with summaries based only on the first two moments, daily return distributions retain substantially richer information about the underlying market behavior, including changes in location, dispersion, skewness, and tail behavior. A distributional representation therefore provides a more comprehensive description of the temporal evolution of financial returns.

We use high-frequency foreign exchange data available from \url{https://www.kaggle.com/datasets/arashnic/stock-data-intraday-minute-bar/data}. Specifically, we consider 5-minute closing prices for the EUR/USD currency pair and compute the corresponding intraday log returns during 2018. To improve comparability across trading days and reduce the influence of shortened trading sessions or data outages, we apply a frequency-based filtering rule. A trading day is retained only if its number of available 5-minute observations is at least 50\% of the modal daily count. The intraday log returns within each retained trading day are then represented by their empirical distribution. This yields a distributional time series of length $n=259$, with each observation corresponding to the empirical distribution of intraday returns on one trading day. Figure~\ref{fig:combine}(a) visualizes the distributional time series, and indicates that the time series appears to be stationary with constant mean. We then apply the proposed long memory estimators to this time series.

\begin{figure}[htbp]
    \centering

    \includegraphics[height=0.30\textwidth]{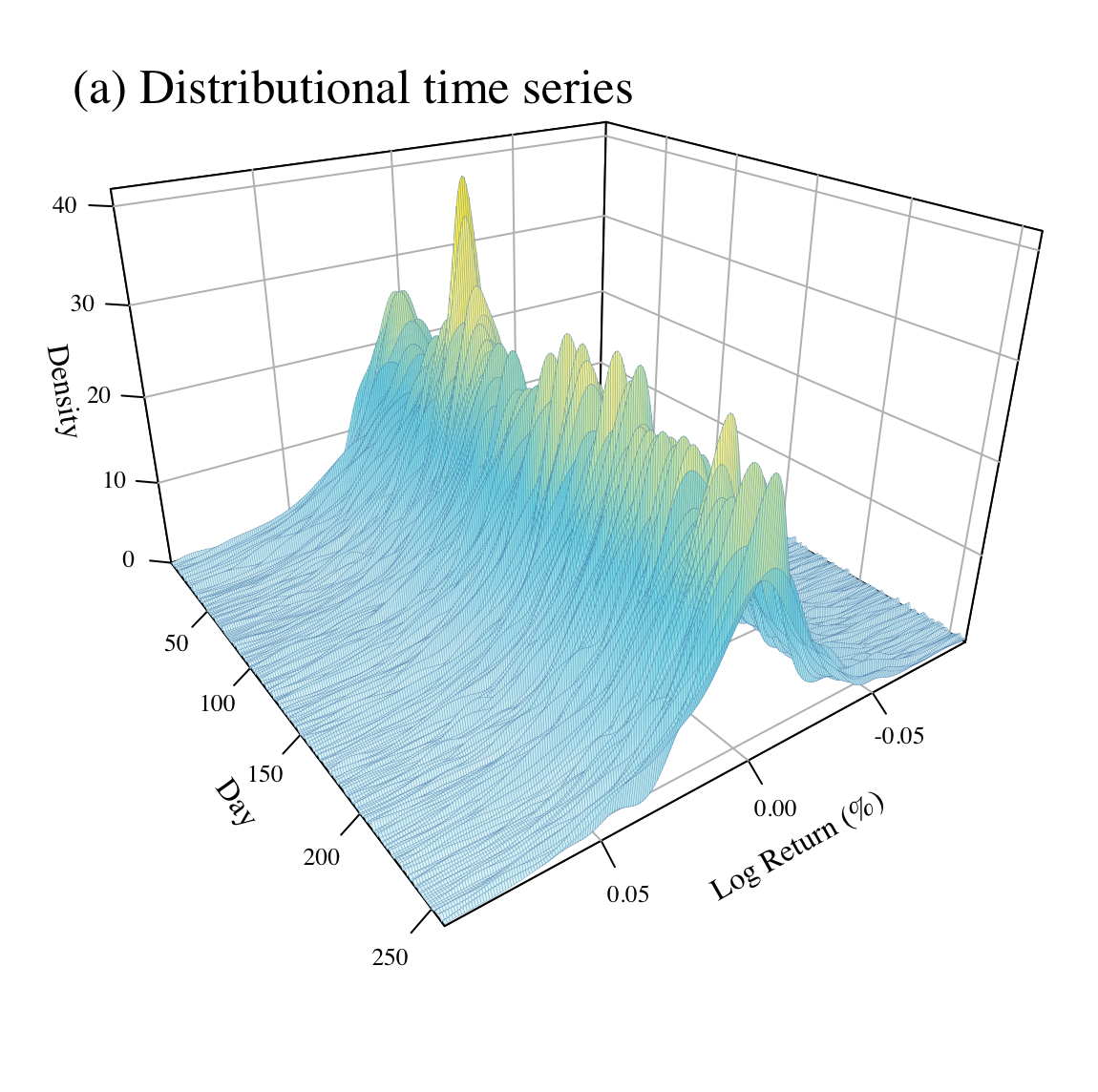}
    \hspace{0.02\textwidth}
    \includegraphics[height=0.30\textwidth]{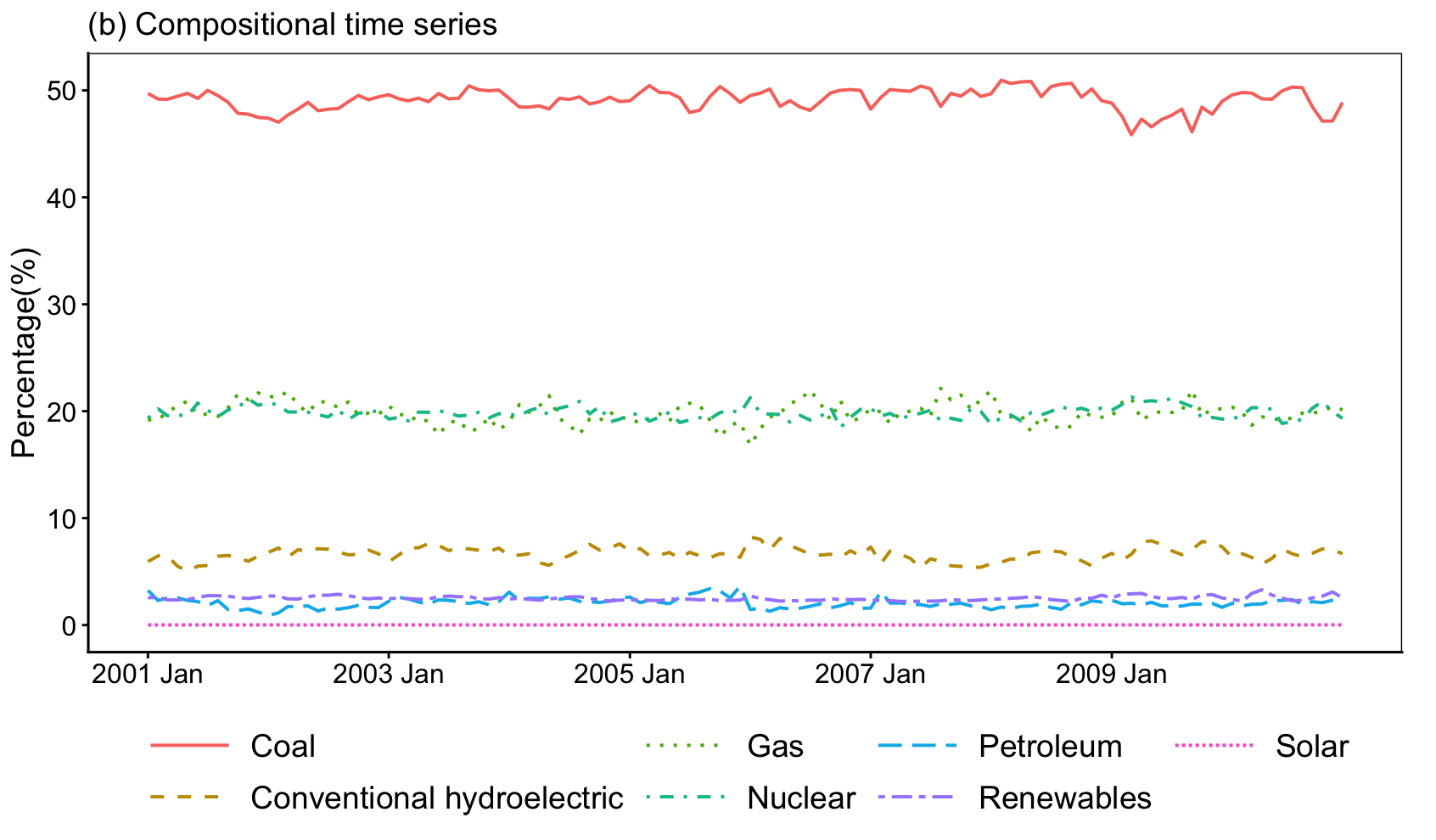}

    \caption{(a) Intraday $\log$ return distributional time series for the EUR/USD foreign exchange in~2018. (b) De-trended and de-seasonalized monthly U.S. electricity-generation compositions.}
    \label{fig:combine}
\end{figure}

For this application, the subsample length is $n_{\rm sub}=181$ which gives 79 overlapping subsamples. Table~\ref{tab:return_subsample} reports the full-sample memory-parameter estimates together with the corresponding subsample stability summaries.

Table~\ref{tab:return_subsample} shows similar results for the log-ratio and log-slope estimators, with full-sample estimates around $0.24$ for Raw and $0.31$--$0.33$ for BC and FP. The corresponding subsample means are around $0.21$ and $0.28$--$0.29$, respectively, and all 2.5\% subsample quantiles exceed zero. These descriptive stability results support a long-memory interpretation of the daily return distributional series.

\begin{table}[!htb]
\centering
\small
\renewcommand{\arraystretch}{0.7}
\caption{\small Memory-parameter estimates and overlapping-subsample stability measures for the intraday return distributional time series. Entries in the last four columns summarize the estimates obtained from the 79 overlapping subsamples of length $n_{\rm sub}=181$.}
\label{tab:return_subsample}
\fontsize{7.3pt}{9.4pt}\selectfont
\setstretch{1.0}
\setlength{\tabcolsep}{16pt}
\scalebox{0.93}{%
\begin{tabular}{@{}llrrrrr@{}}
\toprule
Method & Estimator & \shortstack{Full sample} & \shortstack{Subsample\\mean} & \shortstack{Subsample\\SD} & \shortstack{2.5\%\\quantile} & \shortstack{97.5\%\\quantile} \\
\midrule
Raw & Log-ratio & .244 & .215 & .045 & .139 & .292 \\
Raw & Log-slope & .241 & .211 & .046 & .133 & .290 \\
\addlinespace
BC & Log-ratio & .315 & .282 & .062 & .174 & .390 \\
BC & Log-slope & .312 & .278 & .063 & .166 & .388 \\
\addlinespace
FP & Log-ratio & .328 & .293 & .071 & .175 & .420 \\
FP & Log-slope & .324 & .289 & .072 & .167 & .418 \\
\bottomrule
\end{tabular}
}
\end{table}

\subsection{U.S. energy generation compositional time series}\label{sect::real data energy}

We next investigate the temporal dependence structure of monthly U.S. electricity generation, using raw data obtained from \url{https://www.eia.gov/electricity/data/browser/}. Following \cite{xu2025quantifying}, the original electricity-generation measurements for seven fuel types, such as coal, petroleum, gas, nuclear, conventional hydroelectric, renewables, and solar, are transformed into a sequence of compositional vectors, whose components represent the relative shares of different fuel types in total net electricity generation. The data consist of $n=120$ monthly observations from January 2001 through December 2010, with observations lying in the six-dimensional simplex $\Delta^6=\{(\delta_1,\ldots,\delta_7)^\top\in\mathbb R^7:\sum_{\ell=1}^7\delta_\ell=1,\;\delta_\ell\geq0\}$. For two compositions $x,y\in\Delta^6$, we use the Fisher--Rao distance $\dist(x,y)=\arccos\left(\sum_{i=1}^7\sqrt{x_i y_i}\right)$, where $x_i$ and $y_i$ denote the $i$th components of $x$ and $y$, respectively. 

Figure~\ref{fig::US energy} in Section~\ref{supp sect::additional real data} of the supplement displays the raw compositional data, with each line representing the percentage of the corresponding energy source. The visible trend and periodic components suggest that the series is nonstationary. Before estimating the memory parameter, we therefore apply the spherical trend-periodicity decomposition of \cite{xu2026spherically} to remove the smooth trend and seasonal structure. The resulting de-trended and de-seasonalized residual process, shown in Figure~\ref{fig:combine}(b), appears approximately stationary.

We apply the memory-parameter estimators to this residual process. For this application, the subsample length is 92, giving 29 overlapping subsamples. Table~\ref{tab:energy_subsample} reports the full-sample estimates and the subsample stability summaries. A pattern similar to that in Table~\ref{tab:return_subsample} appears in Table~\ref{tab:energy_subsample}, suggesting that the time series has long memory. Notably, subsample variation is smaller than in Table~\ref{tab:return_subsample}, indicating that the estimated memory parameter is more stable across different portions of the 2001--2010 observation period.
\begin{table}[!htb]
\centering
\small
\renewcommand{\arraystretch}{0.7}
\caption{\small Memory-parameter estimates and overlapping-subsample stability measures for the U.S. electricity-generation compositional time series. Entries in the last four columns summarize the estimates obtained from the 29 overlapping subsamples of length $n_{\rm sub}=92$.}
\label{tab:energy_subsample}
\fontsize{7.3pt}{9.4pt}\selectfont
\setstretch{1.0}
\setlength{\tabcolsep}{16pt}
\scalebox{0.92}{%
\begin{tabular}{@{}llrrrrr@{}}
\toprule
Method & Estimator & \shortstack{Full sample} & \shortstack{Subsample\\mean} & \shortstack{Subsample\\SD} & \shortstack{2.5\%\\quantile} & \shortstack{97.5\%\\quantile} \\
\midrule
Raw & Log-ratio & .206 & .196 & .030 & .160 & .239 \\
Raw & Log-slope & .208 & .197 & .029 & .161 & .239 \\
\addlinespace
BC & Log-ratio & .318 & .329 & .037 & .280 & .383 \\
BC & Log-slope & .320 & .330 & .036 & .281 & .383 \\
\addlinespace
FP & Log-ratio & .350 & .362 & .035 & .310 & .401 \\
FP & Log-slope & .350 & .363 & .035 & .311 & .402 \\
\bottomrule
\end{tabular}
}
\end{table}

\section{Conclusion}\label{sec:conclu}

We establish a foundational framework for defining and estimating long-range dependence in non-Euclidean object-valued time series. We develop log-ratio and multi-bandwidth log-slope estimators via a novel iterated block-difference adjustment and a localized self-consistency refinement. Extensive numerical studies confirm that our bias-corrected estimators substantially reduce finite-sample distortions while remaining robust to short-memory processes. Applications to foreign exchange intraday log returns and U.S. energy generation compositions illustrate the long memory phenomena in distributional time series and compositional time series.

There are several ways in which the current methodology may be extended, and we mention two: (i) Develop a hypothesis test for examining whether an object-valued time series exhibits long memory. (ii) Distinguish genuine long memory from persistence induced by structural breaks \citep[see, e.g.,][]{BKX2024}. Level shifts are known to generate sample autocovariances resembling those of a long-memory process, and the isometric embedding of Section~\ref{sect::memory def} offers a route to formalizing a level shift for object-valued data, as a change in the Fr\'{e}chet mean of the sequence.

\newpage 
\appendix 

\section*{Supplementary Material}

\section{Proofs}\label{supp_sec_proofs}
This section provides proofs of the main results in Sections~\ref{sect::memory def} and \ref{sect::estimation} of the main paper. We first record the following standard discrete form of Karamata's theorem, which is used repeatedly below.
\begin{lemma}[Karamata summation for regularly varying sequences]\label{lem_sum}
Let $\beta\in\mathbb R$ and let $\{a_k\}_{k\geq1}$ satisfy $a_k\sim c k^\rho L(k)$, where $c\neq0$, $L$ is positive and slowly varying, and $\beta+\rho>-1$. Then $\sum_{k=1}^n k^\beta a_k\sim c(\beta+\rho+1)^{-1}n^{\beta+\rho+1}L(n)$.
\end{lemma}
\begin{proof}[Proof of Lemma~\ref{lem_sum}]
Apply the discrete Karamata theorem to the regularly varying sequence $k^\beta a_k$, whose index is $\beta+\rho>-1$; see \citet[Theorem~1.5.11]{bingham1987regular}.
\end{proof}

\subsection{Proof of results in Section~\ref{sect::memory def}}

\begin{proof}[\textbf{Proof of Lemma~\ref{lem_embed}}]
By Schoenberg's embedding theorem \citep{schoenberg1937certain,schoenberg1938metric}, since $(\Omega,\dist)$ is of negative type, there exist a Hilbert space $\mathcal H$ and an isometric embedding $\phi:(\Omega,\dist^{1/2})\to\mathcal H$ such that
\begin{equation*}
    \dist(x,x')=\|\phi(x)-\phi(x')\|_{\mathcal H}^2;
\end{equation*}
see \citep[p.~3292]{lyons2013distance}. Proposition~3.5 of \citet{lyons2013distance} gives the following representation of the $\nu$-centered distance function $\dist_\nu=-K_\nu$:
\begin{equation*}
    \dist_\nu(x,x')=-2\left\langle \phi(x)-m_\phi(\nu),\phi(x')-m_\phi(\nu)\right\rangle_{\mathcal H}.
\end{equation*}
Hence 
\[
K_\nu(x,x')=2\langle\phi(x)-m_\phi(\nu),\phi(x')-m_\phi(\nu)\rangle_{\mathcal H}
\]
follows. In particular, for arbitrary $x_1,\ldots,x_m\in\Omega$ and $r_1,\ldots,r_m\in\mathbb R$,
\begin{align*}
    \sum_{i=1}^m\sum_{j=1}^m r_i r_j K_\nu(x_i,x_j)
    &=\sum_{i=1}^m\sum_{j=1}^m r_i r_j\langle\Phi_\nu(x_i),\Phi_\nu(x_j)\rangle_{\mathcal H}\\
    &=\left\langle \sum_{i=1}^m r_i\Phi_\nu(x_i),\sum_{j=1}^m r_j\Phi_\nu(x_j)\right\rangle_{\mathcal H}=\left\|\sum_{i=1}^m r_i\Phi_\nu(x_i)\right\|_{\mathcal H}^2\geq0.
\end{align*}
Thus $K_\nu$ is positive semidefinite. The representation that $K_\nu(x,x')=\langle\Phi_\nu(x),\Phi_\nu(x')\rangle_{\mathcal H}$ follows directly from the definition of $\Phi_\nu$.
\end{proof}

\begin{proof}[\textbf{Proof of Proposition~\ref{prop_cv}}]
We first verify that $\Gamma_k$ is well defined. For $h_1,h_2\in\mathcal H$, recall that the rank-one operator $h_1\otimes h_2$ is given by $(h_1\otimes h_2)(h_3)=\langle h_1,h_3\rangle_{\mathcal H}h_2$ for $h_3\in\mathcal H$. By the properties of the trace norm $\|\cdot\|_1$, $\mathbb E\|Y_t\otimes Y_{t-k}\|_1\leq\mathbb E\{\|Y_t\|_{\mathcal H}\|Y_{t-k}\|_{\mathcal H}\}\leq(\mathbb E\|Y_t\|_{\mathcal H}^2)^{1/2}(\mathbb E\|Y_{t-k}\|_{\mathcal H}^2)^{1/2}<\infty$, and hence $\Gamma_k=\mathbb E(Y_t\otimes Y_{t-k})$ is a trace-class operator. For a rank-one operator,
\begin{equation*}
 \operatorname{tr}(h_1\otimes h_2)=\sum_{\ell\ge1}\langle (h_1\otimes h_2)e_\ell,e_\ell\rangle_{\mathcal H}=\sum_{\ell\ge1}\langle h_1,e_\ell\rangle_{\mathcal H}\langle h_2,e_\ell\rangle_{\mathcal H}=\langle h_1,h_2\rangle_{\mathcal H},
\end{equation*}
where $\{e_\ell\}_{\ell\ge1}$ is any orthonormal basis of $\mathcal H$. Therefore,
\begin{align*}
    \operatorname{tr}(\Gamma_k)=\operatorname{tr}\{\mathbb E(Y_t\otimes Y_{t-k})\}=\mathbb E\{\operatorname{tr}(Y_t\otimes Y_{t-k})\}=\mathbb E\langle Y_{t-k},Y_t\rangle_{\mathcal H}.
\end{align*}
Note that $K_\nu(x,x')=-\dist_\nu(x,x')=\langle \Phi_\nu(x),\Phi_\nu(x')\rangle_{\mathcal H}$. Putting $x=X_t$ and $x'=X_{t-k}$ gives
\begin{equation*}
    \langle Y_t,Y_{t-k}\rangle_{\mathcal H}=-\dist_\nu(X_t,X_{t-k}).
\end{equation*}
Therefore,
\begin{equation*}
 -\mathbb E \dist_\nu(X_t,X_{t-k}) = \mathbb E\langle Y_t,Y_{t-k}\rangle_{\mathcal H} = \operatorname{tr}(\Gamma_k). 
\end{equation*} Writing $\dist_\nu^{(1)}(x)=\int_\Omega \dist(x,u)\,d\nu(u)$, the definition of $\dist_\nu$ gives
\begin{align*}
    \mathbb E \dist_\nu(X_t,X_{t-k})
    &=\mathbb E \dist(X_t,X_{t-k})-\mathbb E \dist_\nu^{(1)}(X_t)-\mathbb E \dist_\nu^{(1)}(X_{t-k})+D_\nu.
\end{align*}
Since $X_t$ and $X_{t-k}$ both have marginal distribution $\nu$ under stationarity, $\mathbb E \dist_\nu^{(1)}(X_t)=\mathbb E \dist_\nu^{(1)}(X_{t-k})=D_\nu$. Hence
\begin{equation*}
    \mathbb E \dist_\nu(X_t,X_{t-k})=\mathbb E \dist(X_t,X_{t-k})-D_\nu,
\end{equation*} 
and consequently
\begin{equation*}
   \operatorname{tr}(\Gamma_k) =D_\nu-\mathbb E \dist(X_t,X_{t-k}),
\end{equation*}
which proves \eqref{eq_c}. This completes the proof.
\end{proof}

\begin{proof}[\textbf{Proof of Proposition~\ref{prop_common}}]
Put $a_k=k^{2\mathrm d_0-1}L(k)$. The trace-norm assumption implies $\|R_k\|_{\mathcal S_2}=o(|a_k|)$. Hence
\begin{equation*}
    \|\Gamma_k\|_{\mathcal S_2}=|a_k|\left\|A+\frac{R_k}{a_k}\right\|_{\mathcal S_2}=|a_k|\{\|A\|_{\mathcal S_2}+o(1)\}.
\end{equation*}
Moreover, $C(k)=\operatorname{tr}(\Gamma_k)=a_k\operatorname{tr}(A)+\operatorname{tr}(R_k)$ and $|\operatorname{tr}(R_k)|\leq\|R_k\|_1=o(|a_k|)$. Since $\operatorname{tr}(A)\neq0$, this gives $C(k)\sim a_k\operatorname{tr}(A)$. Remark~\ref{rem_trace} implies $\operatorname{tr}(A)>0$, so Definition~\ref{def_trace} gives $\mathrm d=\mathrm d_0$.
\end{proof}

\begin{proof}[\textbf{Proof of Proposition~\ref{prop::example memory parameter}}]
    We first consider Example~\ref{ex_gaussian} where $\Omega=\mathbb R$, $\dist(x,y)=|x-y|$, and $\{X_t\}$ is a standard stationary Gaussian process with $\gamma(k)\sim ck^{2\mathrm{d}_{\rm E}-1}L(k)$. Observe that 
    \begin{equation*} 
    D_\nu=\frac{2}{\sqrt\pi},\qquad \mathbb E|X_t-X_{t-k}|=\frac{2}{\sqrt\pi}\{1-\gamma(k)\}^{1/2}, 
    \end{equation*}
    thus we have $C(k)=2[1-\{1-\gamma(k)\}^{1/2}]/\sqrt\pi\sim\gamma(k)/\sqrt\pi$. By Definition~\ref{def_trace}, this gives $\mathrm d=\mathrm{d}_{\rm E}$.

For Example~\ref{ex_psd}, orthonormality of $u$ and $v$ gives
\begin{equation*}
\dist(X_t,X_{t-k})=\vartheta|G_t-G_{t-k}|\{2+\vartheta^2(G_t+G_{t-k})^2\}^{1/2}.
\end{equation*}
If $(G,H)$ is bivariate standard Gaussian with correlation $\rho$, then $G-H$ and $G+H$ are independent with variances $2(1-\rho)$ and $2(1+\rho)$, respectively. Thus, with $Z\sim N(0,1)$, the expected distance at correlation $\rho$ is
\begin{equation*}
\psi_\vartheta(\rho)=2\vartheta(1-\rho)^{1/2}\mathbb E|Z|\,\mathbb E\{1+\vartheta^2(1+\rho)Z^2\}^{1/2}.
\end{equation*}
A direct differentiation gives
\begin{equation*}
\psi_\vartheta'(0)=-\vartheta\mathbb E|Z|\,\mathbb E(1+\vartheta^2Z^2)^{-1/2}<0.
\end{equation*}
Thus $C(k)=\psi_\vartheta(0)-\psi_\vartheta\{\gamma(k)\}\sim-\psi_\vartheta'(0)\gamma(k)$. This gives $\mathrm d=\mathrm{d}_{\rm E}$ by Definition~\ref{def_trace}.

    In Example~\ref{ex_wass}, $\Omega$ is the set of probability distributions and $\dist$ is the $\mathcal L^2$-Wasserstein metric. Write $Q_t$ for the quantile function of $X_t$ and $\Phi_0$ for the standard-normal distribution function. Since $Q_t(u)=\mu_t+\sigma\Phi_0^{-1}(u)$,
    \begin{equation*} 
    \dist(X_t,X_s)=\left\{\int_0^1\{Q_t(u)-Q_s(u)\}^2\,du\right\}^{1/2}=|\mu_t-\mu_s|. \end{equation*}
    This example reduces exactly to Example~\ref{ex_gaussian}; hence $\mathrm d=\mathrm{d}_{\rm E}$ by Definition~\ref{def_trace}.   

    Moreover, scale variation can be added without changing the fixed-scale result above. In particular, let $X_t=N(\mu_t,\sigma_t^2)$ with
    \begin{equation*}
    \sigma_t=\sigma\exp(\lambda\mu_t),\qquad \sigma>0,\quad\lambda\geq0. 
    \end{equation*}
    The one-dimensional Gaussian Wasserstein identity then gives
    \begin{equation}\label{eq_wass_2} 
    W_2(X_t,X_s)=\left[(\mu_t-\mu_s)^2+\sigma^2\{e^{\lambda\mu_t}-e^{\lambda\mu_s}\}^2\right]^{1/2}. 
    \end{equation}
    Thus both location and scale move with one Gaussian driver, while the standard deviation remains positive. If $(G,H)$ is bivariate standard Gaussian with correlation $\rho$, let $\psi_{\sigma,\lambda}(\rho)$ be the expectation of the right-hand side of \eqref{eq_wass_2} with $(\mu_t,\mu_s)=(G,H)$. First note that the joint density of the standard bivariate Gaussian $(G,H)$ with correlation $\rho$ is
    \[
    \phi_\rho(g,h) = \frac{1}{2\pi\sqrt{1-\rho^2}} \exp\left( -\frac{g^2 - 2\rho gh + h^2}{2(1-\rho^2)} \right).
    \]
    By the chain rule, we get 
    \[
    \frac{\partial \phi_\rho}{\partial \rho} = \phi_\rho(g,h) \times \frac{\partial}{\partial \rho} \left[ -\frac{1}{2}\ln(1-\rho^2) - \frac{g^2 - 2\rho gh + h^2}{2(1-\rho^2)} \right]
    \]
    and thus $\left. \frac{\partial \phi_\rho}{\partial \rho} \right\vert{}_{\rho=0} = g h \, \phi_0(g,h)$. 
    Write $R(g,h)=\sqrt{(g-h)^2+\sigma^2(e^{\lambda g}-e^{\lambda h})^2}$, so that $\psi_{\sigma,\lambda}(\rho)=\iint_{\mathbb R^2}R(g,h)\phi_\rho(g,h)\,dg\,dh$. Since $R(g,h)\leq |g|+|h|+\sigma(e^{\lambda g}+e^{\lambda h})$, Gaussian tail bounds justify differentiation under the integral sign in a neighborhood of $\rho=0$. Using the density derivative above gives
\begin{equation*}
\psi_{\sigma,\lambda}'(0)=\iint_{\mathbb R^2}ghR(g,h)\phi_0(g,h)\,dg\,dh
=\mathbb E\left[\widetilde G\widetilde H R(\widetilde G,\widetilde H)\right].
\end{equation*}
The last equality follows because $\phi_0$ is the joint density of independent standard normal variables $\widetilde G$ and $\widetilde H$.
    Let $D = \{(g,h) \in \mathbb{R}^2 : g > h > 0\}$. Then
    \[
    \psi_{\sigma,\lambda}'(0) = \iint_D 2gh \Big[ R(g,h) + R(-g,-h) - R(g,-h) - R(-g,h) \Big] \phi_0(g,h) \,dg\,dh.
    \]
    Observe that
    \begin{align*}
        R^2(g,-h) - R^2(g,h) &= \left[ (g+h)^2 - (g-h)^2 \right] + \sigma^2 \left[ (e^{\lambda g} - e^{-\lambda h})^2 - (e^{\lambda g} - e^{\lambda h})^2 \right] >0,
    \end{align*}
    and 
    \begin{align*}
        R^2(-g,h) - R^2(-g,-h) &= \left[ (g+h)^2 - (g-h)^2 \right] + \sigma^2 \left[ (e^{-\lambda g} - e^{\lambda h})^2 - (e^{-\lambda g} - e^{-\lambda h})^2 \right]>0.
    \end{align*}
    Thus we can conclude that $\psi_{\sigma,\lambda}'(0)<0$. Given this, a first-order expansion gives $C(k)\sim-\psi_{\sigma,\lambda}'(0)\gamma(k)$ and hence $\mathrm d=\mathrm{d}_{\rm E}$ by Definition~\ref{def_trace}. 
\end{proof}

\subsection{Proof of results in Section~\ref{sect::estimation}}

\begin{proof}[\textbf{Proof of Lemma~\ref{lem_dhat}}]
For each $k\in\{1,\ldots,n-1\}$, exactly $2(n-k)$ ordered pairs $(i,j)$ satisfy $|i-j|=k$. Stationarity and \eqref{eq_c} give $\mathbb E \dist(X_i,X_j)=D_\nu-C(k)$ for every such pair. Grouping the terms in \eqref{eq_dhat} according to $|i-j|$ therefore yields
\begin{equation*} \mathbb E\widehat D_n=\frac{2}{n(n-1)}\sum_{k=1}^{n-1}(n-k)\{D_\nu-C(k)\}. \end{equation*}
Because $2\sum_{k=1}^{n-1}(n-k)=n(n-1)$, the contribution of $D_\nu$ is exactly $D_\nu$, and subtracting it proves \eqref{eq_dhat_1}.

Suppose now that \eqref{eq_trace} holds. Since $2\mathrm d-1>-1$, Lemma~\ref{lem_sum} applies to $C(k)$ with $\beta=0$. It also applies with $\beta=1$, because $2\mathrm d>-1$. Thus
\begin{equation*} 
\sum_{k=1}^{n-1}C(k)\sim\frac{\const_C}{2\mathrm d}n^{2\mathrm d}L_C(n),\qquad \sum_{k=1}^{n-1}kC(k)\sim\frac{\const_C}{2\mathrm d+1}n^{2\mathrm d+1}L_C(n). 
\end{equation*}
Consequently,
\begin{align*}
\sum_{k=1}^{n-1}(n-k)C(k)&=n\sum_{k=1}^{n-1}C(k)-\sum_{k=1}^{n-1}kC(k)\sim\frac{\const_C}{2\mathrm d(2\mathrm d+1)}n^{2\mathrm d+1}L_C(n).
\end{align*}
Substitution into \eqref{eq_dhat_1}, together with $n(n-1)\sim n^2$, proves \eqref{eq_dhat_2}.

If $\sum_{k\geq1}|C(k)|<\infty$, then \eqref{eq_dhat_1} gives
\begin{equation*} 
|\mathbb E\widehat D_n-D_\nu|\leq\frac{2}{n-1}\sum_{k\geq1}|C(k)|=O(n^{-1}).
\end{equation*}
Finally, under $\sum_{k\geq1}k|C(k)|<\infty$, rewrite \eqref{eq_dhat_1} as
\begin{equation*} 
\mathbb E\widehat D_n-D_\nu=-\frac{2}{n-1}\sum_{k=1}^{n-1}C(k)+\frac{2}{n(n-1)}\sum_{k=1}^{n-1}kC(k). 
\end{equation*}
The first-moment condition implies $\sum_{k\geq n}|C(k)|\leq n^{-1}\sum_{k\geq n}k|C(k)|=o(n^{-1})$ and $\sum_{k=1}^{n-1}kC(k)=O(1)$. Hence $\sum_{k=1}^{n-1}C(k)=\mathsf B_\infty/2+o(n^{-1})$, so the first term on the right is $-\mathsf B_\infty/n+O(n^{-2})$, while the second is $O(n^{-2})$. This proves \eqref{eq_dhat_3}.
\end{proof}

\begin{proof}[\textbf{Proof of Proposition~\ref{prop_scale}}]
From \eqref{eq_bartlett_1}, $\mathsf B_C(r)=2\sum_{k=1}^{r-1}C(k)-\frac{2}{r}\sum_{k=1}^{r-1}kC(k)$. Under \eqref{eq_trace}, Lemma~\ref{lem_sum} gives
\begin{equation*} 
\sum_{k=1}^{r-1}C(k)\sim\frac{\const_C}{2\mathrm d}r^{2\mathrm d}L_C(r),\qquad \sum_{k=1}^{r-1}kC(k)\sim\frac{\const_C}{2\mathrm d+1}r^{2\mathrm d+1}L_C(r). 
\end{equation*}
Therefore
\begin{equation}\label{eq_bartlett_2} 
\mathsf B_C(r)\sim\frac{\const_C}{\mathrm d(2\mathrm d+1)}r^{2\mathrm d}L_C(r). 
\end{equation}
Because $2\mathrm d>0$, $\mathsf B_C(r)\to\infty$, whereas $2a_{n,m}D_\nu\to0$. For $r=\lfloor\ell m\rfloor$, $|r/m-\ell|\leq m^{-1}$, and the uniform convergence theorem for slowly varying functions gives $L_C(r)/L_C(m)\to1$ uniformly in $\ell\in[1,q]$. Since $r\geq m\to\infty$, \eqref{eq_bartlett_2} also holds uniformly over these bandwidths. Taking absolute values therefore proves part (i).

Under short memory and the summability condition $\sum_{k=1}^{\infty}k|C(k)|<\infty$,
\begin{equation*} 
\mathsf B_C(r)-\mathsf B_\infty=-2\sum_{k=r}^\infty C(k)-\frac{2}{r}\sum_{k=1}^{r-1}kC(k)=O(r^{-1}). 
\end{equation*}
Indeed, $\sum_{k\geq r}|C(k)|\leq r^{-1}\sum_{k\geq r}k|C(k)|=o(r^{-1})$, while $\sum_{k<r}kC(k)=O(1)$. Since $a_{n,m}\to0$, part (ii) follows when $\mathsf B_\infty\neq0$. If $\mathsf B_\infty=0$, then $\mathsf B_C(r)=O(m^{-1})$ uniformly when $r/m$ ranges over a fixed compact subset of $(0,\infty)$. The conditions $m^{-1}=o(a_{n,m})$ and $D_\nu>0$ therefore give $\mathsf B_{C,n}(\lfloor\ell m\rfloor;m)=2a_{n,m}D_\nu\{1+o(1)\}$, proving part (iii).
Dividing the corresponding expansions by those at $\ell=1$ gives the stated uniform ratio limits.
\end{proof}

\begin{proof}[\textbf{Proof of Lemma~\ref{lem_raw_bias}}]
Since $\mathbb E\widehat C_n(k)=C(k)+b_n$ and $2\sum_{k=1}^{r-1}(1-k/r)=r-1$, the definitions of the signed aggregates give \eqref{eq_raw_bias}.
\end{proof}

\begin{proof}[\textbf{Proof of Theorem~\ref{thm_raw}}]
Under long memory, Proposition~\ref{prop_scale} gives a constant $c>0$ such that
\[
\min_{r\in\mathcal G_m}S_{B,n}(r;m)\geq c m^{2\mathrm d}L_C(m)
\] 
for all sufficiently large $n$. Under short memory with $\mathsf B_\infty\neq0$, Proposition~\ref{prop_scale}(ii) similarly gives $\min_{r\in\mathcal G_m}S_{B,n}(r;m)\geq c|\mathsf B_\infty|$. When $\mathsf B_\infty=0$, part (iii) gives $\min_{r\in\mathcal G_m}S_{B,n}(r;m)\geq ca_{n,m}D_\nu$.

For every $r\in\mathcal G_m$, the reverse triangle inequality gives
\begin{equation*} 
|\widehat S_{B,n}(r;m)-S_{B,n}(r;m)|\leq|\widehat{\mathsf B}_{C,n}(r;m)-\mathsf B_{C,n}(r;m)|. \end{equation*}
Dividing this inequality by $S_{B,n}(r;m)$, taking the maximum over $r\in\mathcal G_m$, and combining Assumption~\ref{ass_raw} with the corresponding lower bound above yields
\begin{equation}\label{eq_raw_rel} 
\max_{r\in\mathcal G_m}\left|\frac{\widehat S_{B,n}(r;m)}{S_{B,n}(r;m)}-1\right|=o_p(1). \end{equation}
With probability tending to one, $\widehat S_{B,n}(r;m)\geq S_{B,n}(r;m)/2>0$ for all $r\in\mathcal G_m$, so all logarithms are well defined. On this event, \eqref{eq_raw_rel} and the bound $|\log x|\leq2|x-1|$ for $x\geq1/2$ give $\max_{r\in\mathcal G_m}|\log\widehat S_{B,n}(r;m)-\log S_{B,n}(r;m)|=o_p(1)$.

At the upper bandwidth $\lfloor qm\rfloor$, \eqref{eq_raw_rel} gives
\begin{equation*}
\frac{\widehat S_{B,n}(\lfloor qm\rfloor;m)}{\widehat S_{B,n}(m;m)}=\frac{S_{B,n}(\lfloor qm\rfloor;m)}{S_{B,n}(m;m)}\{1+o_p(1)\}.
\end{equation*}
The deterministic ratio converges to $q^{2\mathrm d}$ under long memory and to one under short memory. Since $\lfloor qm\rfloor/m\to q$, continuity of the logarithm gives convergence of the expression inside $\Pi_{\mathcal I_{\mathrm O}}$ in \eqref{eq_raw_r}. For the log-slope estimator, uniform convergence over $r/m\in[1,q]$ gives $\log S_{B,n}(r;m)=c_{n,m}+2\mathrm d\log r+o(1)$ under long memory and $c_{n,m}+o(1)$ under short memory. Because $|\mathcal G_m|^{-1}\sum_{r\in\mathcal G_m}(\log r-\bar\ell_m)^2$ converges to a positive constant, substitution into \eqref{eq_raw_s} gives the second unprojected limit. Finally, if $x_n\to\mathrm d$ and $\mathrm d\in\mathcal I_{\mathrm O}$, then $|\Pi_{\mathcal I_{\mathrm O}}(x_n)-\mathrm d|\leq|x_n-\mathrm d|\to0$. Applying this inequality to the two unprojected estimators completes the proof.
\end{proof}

\begin{proof}[\textbf{Proof of Lemma~\ref{lem_block_bias}}]
Fix $s\in\mathcal S$ and write $\ell=\ell_{n,s}$. By stationarity, all $s$ blocks have the same expectation. Within any block there are $2(\ell-k)$ ordered pairs separated by lag $k$, so the counting identity in Lemma~\ref{lem_dhat} gives 
\begin{equation}\label{eq_block_bias_4} 
\mathbb E\widehat D_{n,s}-D_\nu=-\frac{2}{\ell(\ell-1)}\sum_{k=1}^{\ell-1}(\ell-k)C(k). 
\end{equation}
This is exactly \eqref{eq_dhat_1} with $n$ replaced by $\ell$. Applying the long-memory and first-moment short-memory conclusions of Lemma~\ref{lem_dhat} at the block length $\ell$ gives the corresponding expansions. Finally, $\ell_{n,s}/n\to1/s$ and $L_C(\ell_{n,s})/L_C(n)\to1$ for every fixed $s$, so \eqref{eq_dhat_2} yields the stated expansion in terms of $b_n$ under long memory. Since $\mathcal S$ is finite, all expansions hold uniformly over $s$.
\end{proof}

\begin{proof}[\textbf{Proof of Proposition~\ref{prop_block}}]
For $s\in\mathcal S\setminus\{1\}$, Lemma~\ref{lem_block_bias} and the uniform equivalence following \eqref{eq_block_finite} give
\begin{equation*}
\mathbb E(\widehat D_{n,s}-\widehat D_n)=b_nx_{n,s}^{\circ}(\mathrm d)+r_{n,s},\qquad \max_{s\in\mathcal S\setminus\{1\}}|r_{n,s}|=o(|b_n|).
\end{equation*}
Let $H_n(\eth)=\sum_{s>1}\{x_{n,s}^{\circ}(\eth)\}^2$. For every fixed $s>1$, $x_{n,s}(\eth)\to s^{1-2\eth}-1$ uniformly when $\eth$ ranges over a compact neighborhood of $\mathrm d$ contained in $(0,1/2)$. By the same uniform equivalence, $H_n(\eth)\to H(\eth)=\sum_{s>1}\{s^{1-2\eth}-1\}^2$ uniformly on that compact set. Because at least one $s$ exceeds one, $H(\eth)>0$ for every $\eth<1/2$, so $H_n(\eth)$ is uniformly bounded away from zero on that set for all sufficiently large $n$.
At the true parameter,
\begin{align*}
\bar b_n(\mathrm d)&=\frac{\sum_{s>1}x_{n,s}^{\circ}(\mathrm d)\{b_nx_{n,s}^{\circ}(\mathrm d)+r_{n,s}\}}{H_n(\mathrm d)}
=b_n+\frac{\sum_{s>1}x_{n,s}^{\circ}(\mathrm d)r_{n,s}}{H_n(\mathrm d)}=b_n\{1+o(1)\}.
\end{align*}
For the estimated exponent, the derivatives
\begin{equation*}
\frac{\partial}{\partial\eth}x_{n,s}(\eth)=2\log\left(\frac{\ell_{n,s}}{n}\right)\left(\frac{\ell_{n,s}}{n}\right)^{2\eth-1}
\end{equation*}
are uniformly bounded on the compact parameter set and over the finite block set. The mean-value theorem and $\widehat{\mathrm d}_{\mathrm P}\to_p\mathrm d$ imply $\max_s|x_{n,s}(\widehat{\mathrm d}_{\mathrm P})-x_{n,s}(\mathrm d)|=o_p(1)$. The same bound holds for $x_{n,s}^{\circ}$ by the uniform equivalence above. Substitution into \eqref{eq_block_pop}, together with the denominator bound above, proves $\bar b_n(\widehat{\mathrm d}_{\mathrm P})=b_n\{1+o_p(1)\}$.
\end{proof}

\begin{proof}[\textbf{Proof of Theorem~\ref{thm_bd}}]
Under short memory, Assumption~\ref{ass_bd} gives $\Pr(\widehat c_n=0)\to1$. On this event, $\widehat D_n^{\mathrm{BD}}=\widehat D_n$, $\widehat C_n^{\mathrm{BD}}(k)=\widehat C_n(k)$, and $\widehat S_{B,n,\mathrm{BD}}(r;m)=\widehat S_{B,n}(r;m)$ for every $r$. The result follows from Theorem~\ref{thm_raw}.

Under long memory, adding $\widehat c_n$ to $\widehat D_n$ adds the same amount to every $\widehat C_n(k)$; the two-sided Bartlett mass is $r-1$, and hence
\[
\widehat{\mathsf B}_{C,n,\mathrm{BD}}(r;m)-\widehat{\mathsf B}_{C,n}(r;m)=(r-1)\widehat c_n. 
\]
Since $r=O(m)$ uniformly over $\mathcal G_m$,
\begin{equation*} 
\max_{r\in\mathcal G_m}\left|\widehat{\mathsf B}_{C,n,\mathrm{BD}}(r;m)-\widehat{\mathsf B}_{C,n}(r;m)\right|\leq Cm|\widehat c_n|=o_p\{m^{2\mathrm d}L_C(m)\}. 
\end{equation*}
Combining this bound with Assumption~\ref{ass_raw} gives
\begin{equation*} 
\max_{r\in\mathcal G_m}|\widehat{\mathsf B}_{C,n,\mathrm{BD}}(r;m)-\mathsf B_{C,n}(r;m)|=o_p\{m^{2\mathrm d}L_C(m)\}. 
\end{equation*}
Proposition~\ref{prop_scale} supplies the same deterministic lower bound used in the proof of Theorem~\ref{thm_raw}. The reverse triangle inequality therefore yields
\[
\max_{r\in\mathcal G_m}\left|\frac{\widehat S_{B,n,\mathrm{BD}}(r;m)}{S_{B,n}(r;m)}-1\right|=o_p(1). 
\]
The ratio and log-slope arguments in Theorem~\ref{thm_raw} now apply with $\widehat S_{B,n}$ replaced by $\widehat S_{B,n,\mathrm{BD}}$.
\end{proof}

\begin{proof}[\textbf{Proof of Corollary~\ref{cor_update}}]
Theorem~\ref{thm_raw} supplies the base case. At each fixed step $j$, the correction satisfies Assumption~\ref{ass_bd} by hypothesis, so Theorem~\ref{thm_bd} applies at that step. Induction gives the result. This argument establishes consistency only; it does not imply that additional updates reduce finite-sample bias or variance.
\end{proof}

\begin{proof}[\textbf{Proof of Theorem~\ref{thm_fp}}]
Corollary~\ref{cor_update} gives $\widehat{\mathrm d}_{\esttype}^{(1)}\to_p\mathrm d$ and $\widehat{\mathrm d}_{\esttype}^{(2)}\to_p\mathrm d$, and hence $\widehat\varrho_{n,\esttype}=o_p(1)$. Every minimizer in \eqref{eq_fp_est} belongs to $\mathcal N_{n,\esttype}$, so
\begin{equation*} 
\left|\widehat{\mathrm d}_{\esttype,\mathrm{FP}}-\mathrm d\right|\leq\left|\widehat{\mathrm d}_{\esttype}^{(2)}-\mathrm d\right|+\widehat\varrho_{n,\esttype}=o_p(1). 
\end{equation*}
No contraction condition or uniform approximation of $\mathcal T_{n,\esttype}$ over a fixed global domain is required.
\end{proof}

\begin{proof}[\textbf{Proof of Proposition~\ref{prop_bartlett}}]
The Euler--Maclaurin expansion in equation~(5.6) of \citet{ibukiyama2014euler}, written in terms of the upper endpoint $r$, gives
\begin{align*}
\sum_{k=1}^{r-1}k^{2\mathrm d-1}
&=\frac{r^{2\mathrm d}}{2\mathrm d}-\frac12r^{2\mathrm d-1}
+\zeta(1-2\mathrm d)+O(r^{2\mathrm d-2}),\\
\sum_{k=1}^{r-1}k^{2\mathrm d}
&=\frac{r^{2\mathrm d+1}}{2\mathrm d+1}-\frac12r^{2\mathrm d}
+\zeta(-2\mathrm d)+O(r^{2\mathrm d-1}).
\end{align*}
Multiplying the first line by $2\const_C$ and subtracting $2\const_C/r$ times the second cancels the $r^{2\mathrm d-1}$ terms and proves \eqref{eq_curve_1}. Factoring
$\const_Cr^{2\mathrm d}/\{\mathrm d(2\mathrm d+1)\}$ at $r=m$ and
$r=\lfloor qm\rfloor$, and using $\log(1+x)=x+o(x)$, gives
\eqref{eq_curve_2}. Since $0<1-2\mathrm d<1$, one has
$\zeta(1-2\mathrm d)<0$; also $q^{-2\mathrm d}-1<0$, so the displayed coefficient is
positive. Adding the common stabilizer $2a_{n,m}D_\nu$ to both bandwidth aggregates gives
a slope perturbation of order $a_{n,m}m^{-2\mathrm d}$, which is smaller because
$a_{n,m}\to0$.
\end{proof}

\begin{proof}[\textbf{Proof of Lemma~\ref{lem_dt}}]
Put $N_k=n-k$. Since $\widehat\delta_n(k)=N_k^{-1}\sum_{t=k+1}^nZ_{t,k}$ and $\mathbb E\widehat\delta_n(k)=\delta(k)$, stationarity of $\{Z_{t,k}\}_t$ and grouping by lag give
\begin{equation*} 
\operatorname{Var}\{\widehat\delta_n(k)\}=\frac{1}{N_k^2}\sum_{|j|<N_k}(N_k-|j|)\gamma_{Z,k}(j)\leq\frac{1}{N_k}\sum_{|j|<N_k}|\gamma_{Z,k}(j)|\leq\bar cN_k^{2\mathrm d-1}L_C^\star(N_k), \end{equation*}
where the final inequality is Assumption~\ref{ass_dt}. If $k\leq\lfloor qm\rfloor$, then $N_k/n=1-k/n\to1$ uniformly because $m/n\to0$. Potter's bounds for $L_C^\star$ (see, e.g., Theorem~1.5.6 of \citealp{bingham1987regular}) therefore imply
\begin{equation*} 
\sup_{k\leq\lfloor qm\rfloor}\operatorname{Var}\{\widehat\delta_n(k)\}\leq Cn^{2\mathrm d-1}L_C^\star(n) 
\end{equation*}
for all sufficiently large $n$. By Cauchy--Schwarz,
\begin{equation*} 
\sup_{k\leq\lfloor qm\rfloor}\mathbb E|\widehat\delta_n(k)-\delta(k)|\leq Cn^{\mathrm d-1/2}\sqrt{L_C^\star(n)}. 
\end{equation*}
Consequently,
\begin{equation*} 
\mathbb E\left\{\frac{1}{\lfloor qm\rfloor}\sum_{k=1}^{\lfloor qm\rfloor}|\widehat\delta_n(k)-\delta(k)|\right\}\leq Cn^{\mathrm d-1/2}\sqrt{L_C^\star(n)}, 
\end{equation*}
and Markov's inequality proves the desired result.

For the second result, grouping the ordered pairs in \eqref{eq_dhat} by their lag gives the exact identities
\begin{equation*} 
\widehat D_n=\sum_{k=1}^{n-1}w_{n,k}\widehat\delta_n(k),\qquad \mathbb E\widehat D_n=\sum_{k=1}^{n-1}w_{n,k}\delta(k),\qquad w_{n,k}=\frac{2(n-k)}{n(n-1)}. 
\end{equation*}
Hence
\begin{equation*} 
\widehat D_n-\mathbb E\widehat D_n=\sum_{k=1}^{n-1}w_{n,k}\{\widehat\delta_n(k)-\delta(k)\}. \end{equation*}
The variance bound above remains valid for every $k<n$, with $N_k=n-k$. Thus
\begin{align*}
\mathbb E|\widehat D_n-\mathbb E\widehat D_n|\leq\sum_{k=1}^{n-1}w_{n,k}\{\operatorname{Var}(\widehat\delta_n(k))\}^{1/2}&\leq\frac{2\bar c^{1/2}}{n(n-1)}\sum_{k=1}^{n-1}(n-k)^{\mathrm d+1/2}\sqrt{L_C^\star(n-k)}\\
&=\frac{2\bar c^{1/2}}{n(n-1)}\sum_{u=1}^{n-1}u^{\mathrm d+1/2}\sqrt{L_C^\star(u)}.
\end{align*}
The summand is regularly varying with index $\mathrm d+1/2>-1$. Lemma~\ref{lem_sum} therefore makes the last display $O\{n^{\mathrm d-1/2}\sqrt{L_C^\star(n)}\}$. Another application of Markov's inequality proves the desired result.
\end{proof}

\begin{proof}[\textbf{Proof of Lemma~\ref{lem_block_rate}}]
Fix $s\in\mathcal S$ and write $\ell=\ell_{n,s}$. For block $j$ and lag $k<\ell$, define
\begin{equation*}
\widehat\delta_{j,\ell}(k)=\frac{1}{\ell-k}\sum_{t=(j-1)\ell+k+1}^{j\ell}\dist(X_t,X_{t-k}).
\end{equation*}
Counting ordered pairs inside the block by their lag gives
\begin{equation*}
\widehat D_{j,s}=\frac{2}{\ell(\ell-1)}\sum_{k=1}^{\ell-1}(\ell-k)\widehat\delta_{j,\ell}(k),\qquad \widehat D_{n,s}=\frac{1}{s}\sum_{j=1}^s\widehat D_{j,s}.
\end{equation*}
The sample mean $\widehat\delta_{j,\ell}(k)$ contains $\ell-k$ consecutive observations of the stationary process $\{Z_{t,k}\}_t$. Repeating the variance calculation in Lemma~\ref{lem_dt} with sample length $\ell-k$ gives
\begin{equation*}
\mathbb E|\widehat\delta_{j,\ell}(k)-\delta(k)|\leq C(\ell-k)^{\mathrm d-1/2}\sqrt{L_C^\star(\ell-k)}.
\end{equation*}
No independence across blocks is required. By the triangle inequality,
\begin{align*}
\mathbb E|\widehat D_{j,s}-\mathbb E\widehat D_{j,s}|&\leq\frac{2C}{\ell(\ell-1)}\sum_{k=1}^{\ell-1}(\ell-k)^{\mathrm d+1/2}\sqrt{L_C^\star(\ell-k)}=O\{\ell^{\mathrm d-1/2}\sqrt{L_C^\star(\ell)}\},
\end{align*}
where Lemma~\ref{lem_sum} is applied after putting $u=\ell-k$. Averaging over the fixed number $s$ of blocks gives the same order for $\widehat D_{n,s}$. Since $\ell_{n,s}\asymp n$ for fixed $s$, Markov's inequality gives the stated order for each $s$, and finiteness of $\mathcal S$ gives \eqref{eq_block_rate_1}.
For the block-difference slope, write $e_{n,s}=\widehat D_{n,s}-\mathbb E\widehat D_{n,s}$ for $s\in\mathcal S$. The uniform equivalence following \eqref{eq_block_finite} permits the representation
\begin{equation*}
\widehat D_{n,s}-\widehat D_n=b_nx_{n,s}^{\circ}(\mathrm d)+(e_{n,s}-e_{n,1})+r_{n,s},\qquad \max_{s>1}|r_{n,s}|=o(|b_n|).
\end{equation*}
Let $\widetilde x_{n,s}=x_{n,s}^{\circ}(\widehat{\mathrm d}_{\mathrm P})$ and $H_n^{\circ}(\widehat{\mathrm d}_{\mathrm P})=\sum_{s>1}\widetilde x_{n,s}^2$. Compactness and pilot consistency imply that this denominator is bounded away from zero with probability tending to one. Substitution into \eqref{eq_block_est} gives
\begin{align*}
\widehat b_n-b_n&=b_n\left\{\frac{\sum_{s>1}\widetilde x_{n,s}x_{n,s}^{\circ}(\mathrm d)}{H_n^{\circ}(\widehat{\mathrm d}_{\mathrm P})}-1\right\}+\frac{\sum_{s>1}\widetilde x_{n,s}(e_{n,s}-e_{n,1}+r_{n,s})}{H_n^{\circ}(\widehat{\mathrm d}_{\mathrm P})}.
\end{align*}
Choose a compact interval $\mathcal K\subset(0,1/2)$ whose interior contains $\mathrm d$. Pilot consistency implies $\mathbb P(\widehat{\mathrm d}_{\mathrm P}\in\mathcal K)\to1$. On $\mathcal K$, the uniform equivalence following \eqref{eq_block_finite} and the mean-value theorem applied to the limiting regressor $x_s(\eth)=s^{1-2\eth}-1$, whose derivative $x_s'(\eth)=-2(\log s)s^{1-2\eth}$ is uniformly bounded over the finite set $\mathcal S$, give
\begin{equation*}
\max_{s>1}|\widetilde x_{n,s}-x_{n,s}^{\circ}(\mathrm d)|=o_p(1).
\end{equation*}
Consequently,
\begin{equation*}
\frac{\sum_{s>1}\widetilde x_{n,s}x_{n,s}^{\circ}(\mathrm d)}{H_n^{\circ}(\widehat{\mathrm d}_{\mathrm P})}-1=\frac{\sum_{s>1}\widetilde x_{n,s}\{x_{n,s}^{\circ}(\mathrm d)-\widetilde x_{n,s}\}}{H_n^{\circ}(\widehat{\mathrm d}_{\mathrm P})}=o_p(1),
\end{equation*}
because the denominator is bounded away from zero with probability tending to one. Hence the first term is $o_p(|b_n|)$. The second is $O_p\{n^{\mathrm d-1/2}\sqrt{L_C^\star(n)}\}+o_p(|b_n|)$ by \eqref{eq_block_rate_1}. This proves \eqref{eq_block_rate_2}.
\end{proof}

\begin{proof}[\textbf{Proof of Proposition~\ref{prop_rates}}]
For any $r\leq\lfloor qm\rfloor$, use $\widehat C_n(k)-C(k)=(\widehat D_n-D_\nu)-\{\widehat\delta_n(k)-\delta(k)\}$ in the definition of the raw signed Bartlett aggregate. Since the two-sided Bartlett mass excluding zero equals $r-1$, we obtain the exact decomposition
\[
\widehat{\mathsf B}_{C,n}(r;m)-\mathsf B_{C,n}(r;m)=\{2a_{n,m}+r-1\}(\widehat D_n-D_\nu)-2\sum_{k=1}^{r-1}\left(1-\frac{k}{r}\right)\{\widehat\delta_n(k)-\delta(k)\}. 
\]
Because $r\leq qm$, $a_{n,m}=o(1)$, and $0\leq1-k/r\leq1$,
\begin{equation*}
\max_{r\in\mathcal G_m}|\widehat{\mathsf B}_{C,n}(r;m)-\mathsf B_{C,n}(r;m)|\leq C m|\widehat D_n-D_\nu|+2\sum_{k=1}^{\lfloor qm\rfloor}|\widehat\delta_n(k)-\delta(k)|.
\end{equation*}
Under long memory,
\begin{equation*} 
\widehat D_n-D_\nu=(\widehat D_n-\mathbb E\widehat D_n)+(\mathbb E\widehat D_n-D_\nu)=O_p(\mathfrak r_n)+O\{n^{2\mathrm d-1}L_C(n)\}. 
\end{equation*}
Lemma~\ref{lem_dt} also gives $\sum_{k=1}^{\lfloor qm\rfloor}|\widehat\delta_n(k)-\delta(k)|=O_p(m\mathfrak r_n)$. Hence the maximum error in the signed aggregates over $\mathcal G_m$ is $O_p(m\mathfrak r_n)+O\{mn^{2\mathrm d-1}L_C(n)\}$. The first term is $o_p\{m^{2\mathrm d}L_C(m)\}$ by \eqref{eq_rate_lm}. For the deterministic term,
\begin{equation*} \frac{mn^{2\mathrm d-1}L_C(n)}{m^{2\mathrm d}L_C(m)}=\left(\frac{m}{n}\right)^{1-2\mathrm d}\frac{L_C(n)}{L_C(m)}\longrightarrow0 \end{equation*}
by Potter's bound \citep[Theorem~1.5.6]{bingham1987regular} and $m/n\to0$. This verifies Assumption~\ref{ass_raw} under long memory.

Under short memory, Lemma~\ref{lem_dhat} gives $\mathbb E\widehat D_n-D_\nu=O(n^{-1})$. The same argument therefore produces the bound $O_p(m\mathfrak r_n)+O(m/n)$. If $\mathsf B_\infty\neq0$, the deterministic normalization in Assumption~\ref{ass_raw} is asymptotically the fixed positive number $|\mathsf B_\infty|$, and $m\mathfrak r_n+m/n=o(1)$ is sufficient. If $\mathsf B_\infty=0$, that normalization is of order $2a_{n,m}D_\nu$, and \eqref{eq_rate_sm} is sufficient.

It remains to verify Assumption~\ref{ass_bd}. Under long memory, let $\mathcal K=[\underline{\mathrm d},\overline{\mathrm d}]\subset(0,1/2)$. The uniform convergence following \eqref{eq_block_finite}, together with $\inf_{\eth\in\mathcal K}H(\eth)>0$, implies that the denominator in \eqref{eq_block_est} is uniformly bounded away from zero on $\mathcal K$ for sufficiently large $n$. By Cauchy--Schwarz, Lemma~\ref{lem_block_bias}, and \eqref{eq_block_rate_1},
\begin{equation*}
\sup_{\eth\in\mathcal K}|\widehat b_n(\eth)|\leq C\max_{s\in\mathcal S\setminus\{1\}}|\widehat D_{n,s}-\widehat D_n|=O_p(\mathfrak r_n+|b_n|).
\end{equation*}
Since \eqref{eq_corr} sets the correction to zero outside the activation interval, this gives $\widehat c_n=O_p(\mathfrak r_n+|b_n|)$ for any projected pilot. The rate condition \eqref{eq_rate_lm} makes $\mathfrak r_n=o\{m^{2\mathrm d-1}L_C(m)\}$. Moreover,
\begin{equation*}
\frac{|b_n|}{m^{2\mathrm d-1}L_C(m)}=O\left\{\left(\frac{m}{n}\right)^{1-2\mathrm d}\frac{L_C(n)}{L_C(m)}\right\}=o(1),
\end{equation*}
by \eqref{eq_dhat_2} and Potter's bound \citep[Theorem~1.5.6]{bingham1987regular}. Hence $\widehat c_n=o_p\{m^{2\mathrm d-1}L_C(m)\}$, as required.
Under short memory, Theorem~\ref{thm_raw} gives $\widehat{\mathrm d}_{\mathrm P}\to_p0$ for the initial projected raw pilot. Since $\underline{\mathrm d}>0$, the event $\{\underline{\mathrm d}<\widehat{\mathrm d}_{\mathrm P}\leq\overline{\mathrm d}\}$ has probability tending to zero, so \eqref{eq_corr} gives $\Pr(\widehat c_n=0)\to1$. Theorem~\ref{thm_bd} then yields consistency of the updated estimator, allowing the same argument to be repeated at each fixed update.
\end{proof}

\begin{proof}[\textbf{Proof of Corollary~\ref{cor_orders}}]
Under long memory, substituting $m=c_mn^\kappa\{1+o(1)\}$ into \eqref{eq_rate_lm} gives
\begin{equation*} 
\frac{n^{\mathrm d-1/2}\sqrt{L_C^\star(n)}}{m^{2\mathrm d-1}L_C(m)}=c_m^{1-2\mathrm d}n^{(1/2-\mathrm d)(2\kappa-1)}\frac{\sqrt{L_C^\star(n)}}{L_C(n^\kappa)}\{1+o(1)\}. 
\end{equation*}
The polynomial exponent is negative because $\mathrm d<1/2$ and $\kappa<1/2$. A negative polynomial power dominates the slowly varying ratio, so the expression converges to zero.

Under short memory, $\mathfrak r_n=n^{-1/2}\sqrt{L_C^\star(n)}$, so $\kappa<1/2$ implies $m\mathfrak r_n+m/n=o(1)$. Moreover, $a_{n,m}\sim c_a c_m^\eta n^{-\eta(1-\kappa)}$. The first upper bound in \eqref{eq_eta} gives $\kappa-1/2+\eta(1-\kappa)<0$, ensuring $m\mathfrak r_n=o(a_{n,m})$ because a negative polynomial power dominates the slowly varying factor. The second gives $-\kappa+\eta(1-\kappa)<0$, ensuring $m^{-1}=o(a_{n,m})$. Finally, $m/n=o(a_{n,m})$ follows from $\eta<1$, which is implied by \eqref{eq_eta}. This establishes \eqref{eq_rate_sm}.
\end{proof}

\begin{proof}[\textbf{Proof of Lemma~\ref{lem_shift}}]
Set $(a_1,a_2,b_1,b_2)=(t,t-k,t-j,t-j-k)$ in \eqref{eq_pair}. The four pairwise separations are $|j|$, $|j+k|$, $|j-k|$, and $|j|$. Absorbing the repeated $\varpi(|j|)$ term into the constant~gives
\begin{align*}
|\gamma_{Z,k}(j)|&\leq C\{\varpi(|j|)+\varpi(|j-k|)+\varpi(|j+k|)\}, \\ 
\sum_{|j|<n}|\gamma_{Z,k}(j)|&\leq C\sum_{|j|<n}\varpi(|j|)+C\sum_{|j|<n}\varpi(|j-k|)+C\sum_{|j|<n}\varpi(|j+k|)\leq C n^{2\mathrm d}L_C^\star(n),
\end{align*}
where~\eqref{eq_shift} is applied with shifts $0$, $k$, and $-k$. The bound is uniform in $k$, so taking the supremum proves Assumption~\ref{ass_dt}.
\end{proof}

\begin{proof}[\textbf{Proof of Proposition~\ref{prop::example assumption}}]
Write $\{G_t\}$ for the scalar standard Gaussian driver and $\gamma(h)=\mathbb E(G_tG_{t-h})$. Since $\gamma(h)\to0$, we have $\sup_{h\geq1}|\gamma(h)|<1$: an absolute correlation of one would imply nondecaying covariances along multiples of that lag. Thus the covariance matrices of $(G_t,G_{t-k})$, $k\geq1$, are uniformly well-conditioned.
For each example, write $\dist(X_t,X_s)=h(G_t,G_s)$ and let $(U,V)$ be any standard bivariate Gaussian pair. In Example~\ref{ex_gaussian}, $h(x,y)=|x-y|$, so $\mathbb Eh(U,V)^2\leq4$. In Example~\ref{ex_psd}, orthonormality of $u$ and $v$ gives $h(x,y)^2=2\vartheta^2(x-y)^2+\vartheta^4(x^2-y^2)^2$, and hence $\mathbb Eh(U,V)^2\leq8\vartheta^2+4\vartheta^4$. In Example~\ref{ex_wass}, the Gaussian Wasserstein identity gives $h(x,y)^2=(x-y)^2+\sigma^2(e^{\lambda x}-e^{\lambda y})^2$, so $\mathbb Eh(U,V)^2\leq4+4\sigma^2e^{2\lambda^2}$; $\lambda=0$ includes the location-only case. These bounds are uniform in the correlation of $(U,V)$.
The moment and Hermite-rank conditions in Section~\ref{app_dt} therefore hold with $\tau=1$. The Gaussian covariance bound established there gives
\begin{equation*}
|\gamma_{Z,k}(j)|\leq C\{|\gamma(j)|+|\gamma(j-k)|+|\gamma(j+k)|\},
\end{equation*}
uniformly in $k\geq1$ and $j\in\mathbb Z$. Moreover,
\begin{equation*}
\sup_{s\in\mathbb Z}\sum_{|j|<n}|\gamma(j-s)|\leq\sum_{|h|<n}|\gamma(h)|+2n\sup_{h\geq n}|\gamma(h)|=O\{n^{2\mathrm d_{\rm E}}L(n)\},
\end{equation*}
where Karamata's theorem (Lemma \ref{lem_sum}) bounds the first term and Potter's bound \citep[Theorem~1.5.6]{bingham1987regular} controls the second because $2\mathrm d_{\rm E}-1<0$. Thus \eqref{eq_shift} holds with $\varpi(h)=|\gamma(h)|$ and $L_C^\star=L$. Since $\mathrm d=\mathrm d_{\rm E}$ by Proposition~\ref{prop::example memory parameter}, summing the covariance bound establishes Assumption~\ref{ass_dt}.
\end{proof}

\newpage
\section{Additional theoretical results}\label{supp_sec_theory}

\subsection{Trace memory and ADCV memory}\label{sect::compare with ADCV}

We formulate an alternative definition of long memory based on the ADCV considered in the literature \citep{zhou2012measuring,lyons2013distance,jiang2024testing} and compare it with Definition~\ref{def_trace}.
Following \citet{lyons2013distance}, write the distance covariance in metric spaces as
\begin{equation*}
\operatorname{dCov}(X,Y)=\mathbb E\{\dist_{\nu_X}(X,X')\dist_{\nu_Y}(Y,Y')\},
\end{equation*}
where $\nu_X$ and $\nu_Y$ are the marginal laws with finite first moments, $\dist_\mu=-K_\mu$ denotes the $\mu$-centered distance, and $(X',Y')$ is an independent copy of $(X,Y)$. For an object-valued stationary time series $\{X_t\}_{t\in\mathbb Z}$, \citet{jiang2024testing} define the ADCV as
\begin{equation*}
V(k)=\mathbb E\{\dist_\nu(X_t,X_t')\dist_\nu(X_{t-k},X_{t-k}')\},
\end{equation*}
where $\{X_t'\}_{t\in\mathbb Z}$ is an independent copy of the entire process $\{X_t\}_{t\in\mathbb Z}$; see \citet{zhou2012measuring} for the Euclidean setting. Thus, $V(k)$ measures lag-$k$ dependence by averaging products of centered distances between two independent process copies. Under the negative-type condition imposed below, $V(k)\geq0$.
For the embedded process $Y_t=\Phi_\nu(X_t)$, the following proposition identifies $V(k)$ with the squared Hilbert--Schmidt norm of $\Gamma_k$, whereas $C(k)$ is its trace. These two summaries capture different aspects of temporal dependence and need not yield the same memory classification.

\begin{proposition}\label{prop::ADCV}
Assume that $(\Omega,\dist)$ is of negative type and let $Y_t=\Phi_\nu(X_t)$ be the Hilbert-space embedding defined above. Then for $\Gamma_k=\mathbb E(Y_t\otimes Y_{t-k})$,
\begin{equation} \label{eq_v}
V(k)=\|\Gamma_k\|_{\mathcal S_2}^2.
\end{equation}
\end{proposition}

\begin{proof}
    Let $\{Y_t'\}$ be an independent copy of $\{Y_t\}$. Since
\begin{equation*}
    \dist_\nu(X_t,X_t')=-\langle Y_t,Y_t'\rangle_{\mathcal H}, \qquad
    \dist_\nu(X_{t-k},X_{t-k}')=-\langle Y_{t-k},Y_{t-k}'\rangle_{\mathcal H},
\end{equation*}
it follows that
\begin{equation*}
    V(k)=\mathbb E\{\langle Y_t,Y_t'\rangle_{\mathcal H}\langle Y_{t-k},Y_{t-k}'\rangle_{\mathcal H}\}.
\end{equation*}
Using the Hilbert--Schmidt inner product, we obtain
\begin{equation*}
    \langle Y_t,Y_t'\rangle_{\mathcal H}\langle Y_{t-k},Y_{t-k}'\rangle_{\mathcal H}
    =
    \langle Y_t\otimes Y_{t-k},Y_t'\otimes Y_{t-k}'\rangle_{\mathcal S_2}.
\end{equation*}
Therefore,
\begin{equation*}
    V(k)=\mathbb E\langle Y_t\otimes Y_{t-k},Y_t'\otimes Y_{t-k}'\rangle_{\mathcal S_2}.
\end{equation*}
Because $\{Y_t'\}$ is an independent copy of $\{Y_t\}$, the random elements $Y_t\otimes Y_{t-k}$ and $Y_t'\otimes Y_{t-k}'$ are independent and identically distributed in the Hilbert--Schmidt space. Hence
\begin{align*}
    \mathbb E\langle Y_t\otimes Y_{t-k},Y_t'\otimes Y_{t-k}'\rangle_{\mathcal S_2} =
    \left\langle \mathbb E(Y_t\otimes Y_{t-k}),\mathbb E(Y_t'\otimes Y_{t-k}')
    \right\rangle_{\mathcal S_2}=\langle \Gamma_k,\Gamma_k\rangle_{\mathcal S_2}.
\end{align*}
Therefore, \eqref{eq_v} follows. This completes the proof.
\end{proof}

Because $V(k)=\|\Gamma_k\|_{\mathcal S_2}^2$ is on a squared covariance scale, we formulate the corresponding memory notion in terms of its square root, which we call the root ADCV.
\begin{definition}[Root-ADCV long memory]\label{def_adcv}
We say that the object-valued time series $\{X_t\}$ has root-ADCV short memory if $\sum_{k=-\infty}^{\infty}\sqrt{V(k)}<\infty$. We say that it has root-ADCV long memory with parameter $\mathrm d_V\in(0,1/2)$ if
\begin{equation}\label{eq_adcv}
\sqrt{V(k)}\sim\const_V k^{2\mathrm d_V-1}L_V(k),\qquad k\to\infty,
\end{equation}
where $\const_V>0$ and $L_V$ is positive and slowly varying.
\end{definition}
\begin{remark}[Relation to the original ADCV scale]
If \eqref{eq_adcv} holds, then $V(k)\sim\const_V^2k^{4\mathrm d_V-2}L_V(k)^2$. Hence $\sum_kV(k)$ converges for $\mathrm d_V<1/4$ and diverges for $\mathrm d_V>1/4$; at $\mathrm d_V=1/4$, summability depends on $L_V$. Recall that a metric space of negative type is of \emph{strong negative type} if the distance energy between two probability measures with finite first moments vanishes only when the measures coincide. Under this property and the standing condition $\nu\in M_1(\Omega)$, $V(k)=0$ if and only if $X_t$ and $X_{t-k}$ are independent; see \citet{lyons2013distance}.
\end{remark}
The following result gives the corresponding decay rate for the root ADCV under the operator representation in \eqref{eq_model}.
\begin{proposition}\label{prop_common_ADCV}
Suppose that \eqref{eq_model} holds with $A\neq0$ and $\|R_k\|_1=o\{k^{2\mathrm d_0-1}L(k)\}$. Then the root ADCV satisfies $\sqrt{V(k)}\sim k^{2\mathrm d_0-1}L(k)\|A\|_{\mathcal S_2}$, so $\mathrm d_V=\mathrm d_0$.
\end{proposition}
The result follows directly from \eqref{eq_v} and the Hilbert--Schmidt norm calculation in the proof of Proposition~\ref{prop_common}, so we omit the proof.
We now compare the two scalar summaries $C(k)=\operatorname{tr}(\Gamma_k)$ and $\sqrt{V(k)}=\|\Gamma_k\|_{\mathcal S_2}$. The trace can vanish even when $\Gamma_k\neq0$, for example through cancellation of positive and negative contributions. The root ADCV remains positive whenever $\Gamma_k\neq0$. If $\Gamma_k$ is positive semidefinite, then
\begin{equation*}
0\leq\sqrt{V(k)}=\|\Gamma_k\|_{\mathcal S_2}\leq\|\Gamma_k\|_1=C(k).
\end{equation*}
If its rank is $r_k<\infty$, then also $C(k)\leq\sqrt{r_k}\sqrt{V(k)}$. When these ranks are uniformly bounded, the two summaries have comparable orders. When dependence is spread across an increasing number of directions, however, the Hilbert--Schmidt norm may decay faster than the trace. Without positive semidefiniteness, there is no universal ordering between $|C(k)|$ and $\sqrt{V(k)}$.
Thus, although the root ADCV avoids cancellation in the trace, its memory classification need not agree with Definition~\ref{def_trace}.

\subsection{Gaussian-driven verification of Assumption~\ref{ass_dt}}\label{app_dt}
Suppose that $X_t=G(\xi_t)$, where $\{\xi_t\}$ is a centered stationary Gaussian vector process in $\mathbb R^{p_\xi}$. Let
\begin{equation*}
    \mathcal R_\xi(h)=\|\operatorname{Cov}(\xi_t,\xi_{t-h})\|_{\max},
\end{equation*}
and define $h(u,v)=\dist\{G(u),G(v)\}$, identifying $(u,v)$ with the concatenated vector $(u^\top,v^\top)^\top\in\mathbb R^{2p_\xi}$. For each $k\geq1$,
\begin{equation*}
    Z_{t,k}=h(\xi_t,\xi_{t-k})
\end{equation*}
is a scalar function of the $2p_\xi$-dimensional Gaussian vector $W_{t,k}=(\xi_t^\top,\xi_{t-k}^\top)^\top$. Let $\Sigma_k=\operatorname{Var}(W_{t,k})$ and assume uniform conditioning: $0<c_0\leq\lambda_{\min}(\Sigma_k)\leq\lambda_{\max}(\Sigma_k)\leq C_0<\infty$. Put $\widetilde W_{t,k}=\Sigma_k^{-1/2}W_{t,k}$ and $h_k(w)=h(\Sigma_k^{1/2}w)$. The centered transform is
\begin{equation*}
    \bar h_k(w)=h_k(w)-\mathbb Eh_k(\widetilde W_{t,k}).
\end{equation*}
Assume $\sup_k\mathbb EZ_{t,k}^2<\infty$ and that every nondegenerate $\bar h_k$ has Hermite rank at least $\tau\geq1$.

The cross-covariance matrix between $W_{t,k}$ and $W_{t-j,k}$ has four blocks corresponding to lags $j$, $j+k$, $j-k$, and $j$. Therefore
\begin{equation*}
    \|\operatorname{Cov}(W_{t,k},W_{t-j,k})\|_{\max}\leq C\{\mathcal R_\xi(|j|)+\mathcal R_\xi(|j-k|)+\mathcal R_\xi(|j+k|)\}.
\end{equation*}
Uniform conditioning gives the same bound after standardization. Let $\psi_{k,j}$ be the largest row or column sum of absolute entries in $\operatorname{Cov}(\widetilde W_{t,k},\widetilde W_{t-j,k})$. It is bounded by a constant times the preceding three-term envelope. If $\psi_{k,j}<1$, \citet[Lemma~1]{arcones1994limit} gives $|\gamma_{Z,k}(j)|\leq\psi_{k,j}^\tau\operatorname{Var}(Z_{t,k})$; if $\psi_{k,j}\geq1$, the same bound follows from Cauchy--Schwarz. Using the uniform second-moment bound, we obtain
\begin{align*}
|\gamma_{Z,k}(j)|&\leq C\{\mathcal R_\xi(|j|)+\mathcal R_\xi(|j-k|)+\mathcal R_\xi(|j+k|)\}^\tau\\
&\leq C\{\mathcal R_\xi(|j|)^\tau+\mathcal R_\xi(|j-k|)^\tau+\mathcal R_\xi(|j+k|)^\tau\}.
\end{align*}
This argument also covers the overlap indices $j\in\{0,-k,k\}$. Consequently, the shifted-sum argument in Lemma~\ref{lem_shift} verifies Assumption~\ref{ass_dt} whenever $\varpi(h)=\mathcal R_\xi(h)^\tau$ satisfies \eqref{eq_shift}. The baseline case $\tau=1$ uses the Gaussian covariance envelope itself, while a larger Hermite rank permits slower decay of that envelope. See \citet{Viitasaari20} for multivariate Hermite expansions and Hermite rank.

\newpage
\section{Additional simulation results}\label{supp sect::additional numerical}

\subsection{Simulation results for real-valued time series}\label{supp subsect::real}

Tables~\ref{tab_mc_scalar_ratio} and~\ref{tab_mc_scalar_ls} evaluate the finite-sample performance of the log-ratio estimators and the log-slope estimators, respectively, for the Real design.

\begin{table}[H]
\centering
\renewcommand{\arraystretch}{0.85}
\caption{\small Monte Carlo results for the Real design using the log-ratio estimators. Entries are means with RMSEs in parentheses, based on 1,000 replications.}
\label{tab_mc_scalar_ratio}
\fontsize{7.3pt}{8.4pt}\selectfont
\setstretch{1.0}
\setlength{\tabcolsep}{9pt}
\begin{tabular}{@{}crrrrrr@{}}
\toprule
& \multicolumn{3}{c}{Baseline} & \multicolumn{3}{c}{Tuning average} \\
\cmidrule(lr){2-4}\cmidrule(lr){5-7}
$\mathrm d$ & Raw & BC & FP & Raw & BC & FP \\
\midrule
\multicolumn{7}{l}{\hspace{-0.1in}{$n=250$}} \\
.00 & -.022 (.130) & -.016 (.138) & -.015 (.138) & -.021 (.126) & -.015 (.133) & -.014 (.134) \\
.10 & .068 (.125) & .089 (.140) & .090 (.143) & .067 (.122) & .088 (.137) & .090 (.140) \\
.20 & .143 (.121) & .185 (.135) & .190 (.142) & .142 (.119) & .184 (.134) & .190 (.140) \\
.30 & .216 (.124) & .286 (.125) & .297 (.133) & .215 (.124) & .285 (.124) & .297 (.132) \\
.40 & .281 (.143) & .376 (.113) & .393 (.114) & .279 (.144) & .376 (.112) & .392 (.114) \\
\addlinespace
\multicolumn{7}{l}{\hspace{-0.1in}{$n=500$}} \\
.00 & -.013 (.118) & -.010 (.122) & -.010 (.122) & -.013 (.110) & -.010 (.115) & -.010 (.115) \\
.10 & .082 (.097) & .097 (.107) & .097 (.108) & .082 (.094) & .096 (.104) & .097 (.105) \\
.20 & .162 (.093) & .196 (.104) & .199 (.108) & .162 (.092) & .196 (.103) & .199 (.106) \\
.30 & .240 (.092) & .299 (.094) & .309 (.103) & .240 (.091) & .299 (.093) & .309 (.102) \\
.40 & .300 (.118) & .380 (.090) & .394 (.093) & .301 (.117) & .381 (.090) & .395 (.092) \\
\addlinespace
\multicolumn{7}{l}{\hspace{-0.1in}{$n=1000$}} \\
.00 & -.014 (.103) & -.012 (.105) & -.012 (.105) & -.013 (.099) & -.011 (.101) & -.011 (.101) \\
.10 & .092 (.078) & .103 (.085) & .103 (.085) & .092 (.077) & .103 (.084) & .103 (.084) \\
.20 & .179 (.067) & .205 (.077) & .207 (.079) & .179 (.066) & .205 (.076) & .207 (.078) \\
.30 & .257 (.071) & .304 (.075) & .311 (.083) & .256 (.071) & .304 (.075) & .311 (.082) \\
.40 & .318 (.096) & .387 (.071) & .399 (.074) & .318 (.096) & .387 (.071) & .399 (.074) \\
\addlinespace
\multicolumn{7}{l}{\hspace{-0.1in}{$n=1500$}} \\
.00 & -.012 (.093) & -.010 (.095) & -.010 (.095) & -.011 (.089) & -.009 (.091) & -.009 (.091) \\
.10 & .100 (.067) & .109 (.073) & .109 (.073) & .100 (.066) & .109 (.072) & .109 (.072) \\
.20 & .190 (.055) & .213 (.065) & .214 (.067) & .190 (.054) & .213 (.064) & .214 (.066) \\
.30 & .267 (.059) & .308 (.064) & .314 (.070) & .266 (.058) & .308 (.064) & .314 (.070) \\
.40 & .328 (.084) & .390 (.063) & .401 (.066) & .328 (.084) & .390 (.063) & .401 (.066) \\
\addlinespace
\multicolumn{7}{l}{\hspace{-0.1in}{$n=2000$}} \\
.00 & -.011 (.094) & -.010 (.095) & -.010 (.095) & -.010 (.088) & -.009 (.089) & -.009 (.089) \\
.10 & .101 (.062) & .110 (.068) & .110 (.068) & .102 (.061) & .110 (.067) & .110 (.067) \\
.20 & .191 (.052) & .213 (.062) & .214 (.063) & .192 (.051) & .213 (.060) & .214 (.062) \\
.30 & .267 (.056) & .305 (.059) & .310 (.064) & .267 (.055) & .306 (.058) & .310 (.063) \\
.40 & .330 (.080) & .388 (.059) & .398 (.062) & .331 (.080) & .389 (.058) & .399 (.061) \\
\bottomrule
\end{tabular}
\end{table}

\begin{table}[H]
\centering
\renewcommand{\arraystretch}{0.85}
\caption{\small Monte Carlo results for the Real design using the log-slope estimators. Entries are means with RMSEs in parentheses, based on 1,000 replications.}
\label{tab_mc_scalar_ls}
\fontsize{7.3pt}{8.4pt}\selectfont
\setstretch{1.0}
\setlength{\tabcolsep}{9pt}
\begin{tabular}{@{}crrrrrr@{}}
\toprule
& \multicolumn{3}{c}{Baseline} & \multicolumn{3}{c}{Tuning average} \\
\cmidrule(lr){2-4}\cmidrule(lr){5-7}
$\mathrm d$ & Raw & BC & FP & Raw & BC & FP \\
\midrule
\multicolumn{7}{l}{\hspace{-0.1in}{$n=250$}} \\
.00 & -.022 (.132) & -.016 (.139) & -.015 (.140) & -.021 (.126) & -.014 (.134) & -.014 (.135) \\
.10 & .068 (.126) & .089 (.142) & .091 (.144) & .068 (.123) & .089 (.138) & .090 (.141) \\
.20 & .143 (.122) & .185 (.137) & .190 (.143) & .142 (.120) & .184 (.134) & .190 (.140) \\
.30 & .216 (.124) & .285 (.125) & .297 (.133) & .215 (.124) & .285 (.125) & .296 (.132) \\
.40 & .280 (.144) & .376 (.113) & .392 (.114) & .279 (.144) & .375 (.113) & .391 (.114) \\
\addlinespace
\multicolumn{7}{l}{\hspace{-0.1in}{$n=500$}} \\
.00 & -.013 (.119) & -.009 (.123) & -.009 (.123) & -.013 (.112) & -.009 (.116) & -.009 (.116) \\
.10 & .083 (.098) & .097 (.108) & .098 (.109) & .082 (.095) & .097 (.105) & .097 (.106) \\
.20 & .162 (.094) & .196 (.105) & .199 (.109) & .162 (.092) & .196 (.103) & .199 (.107) \\
.30 & .240 (.093) & .299 (.095) & .309 (.103) & .240 (.092) & .299 (.094) & .309 (.102) \\
.40 & .300 (.119) & .379 (.091) & .393 (.093) & .300 (.118) & .380 (.090) & .394 (.092) \\
\addlinespace
\multicolumn{7}{l}{\hspace{-0.1in}{$n=1000$}} \\
.00 & -.014 (.105) & -.012 (.107) & -.012 (.107) & -.012 (.100) & -.010 (.102) & -.010 (.102) \\
.10 & .092 (.080) & .103 (.087) & .103 (.087) & .092 (.078) & .103 (.085) & .103 (.085) \\
.20 & .179 (.068) & .205 (.077) & .207 (.080) & .178 (.067) & .205 (.076) & .207 (.079) \\
.30 & .257 (.072) & .304 (.076) & .311 (.083) & .256 (.072) & .303 (.075) & .310 (.082) \\
.40 & .318 (.096) & .386 (.072) & .398 (.074) & .318 (.096) & .386 (.071) & .398 (.074) \\
\addlinespace
\multicolumn{7}{l}{\hspace{-0.1in}{$n=1500$}} \\
.00 & -.011 (.095) & -.010 (.097) & -.010 (.097) & -.010 (.090) & -.008 (.091) & -.008 (.091) \\
.10 & .100 (.068) & .109 (.075) & .109 (.075) & .100 (.067) & .109 (.073) & .109 (.073) \\
.20 & .190 (.055) & .213 (.066) & .214 (.067) & .190 (.055) & .213 (.065) & .214 (.066) \\
.30 & .266 (.059) & .308 (.064) & .313 (.070) & .266 (.059) & .308 (.064) & .313 (.070) \\
.40 & .328 (.084) & .390 (.064) & .400 (.067) & .328 (.084) & .390 (.063) & .400 (.066) \\
\addlinespace
\multicolumn{7}{l}{\hspace{-0.1in}{$n=2000$}} \\
.00 & -.011 (.097) & -.010 (.098) & -.010 (.098) & -.010 (.089) & -.009 (.090) & -.009 (.090) \\
.10 & .101 (.063) & .109 (.069) & .110 (.069) & .102 (.062) & .110 (.068) & .110 (.068) \\
.20 & .191 (.053) & .213 (.062) & .213 (.064) & .192 (.052) & .213 (.061) & .214 (.062) \\
.30 & .267 (.056) & .305 (.059) & .309 (.064) & .267 (.055) & .305 (.058) & .310 (.063) \\
.40 & .330 (.081) & .388 (.059) & .398 (.062) & .330 (.080) & .388 (.059) & .398 (.061) \\
\bottomrule
\end{tabular}
\end{table}

Tables~\ref{tab_mc_scalar_ratio} and~\ref{tab_mc_scalar_ls} show patterns similar to those in Table~\ref{tab_mc_psd_ratio}.

\subsection{Simulation results for matrix-valued time series}\label{supp subsect::matrix}

Table~\ref{tab_mc_psd_ls} reports the finite-sample performance of the log-slope estimators for the Matrix design, with RMSEs decreasing as $n$ increases, similar to Table~\ref{tab_mc_psd_ratio}. The log-ratio and log-slope constructions perform similarly, with slightly higher aggregate RMSEs for the log-slope estimators. The tuning-average columns below have smaller RMSEs than the baseline columns in most cells, which aggregates in Table~\ref{tab_mc_aggregate} to a reduction from $0.0919$ to $0.0902$ for Raw, from $0.0907$ to $0.0889$ for BC, and from $0.0930$ to $0.0912$ for FP, while leaving MAB essentially unchanged. At $\mathrm d=0$, the Monte Carlo means remain close to zero and negative for all methods, with BC and FP giving very similar results.

\begin{table}[H]
\centering
\renewcommand{\arraystretch}{0.85}
\caption{\small Monte Carlo results for the Matrix design using the log-slope estimators. Entries are means with RMSEs in parentheses, based on 1,000 replications.}
\label{tab_mc_psd_ls}
\fontsize{7.3pt}{8.4pt}\selectfont
\setstretch{1.0}
\setlength{\tabcolsep}{9pt}
\begin{tabular}{@{}crrrrrr@{}}
\toprule
& \multicolumn{3}{c}{Baseline} & \multicolumn{3}{c}{Tuning average} \\
\cmidrule(lr){2-4}\cmidrule(lr){5-7}
$\mathrm d$ & Raw & BC & FP & Raw & BC & FP \\
\midrule
\multicolumn{7}{l}{\hspace{-0.1in}{$n=250$}} \\
.00 & -.021 (.129) & -.015 (.136) & -.015 (.137) & -.021 (.124) & -.014 (.131) & -.014 (.132) \\
.10 & .066 (.124) & .086 (.138) & .087 (.141) & .065 (.121) & .085 (.135) & .087 (.138) \\
.20 & .139 (.122) & .179 (.135) & .185 (.141) & .138 (.121) & .179 (.133) & .184 (.139) \\
.30 & .212 (.127) & .279 (.125) & .291 (.132) & .211 (.127) & .279 (.124) & .290 (.131) \\
.40 & .277 (.146) & .372 (.114) & .388 (.115) & .276 (.147) & .371 (.114) & .387 (.115) \\
\addlinespace
\multicolumn{7}{l}{\hspace{-0.1in}{$n=500$}} \\
.00 & -.013 (.114) & -.009 (.118) & -.009 (.118) & -.013 (.108) & -.009 (.113) & -.009 (.113) \\
.10 & .080 (.097) & .093 (.106) & .094 (.107) & .079 (.094) & .093 (.103) & .094 (.104) \\
.20 & .158 (.095) & .191 (.104) & .193 (.108) & .158 (.093) & .191 (.102) & .194 (.106) \\
.30 & .236 (.095) & .293 (.094) & .303 (.102) & .236 (.094) & .294 (.093) & .303 (.101) \\
.40 & .297 (.121) & .376 (.092) & .390 (.094) & .297 (.120) & .377 (.091) & .391 (.093) \\
\addlinespace
\multicolumn{7}{l}{\hspace{-0.1in}{$n=1000$}} \\
.00 & -.013 (.101) & -.011 (.103) & -.011 (.103) & -.012 (.096) & -.010 (.098) & -.010 (.098) \\
.10 & .089 (.078) & .099 (.085) & .099 (.085) & .088 (.077) & .098 (.083) & .099 (.083) \\
.20 & .174 (.069) & .200 (.076) & .201 (.078) & .174 (.068) & .200 (.075) & .201 (.077) \\
.30 & .253 (.074) & .299 (.075) & .305 (.082) & .252 (.074) & .298 (.075) & .305 (.081) \\
.40 & .315 (.098) & .383 (.072) & .395 (.075) & .315 (.098) & .383 (.072) & .395 (.074) \\
\addlinespace
\multicolumn{7}{l}{\hspace{-0.1in}{$n=1500$}} \\
.00 & -.010 (.090) & -.008 (.092) & -.008 (.092) & -.010 (.087) & -.008 (.089) & -.008 (.089) \\
.10 & .096 (.067) & .105 (.073) & .105 (.073) & .096 (.066) & .105 (.071) & .105 (.071) \\
.20 & .186 (.056) & .208 (.064) & .209 (.065) & .186 (.055) & .208 (.063) & .209 (.065) \\
.30 & .263 (.061) & .303 (.064) & .308 (.069) & .262 (.061) & .303 (.063) & .308 (.068) \\
.40 & .325 (.087) & .386 (.064) & .397 (.067) & .325 (.087) & .386 (.064) & .397 (.067) \\
\addlinespace
\multicolumn{7}{l}{\hspace{-0.1in}{$n=2000$}} \\
.00 & -.010 (.089) & -.009 (.090) & -.009 (.090) & -.010 (.085) & -.009 (.086) & -.009 (.086) \\
.10 & .098 (.062) & .106 (.067) & .106 (.067) & .098 (.060) & .106 (.065) & .106 (.066) \\
.20 & .187 (.053) & .208 (.061) & .209 (.062) & .187 (.052) & .208 (.059) & .209 (.060) \\
.30 & .263 (.058) & .301 (.058) & .305 (.063) & .263 (.057) & .301 (.058) & .305 (.062) \\
.40 & .327 (.083) & .384 (.060) & .394 (.062) & .328 (.082) & .385 (.059) & .395 (.062) \\
\bottomrule
\end{tabular}
\end{table}

\subsection{Simulation results for distributional time series}\label{supp subsect::distribution}

\subsubsection*{Dist--Location design}

For the Dist--Location design, Tables~\ref{tab_mc_plwass_ratio} and~\ref{tab_mc_plwass_ls} report the finite-sample performance of the log-ratio and log-slope estimators, respectively.

\begin{table}[H]
\centering
\renewcommand{\arraystretch}{0.85}
\caption{Monte Carlo results for the Dist--Location design using the log-ratio estimators. Entries are means with RMSEs in parentheses, based on 1,000 replications.}
\label{tab_mc_plwass_ratio}
\fontsize{7.3pt}{8.4pt}\selectfont
\setstretch{1.0}
\setlength{\tabcolsep}{9.5pt}
\begin{tabular}{@{}crrrrrr@{}}
\toprule
& \multicolumn{3}{c}{Baseline} & \multicolumn{3}{c}{Tuning average} \\
\cmidrule(lr){2-4}\cmidrule(lr){5-7}
$\mathrm d$ & Raw & BC & FP & Raw & BC & FP \\
\midrule
\multicolumn{7}{l}{\hspace{-0.1in}{$n=250$}} \\
.00 & -.022 (.130) & -.016 (.138) & -.015 (.138) & -.021 (.126) & -.015 (.133) & -.014 (.134) \\
.10 & .068 (.125) & .089 (.140) & .090 (.143) & .067 (.122) & .088 (.137) & .090 (.140) \\
.20 & .143 (.121) & .185 (.135) & .190 (.142) & .142 (.119) & .184 (.134) & .190 (.140) \\
.30 & .216 (.124) & .286 (.125) & .297 (.133) & .215 (.124) & .285 (.124) & .297 (.132) \\
.40 & .281 (.143) & .376 (.113) & .393 (.114) & .279 (.144) & .376 (.112) & .392 (.114) \\
\addlinespace
\multicolumn{7}{l}{\hspace{-0.1in}{$n=500$}} \\
.00 & -.013 (.118) & -.010 (.122) & -.010 (.122) & -.013 (.110) & -.010 (.115) & -.010 (.115) \\
.10 & .082 (.097) & .097 (.107) & .097 (.108) & .082 (.094) & .096 (.104) & .097 (.105) \\
.20 & .162 (.093) & .196 (.104) & .199 (.108) & .162 (.092) & .196 (.103) & .199 (.106) \\
.30 & .240 (.092) & .299 (.094) & .309 (.103) & .240 (.091) & .299 (.093) & .309 (.102) \\
.40 & .300 (.118) & .380 (.090) & .394 (.093) & .301 (.117) & .381 (.090) & .395 (.092) \\
\addlinespace
\multicolumn{7}{l}{\hspace{-0.1in}{$n=1000$}} \\
.00 & -.014 (.103) & -.012 (.105) & -.012 (.105) & -.013 (.099) & -.011 (.101) & -.011 (.101) \\
.10 & .092 (.078) & .103 (.085) & .103 (.085) & .092 (.077) & .103 (.084) & .103 (.084) \\
.20 & .179 (.067) & .205 (.077) & .207 (.079) & .179 (.066) & .205 (.076) & .207 (.078) \\
.30 & .257 (.071) & .304 (.075) & .311 (.083) & .256 (.071) & .304 (.075) & .311 (.082) \\
.40 & .318 (.096) & .387 (.071) & .399 (.074) & .318 (.096) & .387 (.071) & .399 (.074) \\
\addlinespace
\multicolumn{7}{l}{\hspace{-0.1in}{$n=1500$}} \\
.00 & -.012 (.093) & -.010 (.095) & -.010 (.095) & -.011 (.089) & -.009 (.091) & -.009 (.091) \\
.10 & .100 (.067) & .109 (.073) & .109 (.073) & .100 (.066) & .109 (.072) & .109 (.072) \\
.20 & .190 (.055) & .213 (.065) & .214 (.067) & .190 (.054) & .213 (.064) & .214 (.066) \\
.30 & .267 (.059) & .308 (.064) & .314 (.070) & .266 (.058) & .308 (.064) & .314 (.070) \\
.40 & .328 (.084) & .390 (.063) & .401 (.066) & .328 (.084) & .390 (.063) & .401 (.066) \\
\addlinespace
\multicolumn{7}{l}{\hspace{-0.1in}{$n=2000$}} \\
.00 & -.011 (.094) & -.010 (.095) & -.010 (.095) & -.010 (.088) & -.009 (.089) & -.009 (.089) \\
.10 & .101 (.062) & .110 (.068) & .110 (.068) & .102 (.061) & .110 (.067) & .110 (.067) \\
.20 & .191 (.052) & .213 (.062) & .214 (.063) & .192 (.051) & .213 (.060) & .214 (.062) \\
.30 & .267 (.056) & .305 (.059) & .310 (.064) & .267 (.055) & .306 (.058) & .310 (.063) \\
.40 & .330 (.080) & .388 (.059) & .398 (.062) & .331 (.080) & .389 (.058) & .399 (.061) \\
\bottomrule
\end{tabular}
\end{table}

\begin{table}[H]
\centering
\renewcommand{\arraystretch}{0.85}
\caption{Monte Carlo results for the Dist--Location design using the log-slope estimators. Entries are means with RMSEs in parentheses, based on 1,000 replications.}
\label{tab_mc_plwass_ls}
\fontsize{7.3pt}{8.4pt}\selectfont
\setstretch{1.0}
\setlength{\tabcolsep}{9.5pt}
\begin{tabular}{@{}crrrrrr@{}}
\toprule
& \multicolumn{3}{c}{Baseline} & \multicolumn{3}{c}{Tuning average} \\
\cmidrule(lr){2-4}\cmidrule(lr){5-7}
$\mathrm d$ & Raw & BC & FP & Raw & BC & FP \\
\midrule
\multicolumn{7}{l}{\hspace{-0.1in}{$n=250$}} \\
.00 & -.022 (.132) & -.016 (.139) & -.015 (.140) & -.021 (.126) & -.014 (.134) & -.014 (.135) \\
.10 & .068 (.126) & .089 (.142) & .091 (.144) & .068 (.123) & .089 (.138) & .090 (.141) \\
.20 & .143 (.122) & .185 (.137) & .190 (.143) & .142 (.120) & .184 (.134) & .190 (.140) \\
.30 & .216 (.124) & .285 (.125) & .297 (.133) & .215 (.124) & .285 (.125) & .296 (.132) \\
.40 & .280 (.144) & .376 (.113) & .392 (.114) & .279 (.144) & .375 (.113) & .391 (.114) \\
\addlinespace
\multicolumn{7}{l}{\hspace{-0.1in}{$n=500$}} \\
.00 & -.013 (.119) & -.009 (.123) & -.009 (.123) & -.013 (.112) & -.009 (.116) & -.009 (.116) \\
.10 & .083 (.098) & .097 (.108) & .098 (.109) & .082 (.095) & .097 (.105) & .097 (.106) \\
.20 & .162 (.094) & .196 (.105) & .199 (.109) & .162 (.092) & .196 (.103) & .199 (.107) \\
.30 & .240 (.093) & .299 (.095) & .309 (.103) & .240 (.092) & .299 (.094) & .309 (.102) \\
.40 & .300 (.119) & .379 (.091) & .393 (.093) & .300 (.118) & .380 (.090) & .394 (.092) \\
\addlinespace
\multicolumn{7}{l}{\hspace{-0.1in}{$n=1000$}} \\
.00 & -.014 (.105) & -.012 (.107) & -.012 (.107) & -.012 (.100) & -.010 (.102) & -.010 (.102) \\
.10 & .092 (.080) & .103 (.087) & .103 (.087) & .092 (.078) & .103 (.085) & .103 (.085) \\
.20 & .179 (.068) & .205 (.077) & .207 (.080) & .178 (.067) & .205 (.076) & .207 (.079) \\
.30 & .257 (.072) & .304 (.076) & .311 (.083) & .256 (.072) & .303 (.075) & .310 (.082) \\
.40 & .318 (.096) & .386 (.072) & .398 (.074) & .318 (.096) & .386 (.071) & .398 (.074) \\
\addlinespace
\multicolumn{7}{l}{\hspace{-0.1in}{$n=1500$}} \\
.00 & -.011 (.095) & -.010 (.097) & -.010 (.097) & -.010 (.090) & -.008 (.091) & -.008 (.091) \\
.10 & .100 (.068) & .109 (.075) & .109 (.075) & .100 (.067) & .109 (.073) & .109 (.073) \\
.20 & .190 (.055) & .213 (.066) & .214 (.067) & .190 (.055) & .213 (.065) & .214 (.066) \\
.30 & .266 (.059) & .308 (.064) & .313 (.070) & .266 (.059) & .308 (.064) & .313 (.070) \\
.40 & .328 (.084) & .390 (.064) & .400 (.067) & .328 (.084) & .390 (.063) & .400 (.066) \\
\addlinespace
\multicolumn{7}{l}{\hspace{-0.1in}{$n=2000$}} \\
.00 & -.011 (.097) & -.010 (.098) & -.010 (.098) & -.010 (.089) & -.009 (.090) & -.009 (.090) \\
.10 & .101 (.063) & .109 (.069) & .110 (.069) & .102 (.062) & .110 (.068) & .110 (.068) \\
.20 & .191 (.053) & .213 (.062) & .213 (.064) & .192 (.052) & .213 (.061) & .214 (.062) \\
.30 & .267 (.056) & .305 (.059) & .309 (.064) & .267 (.055) & .305 (.058) & .310 (.063) \\
.40 & .330 (.081) & .388 (.059) & .398 (.062) & .330 (.080) & .388 (.059) & .398 (.061) \\
\bottomrule
\end{tabular}
\end{table}

\subsubsection*{Dist--Location\&Scale design}

For the Dist--Location\&Scale design, Tables~\ref{tab_mc_wass_ratio} and~\ref{tab_mc_wass_ls} report the finite-sample performance of the log-ratio and log-slope estimators, respectively.

\begin{table}[H]
\centering
\renewcommand{\arraystretch}{0.85}
\caption{Monte Carlo results for the Dist--Location\&Scale design using the log-ratio estimators. Entries are means with RMSEs in parentheses, based on 1,000 replications.}
\label{tab_mc_wass_ratio}
\fontsize{7.3pt}{8.4pt}\selectfont
\setstretch{1.05}
\setlength{\tabcolsep}{9.5pt}
\begin{tabular}{@{}crrrrrr@{}}
\toprule
& \multicolumn{3}{c}{Baseline} & \multicolumn{3}{c}{Tuning average} \\
\cmidrule(lr){2-4}\cmidrule(lr){5-7}
$\mathrm d$ & Raw & BC & FP & Raw & BC & FP \\
\midrule
\multicolumn{7}{l}{\hspace{-0.1in}{$n=250$}} \\
.00 & -.022 (.130) & -.016 (.138) & -.015 (.138) & -.021 (.125) & -.015 (.133) & -.014 (.134) \\
.10 & .068 (.125) & .088 (.140) & .090 (.143) & .067 (.122) & .088 (.137) & .090 (.140) \\
.20 & .143 (.121) & .184 (.135) & .190 (.142) & .142 (.119) & .184 (.133) & .189 (.139) \\
.30 & .216 (.124) & .285 (.125) & .297 (.133) & .215 (.124) & .284 (.124) & .296 (.132) \\
.40 & .280 (.143) & .376 (.113) & .392 (.114) & .279 (.144) & .376 (.112) & .392 (.114) \\
\addlinespace
\multicolumn{7}{l}{\hspace{-0.1in}{$n=500$}} \\
.00 & -.013 (.117) & -.010 (.122) & -.010 (.122) & -.013 (.110) & -.010 (.114) & -.010 (.115) \\
.10 & .082 (.097) & .097 (.107) & .097 (.108) & .082 (.094) & .096 (.104) & .097 (.105) \\
.20 & .162 (.094) & .196 (.104) & .199 (.108) & .162 (.092) & .196 (.103) & .199 (.106) \\
.30 & .240 (.092) & .299 (.094) & .309 (.103) & .240 (.091) & .299 (.093) & .309 (.102) \\
.40 & .300 (.119) & .380 (.090) & .394 (.093) & .300 (.118) & .381 (.090) & .395 (.092) \\
\addlinespace
\multicolumn{7}{l}{\hspace{-0.1in}{$n=1000$}} \\
.00 & -.014 (.103) & -.012 (.105) & -.012 (.105) & -.013 (.099) & -.011 (.101) & -.011 (.101) \\
.10 & .092 (.078) & .102 (.085) & .103 (.085) & .092 (.077) & .102 (.084) & .102 (.084) \\
.20 & .179 (.067) & .205 (.077) & .207 (.079) & .178 (.066) & .205 (.076) & .207 (.078) \\
.30 & .257 (.072) & .304 (.075) & .311 (.082) & .256 (.071) & .303 (.075) & .310 (.082) \\
.40 & .318 (.096) & .387 (.071) & .399 (.074) & .318 (.096) & .387 (.071) & .399 (.074) \\
\addlinespace
\multicolumn{7}{l}{\hspace{-0.1in}{$n=1500$}} \\
.00 & -.012 (.093) & -.010 (.095) & -.010 (.095) & -.011 (.089) & -.009 (.091) & -.009 (.091) \\
.10 & .100 (.067) & .109 (.073) & .109 (.073) & .100 (.066) & .109 (.072) & .109 (.072) \\
.20 & .190 (.055) & .213 (.065) & .214 (.066) & .190 (.054) & .213 (.064) & .214 (.066) \\
.30 & .266 (.059) & .308 (.064) & .313 (.070) & .266 (.058) & .308 (.064) & .313 (.070) \\
.40 & .328 (.084) & .390 (.063) & .400 (.066) & .328 (.084) & .390 (.063) & .400 (.066) \\
\addlinespace
\multicolumn{7}{l}{\hspace{-0.1in}{$n=2000$}} \\
.00 & -.011 (.093) & -.010 (.094) & -.010 (.094) & -.010 (.087) & -.010 (.088) & -.010 (.088) \\
.10 & .101 (.062) & .109 (.068) & .109 (.068) & .102 (.061) & .110 (.067) & .110 (.067) \\
.20 & .191 (.052) & .213 (.061) & .214 (.063) & .192 (.051) & .213 (.060) & .214 (.061) \\
.30 & .267 (.056) & .305 (.059) & .310 (.064) & .267 (.055) & .305 (.058) & .310 (.063) \\
.40 & .330 (.081) & .388 (.059) & .398 (.062) & .331 (.080) & .389 (.058) & .399 (.061) \\
\bottomrule
\end{tabular}
\end{table}

\begin{table}[H]
\centering
\renewcommand{\arraystretch}{0.85}
\caption{Monte Carlo results for the Dist--Location\&Scale design using the log-slope estimators. Entries are means with RMSEs in parentheses, based on 1,000 replications.}
\label{tab_mc_wass_ls}
\fontsize{7.3pt}{8.4pt}\selectfont
\setstretch{1.0}
\setlength{\tabcolsep}{9.5pt}
\begin{tabular}{@{}crrrrrr@{}}
\toprule
& \multicolumn{3}{c}{Baseline} & \multicolumn{3}{c}{Tuning average} \\
\cmidrule(lr){2-4}\cmidrule(lr){5-7}
$\mathrm d$ & Raw & BC & FP & Raw & BC & FP \\
\midrule
\multicolumn{7}{l}{\hspace{-0.1in}{$n=250$}} \\
.00 & -.022 (.132) & -.016 (.139) & -.015 (.140) & -.021 (.126) & -.014 (.134) & -.014 (.134) \\
.10 & .068 (.126) & .089 (.141) & .091 (.144) & .068 (.123) & .089 (.138) & .090 (.140) \\
.20 & .143 (.122) & .184 (.137) & .190 (.143) & .142 (.120) & .184 (.134) & .190 (.140) \\
.30 & .216 (.124) & .285 (.125) & .297 (.133) & .215 (.125) & .284 (.125) & .296 (.132) \\
.40 & .280 (.144) & .376 (.113) & .392 (.114) & .279 (.144) & .375 (.113) & .391 (.114) \\
\addlinespace
\multicolumn{7}{l}{\hspace{-0.1in}{$n=500$}} \\
.00 & -.013 (.119) & -.009 (.123) & -.009 (.123) & -.013 (.111) & -.009 (.116) & -.009 (.116) \\
.10 & .082 (.098) & .097 (.108) & .097 (.109) & .082 (.095) & .097 (.105) & .097 (.106) \\
.20 & .162 (.094) & .196 (.105) & .199 (.109) & .162 (.092) & .196 (.103) & .199 (.107) \\
.30 & .240 (.093) & .299 (.095) & .309 (.103) & .240 (.092) & .299 (.094) & .309 (.102) \\
.40 & .299 (.119) & .379 (.091) & .393 (.093) & .300 (.118) & .380 (.090) & .394 (.092) \\
\addlinespace
\multicolumn{7}{l}{\hspace{-0.1in}{$n=1000$}} \\
.00 & -.014 (.105) & -.012 (.107) & -.012 (.107) & -.012 (.099) & -.011 (.101) & -.011 (.102) \\
.10 & .092 (.080) & .103 (.087) & .103 (.087) & .092 (.078) & .102 (.085) & .102 (.085) \\
.20 & .178 (.068) & .205 (.077) & .207 (.080) & .178 (.067) & .205 (.076) & .206 (.079) \\
.30 & .256 (.072) & .303 (.076) & .310 (.083) & .256 (.072) & .303 (.075) & .310 (.082) \\
.40 & .318 (.096) & .386 (.072) & .398 (.074) & .318 (.096) & .386 (.071) & .398 (.074) \\
\addlinespace
\multicolumn{7}{l}{\hspace{-0.1in}{$n=1500$}} \\
.00 & -.011 (.095) & -.010 (.096) & -.010 (.096) & -.010 (.090) & -.008 (.091) & -.008 (.091) \\
.10 & .100 (.068) & .109 (.075) & .109 (.075) & .100 (.067) & .109 (.073) & .109 (.073) \\
.20 & .190 (.055) & .213 (.066) & .214 (.067) & .190 (.055) & .213 (.065) & .214 (.066) \\
.30 & .266 (.059) & .308 (.064) & .313 (.070) & .266 (.059) & .308 (.064) & .313 (.070) \\
.40 & .328 (.085) & .389 (.064) & .400 (.067) & .328 (.084) & .390 (.063) & .400 (.066) \\
\addlinespace
\multicolumn{7}{l}{\hspace{-0.1in}{$n=2000$}} \\
.00 & -.011 (.097) & -.010 (.098) & -.010 (.098) & -.010 (.089) & -.009 (.090) & -.009 (.090) \\
.10 & .101 (.063) & .109 (.069) & .109 (.069) & .102 (.062) & .110 (.067) & .110 (.067) \\
.20 & .191 (.053) & .212 (.062) & .213 (.063) & .191 (.052) & .213 (.061) & .214 (.062) \\
.30 & .266 (.056) & .305 (.059) & .309 (.064) & .267 (.055) & .305 (.058) & .310 (.063) \\
.40 & .330 (.081) & .388 (.059) & .398 (.062) & .330 (.080) & .388 (.059) & .398 (.061) \\
\bottomrule
\end{tabular}
\end{table}

Tables~\ref{tab_mc_plwass_ratio}, \ref{tab_mc_plwass_ls}, \ref{tab_mc_wass_ratio}, and~\ref{tab_mc_wass_ls} show patterns similar to those in Tables~\ref{tab_mc_psd_ratio} and~\ref{tab_mc_psd_ls}.

\subsection{Finite-sample behavior of the bias-corrected estimators}\label{supp-bias-behavior}
The corrected estimators exhibit a transition from underestimation to overestimation at some memory levels. This is anticipated by the diagnostics of Section~\ref{sec_bias} of the main paper: the raw estimators carry a negative common-centering effect, the population slope itself carries a positive finite-bandwidth Bartlett curvature, and reducing the former can leave the latter visible. Both components are finite-sample in origin and vanish as $n$ grows, at the rates $n^{-(1-\kappa)(1-2\mathrm d)}$ and $n^{-2\mathrm d\kappa}$ with $m\sim c_mn^\kappa$, so the sign of the residual bias reflects which component is currently dominant rather than any asymptotic distortion. The two rates agree at $\mathrm d^\star=(1-\kappa)/2$, equal to $1/3$ for the cube-root bandwidth.

For example, in the baseline log-slope results in Table~\ref{tab_mc_wass_ls}, BC and FP first become positively biased at $n=1000$ for $\mathrm d=.1$ and $.2$. At $n=2000$, their biases at these two memory levels are $.009$ and $.012$ for BC and $.009$ and $.013$ for FP, and at $\mathrm d=.3$ they are $.005$ for BC and $.009$ for FP. The positive bias is therefore largest at $\mathrm d=.2$, below $\mathrm d^\star$, where the negative component decays faster. That it persists at $n=2000$ is expected rather than anomalous: with $\kappa=1/3$ and $\mathrm d=.2$, the curvature term is of order $n^{-2/15}$, which falls only from about $.51$ to $.39$ as $n$ increases from $250$ to $2000$. Removing the dominant negative bias therefore leaves a small positive remainder over the sample sizes considered, and both components vanish under the rate conditions of Proposition~\ref{prop_rates}. The magnitudes remain small relative to the corresponding reduction in bias reported in Table~\ref{tab_mc_aggregate}, and the RMSEs decrease with $n$ in nearly every cell. As noted in the main paper, however, these are deterministic comparisons and do not imply a universal crossing or monotonicity in $n$.
\newpage 
\section{Additional real data figure}\label{supp sect::additional real data}

\begin{figure}[!htb]
\centering
\includegraphics[width=0.8\linewidth]{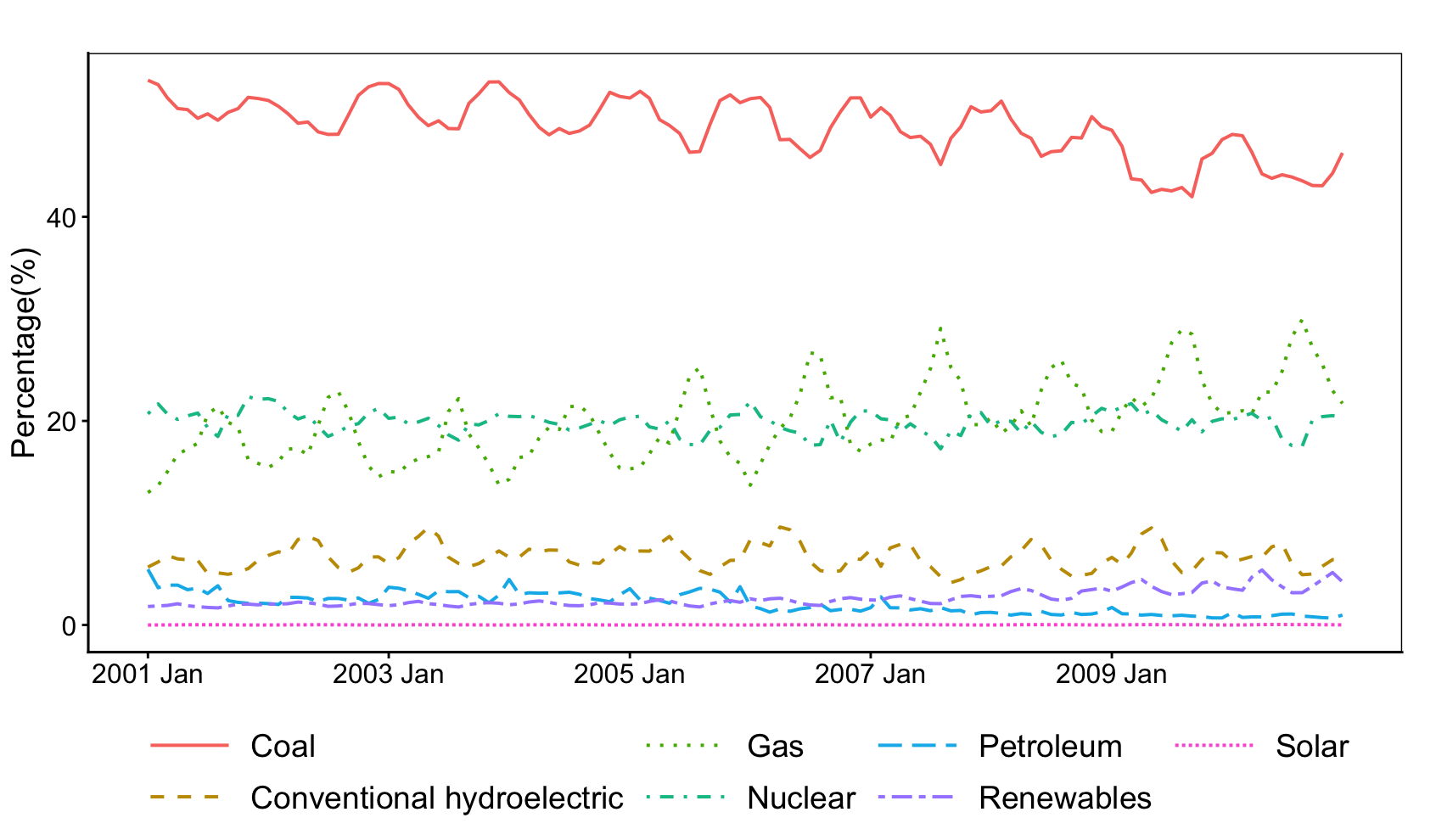} 
\caption{\small Monthly U.S. electricity-generation compositions from January 2001 to December 2010. Each line represents the percentage of one of the seven fuel categories: coal, petroleum, gas, nuclear, conventional hydroelectric, renewables, and~solar.}\label{fig::US energy}
\end{figure}

\newpage
\bibliography{memory}
\end{document}